\documentclass[acmsmall,screen,nonacm,natbib=false]{acmart}
\usepackage[utf8]{inputenc}
\usepackage[toc,page]{appendix}
\usepackage{verbatim}
\usepackage{ebproof} 
\usepackage{stmaryrd}
\usepackage{mathtools}
\usepackage{mathrsfs} 
\usepackage{accents} 
\usepackage{xfrac}
\usepackage{xcolor}
\usepackage{tikz-cd} 
\usepackage{thmtools,thm-restate} 
\usepackage{tikzit} 

\tikzstyle{texthere}=[inline text, font={\footnotesize}]
\tikzstyle{tiny box}=[rectangle, inline text, fill=white, draw, minimum height=5mm, yshift=-0.5mm, minimum width=5mm, font={\small}]
\tikzstyle{small box}=[rectangle, inline text, fill=white, draw, minimum height=7.5mm, yshift=-0.5mm, minimum width=7.5mm, font={\small}]
\tikzstyle{medium box}=[rectangle, inline text, fill=white, draw, minimum height=10mm, yshift=-0.5mm, minimum width=7.5mm, font={\small}]
\tikzstyle{semilarge box}=[rectangle, inline text, fill=white, draw, minimum height=12.5mm, yshift=-0.5mm, minimum width=7.5mm, font={\small}]
\tikzstyle{narrow semilarge box}=[rectangle, inline text, fill=white, draw, minimum height=12.5mm, yshift=-0.5mm, minimum width=5mm, font={\small}]
\tikzstyle{almost large box}=[rectangle, inline text, fill=white, draw, minimum height=15mm, yshift=-0.5mm, minimum width=7.5mm, font={\small}]
\tikzstyle{big box}=[rectangle, inline text, fill=white, draw, minimum height=17.5mm, yshift=-0.5mm, minimum width=7.5mm, font={\small}]
\tikzstyle{large box}=[rectangle, inline text, fill=white, draw, minimum height=20mm, yshift=-0.5mm, minimum width=7.5mm, font={\small}]
\tikzstyle{very large box}=[rectangle, inline text, fill=white, draw, minimum height=22.5mm, yshift=-0.5mm, minimum width=7.5mm, font={\small}]
\tikzstyle{upground}=[circuit ee IEC, thick, ground, scale=2.1]
\tikzstyle{downground}=[circuit ee IEC, thick, ground, rotate=180, scale=2.1]
\tikzstyle{point}=[regular polygon, regular polygon sides=3, draw, scale=0.67, inner sep=-0.5pt, minimum width=9mm, fill=white, regular polygon rotate=90]
\tikzstyle{copoint}=[regular polygon, regular polygon sides=3, draw, scale=0.67, inner sep=-0.5pt, minimum width=9mm, regular polygon rotate=-90, fill=white]
\tikzstyle{wide copoint}=[fill=white, draw, shape=isosceles triangle, shape border rotate=180, isosceles triangle stretches=true, inner sep=0pt, minimum width=1.5cm, minimum height=6.12mm]
\tikzstyle{wide point}=[fill=white, draw, shape=isosceles triangle, shape border rotate=0, isosceles triangle stretches=true, inner sep=0pt, minimum width=1.5cm, minimum height=6.12mm, yshift=-0.0mm]
\tikzstyle{white dot}=[fill=white, draw=black, shape=circle]
\tikzstyle{gray dot}=[fill={rgb,255: red,148; green,148; blue,148}, draw=black, shape=circle]
\tikzstyle{dot}=[inner sep=0mm, minimum width=1mm, minimum height=1mm, draw, shape=circle, fill=black]
\tikzstyle{empty dot}=[inner sep=0mm, minimum width=3mm, minimum height=3mm, draw, shape=circle, fill=none]

\tikzstyle{dashed line}=[-, dashed, dash pattern=on 1mm off 0.5mm]
\tikzstyle{dotted line}=[-, style=dotted, tikzit draw=brown]
\tikzstyle{cyan dashed line}=[-, style=dashed, tikzit draw=cyan]
\tikzstyle{green fill line}=[-, fill={green!90!black}, tikzit draw=green]
\tikzstyle{blue fill}=[-, fill=blue, tikzit fill=blue, tikzit draw={rgb,255: red,102; green,117; blue,255}]
\tikzstyle{red}=[-, draw=red, tikzit draw=red]
\tikzstyle{blue}=[-, draw=blue, tikzit draw=blue]
\tikzstyle{thick black}=[-, draw=black, tikzit draw=black, line width=1.8pt]
\tikzstyle{dotted red}=[-, draw=red, style=dotted, tikzit draw=red]
\tikzstyle{dotted blue}=[-, draw=blue, tikzit draw=blue, style=dotted]
\tikzstyle{dashed thick blue}=[-, draw={rgb,255: red,28; green,176; blue,255}, tikzit draw={rgb,255: red,83; green,19; blue,156}, line width=1pt, style=dashed]
\tikzstyle{dashed thick red}=[-, draw=red, tikzit draw={rgb,255: red,255; green,100; blue,10}, line width=1pt, style=dashed]
\tikzstyle{green}=[-, draw=green, tikzit draw=green]
\tikzstyle{dotted green}=[-, draw=green, tikzit draw=green, style=dotted]
\tikzstyle{arrow}=[->]
\tikzstyle{arrow green dashed}=[draw=green, ->, tikzit draw=green, style=dashed]
\tikzstyle{arrow dashed red}=[draw=red, ->, style=dashed, tikzit draw=red]
\tikzstyle{dashed green}=[-, tikzit draw=green, draw=green, style=dashed]

\input{diagrams.tikzdefs}
\usepackage{todonotes}

\def\longueversion{}
\newcommand{\version}[2]{\ifdefined\longueversion #1 \else #2 \fi}

\RequirePackage[
datamodel=acmdatamodel,
style=acmauthoryear
]{biblatex}
\DeclarePairedDelimiter\pair{\langle}{\rangle}
\DeclarePairedDelimiter\bra{\langle}{\rvert}
\DeclarePairedDelimiter\ket{\lvert}{\rangle}
\DeclarePairedDelimiterX\braket[2]{\langle}{\rangle}{#1 \delimsize\vert #2}

\newcommand{\db}[1]{\{\mskip-5mu\{ #1 \}\mskip-5mu\}}

\newcommand{\meas}{\mathtt{meas}}
\newcommand{\qcase}{\mathtt{qcase}}
\newcommand{\letrec}{\mathtt{letrec}}
\newcommand{\unit}{\text{\texttt{()}}}

\newcommand{\sem}[1]{\llbracket #1 \rrbracket}

\newcommand{\Caus}{\mathbf{Caus}}
\newcommand{\CPM}{\mathbf{CPM}}
\newcommand{\FHilb}{\mathbf{FHilb}}
\newcommand{\Mem}{\text{\textnormal{Mem}}}
\newcommand{\Supp}{\text{\textnormal{Supp}}}

\newcommand{\q}{\text{\textnormal{q}}}
\newcommand{\parr}{\text{\rotatebox[origin=c]{180}{$\&$}}}

\newcommand{\lolli}{\multimap}
\newcommand{\one}{\text{\textbf{1}}}

\newcommand{\alift}[1]{\lceil #1 \rceil}

\newcommand{\chan}{\mathcal{C}}

\newcommand{\smap}{\mathcal{S}}
\newcommand{\device}{\mathbb{D}}

\newcommand{\Meas}{\texttt{Meas}}
\newcommand{\Measd}[1]{\Meas [#1]}

\newcommand{\cat}{\mathsf{C}}

\begin{document}

\theoremstyle{acmdefinition}
\newtheorem{remark}[theorem]{Remark}

\title[Higher-Order Programs with Indefinite Causal Orders]{Higher-Order Programs with Indefinite Causal Orders: a Linear Approach to Coherent Control of Quantum Processes
}
\author{Kathleen Barsse}
\affiliation{
  \institution{Université de Lorraine, CNRS, Inria, LORIA}
  \country{France}
}
\author{Romain Péchoux}
\affiliation{
  \institution{Université de Lorraine, CNRS, Inria, LORIA}
  \country{France}
}
\author{Simon Perdrix}
\affiliation{
  \institution{Université de Lorraine, CNRS, Inria, LORIA}
  \country{France}
}
\authorsaddresses{}

\begin{abstract}
Processes with indefinite causal orders (ICOs), such as the quantum switch, are higher-order quantum processes that superpose the order in which quantum operations are performed. Such coherent control yields computational advantages but is not faithfully captured by existing quantum programming languages: either they are restricted to the unitary case, and thus cannot combine ICOs with measurement, or they treat coherent control nonlinearly. In both cases, they do not realize the full computational power of ICOs.
We introduce a higher-order quantum functional language that supports general quantum computation, not merely the permutation of channels, and whose linear type system allows quantum control to be well-defined beyond the unitary case, on arbitrary quantum channels.
We equip this language with a small-step operational semantics that synchronizes measurement outcomes across superposed branches, using device references and a memory function. 
We also give a denotational semantics by means of completely positive maps. With linearity as the only constraint, some well-typed terms would denote unphysical maps. We therefore impose a typing discipline that goes beyond linearity, and interpret programs in the causal category $\Caus[\CPM]$, under which every well-typed program is physically meaningful, a property that can be checked statically and efficiently.
We prove soundness, and study the language's expressive power: it can express every quantum channel at first order, and at second order a large subclass of the so-called quantum circuits with quantum control (\textbf{QC-QC}s), containing the quantum switch. Last but not least, we show that this language is well-designed enough to be extended to the nonlinear setting with recursion.
\end{abstract}

\maketitle

\section{Introduction}

\subsection{Context and Motivations}

\emph{Coherent control}, or \emph{quantum control}, is a concept that has attracted significant attention in the development of quantum programming languages.
It refers to the ability to express not only superpositions of data, but also superpositions of programs, which are typically controlled by a quantum state.
A fundamental example of coherent control is the \emph{quantum switch}~\cite{chiribella}.
The \emph{quantum switch} is the higher order quantum process whose input is pair of quantum operations, say $\mathcal{C}$ and $\mathcal{D}$, and whose output is the operation that consists in applying $\mathcal{D}\circ \mathcal{C}$ and $\mathcal{C}\circ \mathcal{D}$ in a superposition, based on the state of a control qubit. The two programs in superposition are the two possible orderings of the input operations: $\mathcal{C}$ before $\mathcal{D}$ or $\mathcal{D}$ before $\mathcal{C}$. This is called an \emph{indefinite causal order}.

Indefinite causal orders~\cite{oreshkov} are a subclass of coherent control that arises from higher-order quantum processes.
They can be described as processes in which the ordering between a set of processes is in a superposition.
Indefinite causal orders are of particular interest in quantum foundations~\cite{consistent_circuits,grenoble_process,projective},
but have also been shown to grant an advantage in computational complexity in a variety of tasks~\cite{araujo,araujo_guerin_baumeler,renner2021,renner2022,abbott_boolean,kristjansson_exponential,photonic}.
A significant number of these examples use the quantum switch.
Moreover, access to processes with indefinite causal orders can also provide advantages in some communication problems~\cite{ebler,grenoble_coherent_control} or information processing tasks~\cite{grenoble_process}.

Yet, existing quantum programming languages do not fully capture the expressive power of processes with indefinite causal orders. There are two major reasons for this.

Firstly, as we will argue in the paper, \emph{linearity} is a key feature to implementing such processes. Processes with indefinite causal orders are linear by nature: the different orderings in superposition all have the same underlying set of processes. However, programming languages that support coherent control often do so in a nonlinear fashion, which implements a closer analog to the classical \texttt{if/then/else} statement.
Nonetheless, some processes with indefinite causal orders can be simulated. For instance, in the unitary case, the quantum switch can be simulated by a quantum circuit in which one of the inputs is called twice~\cite{chiribella}.
However, due to using more resources, this implementation loses the potential computational advantages that the quantum switch would usually enable.

Secondly, integrating both indefinite causal orders and measurements into a quantum programming language is a non-trivial problem. In particular, the calculi in which quantum control \emph{is} handled linearly do not feature measurement~\cite{realizability,lambda_s1,sabry}. 
Fundamentally, the problem concerns synchronizing measurement outcomes occurring in different branches of a superposition. That is, a process with indefinite causal orders consists of a superposition of a set of processes being executed in different orders. If one of those processes contains a measurement, which is a probabilistic operation, the observed outcome must be the same in every branch.
In particular, the simulation of the quantum switch where one input is used twice does not work beyond the unitary case for exactly this reason.
The ability to implement measurements is necessary to express coherent control in its full generality. Remarkably, \citet{kristjansson_exponential} demonstrated an exponential separation in quantum query complexity between models with indefinite causal orders and standard quantum circuits, and this advantage specifically requires non-unitary gates.

In this paper, we develop a higher-order programming language that faithfully implements a large class of processes with indefinite causal orders, accounting for the presence of measurements. Crucially, one can check statically in polynomial time (see Remark~\ref{rmk:tit}) whether a program represents a physically meaningful process.

\subsection{Overview}

We give a brief overview of the design of the language and present its main features.
The syntax of the language (Section~\ref{section:syntax}) consists of the terms of the lambda calculus together with primitives implementing specific quantum features (including local unitaries, measurements and state initialisation). 
Quantum control is implemented {with} a dedicated $\qcase$ primitive,  written as
\[
 \qcase\ P\ \{0\to M\ |\ 1\to N\}.
\]
This program reads as follows. The term $P$ represents a control qubit: if $P$ evaluates to $\ket{0}$, then $M$ is executed; if $P$ evaluates to $\ket{1}$, then $N$ is executed. If $P$ evaluates to a superposition, then a coherent superposition of $M$ and $N$ is executed. The example of the quantum switch is given in  Fig.~\ref{fig:switch}. A linear typing discipline guarantees that  any variables appearing freely in $M$ or $N$ must appear exactly once in $M$ and exactly once in $N$. It allows an \emph{indefinite-causal-order} interpretation of higher order computations: roughly speaking, the same quantum channels are called in both $M$ and $N$ and only the order in which they are applied differs.

We endow this language with both operational (Section~\ref{section:operational}) and denotational semantics (Section~\ref{section:denotational}).

\begin{figure}[t]
$\begin{array}{l}\mathtt{switch} \ := \\ \quad
\lambda \pair{f,g}.\lambda q. \qcase\ q\ \{\\
\quad\quad\quad 0\rightarrow \mathtt{comp} \ \pair{f,g} \ |\ \\ \quad\quad\quad1\rightarrow \mathtt{comp} \ \pair{g,f} \\\quad\}\\\text{with } \mathtt{comp}\ :=\ \lambda \pair{f, g}.\lambda t. g(ft)
\end{array}\qquad\qquad$\raisebox{-0.2cm}{\scalebox{0.85}{\begin{tikzpicture}
	\begin{pgfonlayer}{nodelayer}
		\node [style=none] (2) at (-1, 1) {};
		\node [style=none] (3) at (1, 1) {};
		\node [style=none] (4) at (-1, 0) {};
		\node [style=none] (5) at (1, 0) {};
		\node [style=none] (6) at (-1, 2.75) {};
		\node [style=none] (7) at (-2.25, 3.375) {};
		\node [style=none] (8) at (1, 2.75) {};
		\node [style=none] (9) at (2.25, 3.375) {};
		\node [style=none] (10) at (1, -1.75) {};
		\node [style=none] (11) at (2.25, -2.375) {};
		\node [style=none] (12) at (-1, -1.75) {};
		\node [style=none] (13) at (-2.25, -2.375) {};
		\node [style=none] (14) at (-2.25, 1.875) {};
		\node [style=none] (15) at (-1, 1.875) {};
		\node [style=none] (16) at (-1, -0.875) {};
		\node [style=none] (18) at (-1.5, -0.375) {};
		\node [style=none] (19) at (1.5, -0.375) {};
		\node [style=none] (20) at (1, -0.875) {};
		\node [style=none] (21) at (0, 0.5) {};
		\node [style=none] (22) at (1.5, 1.375) {};
		\node [style=none] (23) at (1, 1.875) {};
		\node [style=none] (24) at (-1.5, 1.375) {};
		\node [style=none] (25) at (2.25, 1.875) {};
		\node [style=none] (26) at (-1.75, 0.5) {};
		\node [style=none] (27) at (1.75, 0.5) {};
		\node [style=none] (28) at (-3.75, 1.875) {};
		\node [style=none] (29) at (-3.75, -0.875) {};
		\node [style=none] (30) at (-2.25, -0.875) {};
		\node [style=none] (31) at (2.25, -0.875) {};
		\node [style=none] (32) at (3.75, -0.875) {};
		\node [style=none] (33) at (3.75, 1.875) {};
		\node [style=texthere] (34) at (-3.95, 2.125) {$q$};
		\node [style=texthere] (35) at (3.95, 2.125) {$q$};
		\node [style=texthere] (36) at (-3.95, -0.625) {$t$};
		\node [style=texthere] (37) at (3.95, -0.625) {$t$};
		\node [style=tiny box] (38) at (0, 1.875) {$f$};
		\node [style=tiny box] (39) at (0, -0.875) {$g$};
	\end{pgfonlayer}
	\begin{pgfonlayer}{edgelayer}
		\draw (4.center) to (5.center);
		\draw (2.center) to (3.center);
		\draw (3.center) to (8.center);
		\draw (9.center) to (11.center);
		\draw (10.center) to (5.center);
		\draw (12.center) to (4.center);
		\draw (2.center) to (6.center);
		\draw (7.center) to (13.center);
		\draw [style=dashed line, color=red, bend left=45] (23.center) to (22.center);
		\draw [style=dashed line, color=red, in=0, out=-90] (22.center) to (21.center);
		\draw [densely dash dot, color=blue, bend left=45, looseness=1.25] (24.center) to (15.center);
		\draw [densely dash dot, color=blue, in=180, out=-90, looseness=0.75] (24.center) to (21.center);
		\draw [style=dashed line, color=red, in=90, out=180] (21.center) to (18.center);
		\draw [style=dashed line, color=red, bend right=45] (18.center) to (16.center);
		\draw [densely dash dot, color=blue, bend right=45] (20.center) to (19.center);
		\draw [densely dash dot, color=blue, in=0, out=90] (19.center) to (21.center);
		\draw (25.center) to (33.center);
		\draw (31.center) to (32.center);
		\draw (29.center) to (30.center);
		\draw (28.center) to (14.center);
		\draw (15.center) to (23.center);
		\draw (16.center) to (20.center);
		\draw (7.center) to (9.center);
		\draw (6.center) to (8.center);
		\draw (12.center) to (10.center);
		\draw (13.center) to (11.center);
		\draw [style=dashed line, color=red, in=180, out=90, looseness=1.25] (26.center) to (15.center);
		\draw [style=dashed line, color=red, in=0, out=-90, looseness=0.75] (26.center) to (30.center);
		\draw [densely dash dot, color=blue, in=90, out=0, looseness=1.25] (23.center) to (27.center);
		\draw [densely dash dot, color=blue, in=-180, out=-90, looseness=0.75] (27.center) to (31.center);
		\draw [densely dash dot, color=blue] (30.center) to (16.center);
		\draw [style=dashed line, color=red] (20.center) to (31.center);
	\end{pgfonlayer}
\end{tikzpicture}}}
\caption{(Left) The quantum switch program. (Right) A diagrammatic representation of the quantum switch. }
\label{fig:switch}
\end{figure}
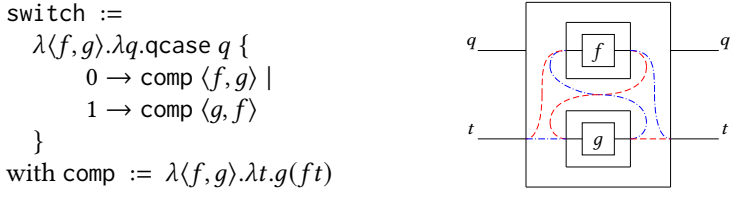

\subsubsection{Operational Semantics}
In order to provide some intuition of the step-by-step execution of programs, let us consider the example of the quantum switch in more detail.
The quantum switch is the higher order function whose input is a pair of quantum operations, given by variables $f$ and $g$,
and whose output is a quantum operation acting on a pair of qubits $q$ (``control'') and $t$ (``target''). If the control qubit is in state $\ket{0}$, $f$ then $g$ are applied to the target qubit. If it is in state $\ket{1}$, $g$ then $f$ are applied to the target qubit. Otherwise, the two cases are executed in a coherent superposition depending on the state of the control qubit. 
This process can be summarized by the diagram in Figure~\ref{fig:switch}, where the red dashed lines and the blue dash-dotted lines represent the two possible evolutions of the target qubit depending on the control qubit state. When the control qubit $q$ is in state $\ket{0}$, the target qubit $t$ follows the red path. When the control qubit is in state $\ket{1}$, the target qubit follows the blue path. Otherwise,  the two alternative paths are in a superposition.

Consider its application $\mathtt{switch}\ \pair{C,D}\ (H\ket{0})\ \ket k$ to a pair of operations $\pair{C,D}$, with a control qubit in a superposition ($H$ being the Hadamard gate), and with a target qubit in state $\ket k$. Diagrammatically, the execution actually corresponds to two separate evolutions, represented as two alternative paths in a coherent superposition. This will be written using sums of terms. Therefore $\mathtt{switch}\ \pair{C,D} \ (H\ket{0})\ \ket{k}$ will reduce, after a certain number of steps, to
\begin{align}
\label{eq:paths}
\frac{1}{\sqrt{2}} \cdot {\color{red}  \pair{ \ket{0}, D(C\ket{k})}} + \frac{1}{\sqrt{2}}\cdot {\color{blue} \pair{\ket{1},C(D\ket{k})} }.
\end{align}
where the first term of the superposition represents the red path with the control qubit in state $\ket{0}$, and the second term represents the blue path with the control qubit in state $\ket{1}$.

One of the main challenges we address is the combination of indefinite causal orders with measurements.
Suppose that the input $C$ contains a measurement (e.g., $C=\meas \ P\ \{0\rightarrow M\ |\ 1\rightarrow N\}$, for some $P$, $N$, and $M$).
The two instances of $C$ in~(\ref{eq:paths}) correspond to two distinct paths in which that same process appears, and \emph{not} to two copies of the process; diagrammatically, the two copies of $C$ correspond to the same box. Therefore,
we expect the same measurement outcome for each copy of $C$. To achieve this, we need to keep track of which identical subterms correspond to the same physical process.
This is done by assigning a unique \emph{device reference} to each measurement, which can be thought of as pointing to a specific physical device.
To synchronize the measurement outcomes occurring in different branches of the superposition,
we record measurement outcomes by means of a \emph{memory function}, which is a partial function from device references to $\{0,1\}$.

The operational semantics (Section~\ref{section:operational}) makes use of all these features: superpositions, device references, and memory functions, through the concept of \emph{configuration}.
It is important to note that superpositions of terms and device references are not actually accessible to the programmer, but rather serve as semantic tools to define the evolution of terms.

\subsubsection{Denotational Semantics.}
At first order, the most general evolutions from qubits to qubits are \emph{quantum channels},
formally described as completely positive trace-preserving maps. They correspond to arbitrary
compositions of unitary evolutions and measurements. At second order, the most general
transformations from quantum channels to quantum channels are the \emph{quantum
supermaps}~\cite{chiribella}. Iterating this construction yields quantum transformations of
arbitrarily high order. Processes with indefinite causal orders, such as the quantum switch, can be
described at the level of quantum supermaps.

We introduce a denotational semantics based on completely positive maps, which can be seen as 
morphisms of the category $\CPM$. However, some higher-order morphisms of $\CPM$ are not physically
meaningful, in particular those that fail to satisfy the appropriate trace-preservation conditions. Such
unphysical morphisms can be described by means of a non careful use of the $\qcase$ construct (see Remark~\ref{rm:qcase}). To rule them
out, we impose on $\qcase$ a typing discipline that goes beyond linearity, so that only functions
whose output consists of qubits can be coherently controlled. We then crucially rely on the $\Caus$
construction~\cite{kissinger2019categorical} to prove that every well-typed term is physically
meaningful. 
This denotational semantics provides a fine-grained treatment of the causal structure of quantum
maps and supermaps. For instance, a value of type $(\q\lolli\q)\otimes(\q\lolli\q)$ need not be a
pair of single-qubit channels, but may be an arbitrary no-signalling two-qubit channel~\cite{chiribella}.
Consequently, the quantum switch can be applied not only to pairs of single-qubit channels, but
also to no-signalling two-qubit channel.

\subsection{Summary of Contributions}

In this paper, we define a higher-order functional programming language that supports indefinite causal orders and measurements.
We give an operational semantics, which consists of a small-step transition system over configurations,  i.e., superpositions of terms, and a denotational semantics in the category $\Caus[\CPM]$. The language is shown to satisfy standard safety properties.

We prove the following main results:

\subsubsection{Soundness Theorem}
We prove that the denotational semantics is sound with respect to the operational semantics (Theorem~\ref{thm:soundness}): the reduction relation preserves the interpretation of terms.

\subsubsection{Physicality of Programs}
By construction, well-typed programs represent physically meaningful processes.
This is because the denotational semantics can be defined not only in $\CPM$, but also in the causal category $\Caus[\CPM]$. Therefore, explicit connections can be made between types and collections of higher-order quantum maps.

It is especially interesting to relate this result to the operational semantics.
In quantum programming languages based on superpositions of $\lambda$-terms, one of the main difficulties is checking that terms represent physically meaningful processes, typically unitary maps.
This is generally achieved through orthogonality checks either at the price of sacrificing intensional completeness~\cite{dave2025combining} or at the price of being undecidable~\cite{vinet}.
Here, since superpositions of terms are not a feature that is accessible to the programmer, we will be considering sequences of reductions that start from a term without sums.
Hence, the physicality of terms can be checked statically and efficiently: it suffices to check that the initial term is well-typed.

\subsubsection{Expressivity}

We study the expressivity of the language at the level of first-order and second-order functions. For first-order  functions, we show that all quantum channels can be realized in the language (Proposition~\ref{prop:channels}).
To study second-order functions, it is helpful to consider the expressivity of the language in comparison to the hierarchy of supermaps defined in~\cite{grenoble_process}, which we represent in Figure~\ref{fig:supermaps}. The innermost class, \emph{Quantum Circuits with Fixed Causal Order}, or \textbf{QC-FO}s, represents processes that do not have any control structure. The next class, \emph{Quantum Circuits with Classical Control}, or \textbf{QC-CC}s, corresponds to processes fitting the ``quantum data, classical control'' paradigm~\cite{QPL}. The next class is \emph{Quantum Circuits with Quantum Control}, or \textbf{QC-QC}s. This class contains all known examples of physically realizable processes with indefinite causal orders, including the quantum switch. Finally, \textbf{QC-QC}s do not encompass all existing supermaps—for instance the OCB process defined in~\cite{oreshkov} lies outside of the class of \textbf{QC-QC}s.

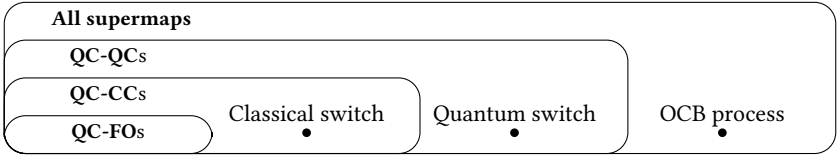
\begin{figure}[!ht]
\[
\scalebox{1}{\begin{tikzpicture}
	\draw [rounded corners=8pt] (0,0) rectangle (22,4);     
	\draw [rounded corners=8pt] (0,0) rectangle (16.5,3);   
	\draw [rounded corners=8pt] (0,0) rectangle (11,2);     
	\draw [rounded corners=8pt] (0,0) rectangle (5.5,1);    
 
	\node [style=texthere] at (2.75,0.5) {\textbf{QC-FO}s};
	\node [style=texthere] at (2.75,1.5) {\textbf{QC-CC}s};
	\node [style=texthere] at (2.75,2.5) {\textbf{QC-QC}s};
	\node [style=texthere] at (3.1,3.5)  {\textbf{All supermaps}};
 
	\node [style=dot] (cs) at (8,0.55)   {};
	\node [style=dot] (qs) at (13.5,0.55) {};
	\node [style=dot] (ocb) at (19,0.55) {};
 
	\node [style=texthere] at (8,0.95)    {\small Classical switch};
	\node [style=texthere] at (13.5,0.95) {\small Quantum switch};
	\node [style=texthere] at (19,0.95)   {\small OCB process};
\end{tikzpicture}}
\]
\caption{The nested classes of supermaps defined in~\cite{grenoble_process}}
\label{fig:supermaps}
\Description[A diagram of nested classes of supermaps.]{A diagram of nested classes of supermaps. From innermost to outermost, the classes are \texttt{QC-FO}s, \texttt{QC-CC}s (which contain the classical switch), \texttt{QC-QC}s (which contain the quantum switch) and all supermaps (which contain the OCB process).}
\end{figure}

We prove that the language implements a subclass of \textbf{QC-QC}s that contains all \textbf{QC-CC}s, which we call \textbf{QC-QC}s with memory (Proposition~\ref{prop:qcqc}).
This class contains the quantum switch and its direct generalizations~\cite{unitary2012,unitary2014}, but also some processes with dynamical control of causal orders—that is, processes where the orderings in a superposition are not fixed from the beginning and constructed on the fly.

\subsubsection{An Extension for Non-Linearity and Recursion}
Lastly, in Section~\ref{section:extension} we extend the language to include nonlinear features and recursive functions.
The ability to duplicate inputs is necessary to implement standard quantum algorithms whose inputs are oracles that are queried several times (for instance,~\citet{grover}'s algorithm).
Moreover, this enables us to define recursive functions, which will be considered duplicable.
For example, this allows in to implement~\citet{repeat_until_success}'s ``Repeat-Until-Success'' circuits (Example~\ref{ex:rus}).
We show that the syntax and operational semantics can be extended to include both linear and nonlinear features in a straightforward way.

\subsection{Related Work}

There are numerous examples of quantum programming languages with coherent control, for instance~\cite{altenkirch,qunity,silq,ying_verification,sabry,dave2025combining,quantum_control_machine,zhang,HLMR25}.
In most cases, coherent control is expressed as a quantum version of the classical \texttt{if/then/else} statement.
Usually this primitive is called quantum case (``qcase''), quantum alternation or conditional, and corresponds to the process in which two branching programs are executed in a superposition depending on the state of a control qubit.
Generally, there is no linearity constraint between the \texttt{then} and \texttt{else} branches (although we point out that quantum control is treated linearly in~\cite{realizability,lambda_s1} in the absence of measurement). The lack of linearity constraint enables one to write operations such as mapping a unitary gate to its controlled version. However, processes with indefinite causal orders cannot be faithfully implemented.

In fact, linearity is not just a means to implement processes with indefinite causal orders. At a more fundamental level, it allows quantum control itself to be well-defined. Indeed, without any linearity constraints, quantum control is defined on unitary maps by $(U,V)\mapsto \ket{0}\bra{0} \otimes U +\ket{1}\bra{1}\otimes V$. However, the extension of this semantics to completely positive maps is ill-defined~\cite{badescu}.
For this reason, programming languages featuring both the nonlinear quantum case and measurements use semantic models such as Kraus decompositions~\cite{badescu,ying_control_flow,ying_alternation} or vacuum-extended channels~\cite{bpp}.

Operationally, coherent control can be implemented in models with superpositions of $\lambda$-terms~\cite{lineal,realizability,lambda_s1,meas2017,meas2019,dave2025combining,vinet}.
In such calculi, one can write a unitary map directly from its matrix representation, which removes the need for built-in unitary gates.
In fact, unitaries themselves are written as an instance of quantum control.
The drawback is that verifying whether a term represents a physically meaningful process is generally difficult.
To this end, it is necessary to define orthogonality constraints~\cite{altenkirch,sabry,realizability}, where in particular, it was shown in~\cite{vinet} that deciding the orthogonality of terms is a $\Pi^0_2$-complete problem. 
Finally, \citet{vantonder}'s quantum $\lambda$-calculus, like ours, restricts sums to the topmost
level of terms; however, its additional constraints on the shape of terms permit only classical
control.

There exist several frameworks for the study of higher-order quantum maps~\cite{jencova,bisio,projective}.
The framework of causal categories~\cite{kissinger2019categorical} can be seen as a generalization to higher order of the trace-preserving condition for quantum channels.
Another approach  is presented in~\cite{tsukada}, which constructs a model  guaranteeing the property of ``total probability $\leq 1$'' for higher-order programs.

\section{Linear Language with Indefinite Causal Orders}
\label{section:syntax}
In this section, we introduce the syntax and type system of a linear higher-order programming language for indefinite causal orders.
We illustrate its expressiveness through several examples,
including the quantum switch (Example~\ref{ex:qs}).

\subsection{Syntax}
\begin{figure}[t]
\hrulefill
\begin{align*}
 \begin{array}{lrl}
 \
  \text{(Terms)}&\quad \mathscr{T} \ni M,N,P \  ::=  &  x \ |\ M N \ |\ \lambda x.M \\
  && \!\!\!\! |\ \ \pair{M,N} \ |\ \mathtt{let}\ \pair{x,y} = M \  \mathtt{in}\ N \ |\ \unit\ |\ M;N \\
  && \!\!\!\! |\ \ U\ | \ \ket{0}\ |\ \ket{1}
  \\
  && \!\!\!\! |\ \ \meas \ P\ \{0\rightarrow M\ |\ 1\rightarrow N\} \\
  && \!\!\!\! |\ \ \qcase \ P\ \{0\rightarrow M\ |\ 1\rightarrow N\}
 \end{array}
\end{align*}
\hrulefill
\caption{Syntax of the linear language for indefinite causal orders}
\label{fig:syntax}
\end{figure}
Let $\mathcal{X}$ be a countably infinite set of \emph{linear variables}, written $x,y,z...$
The set of terms $\mathscr{T}$ is defined by the grammar of Figure~\ref{fig:syntax}.
The terms $MN,\lambda x. M$ denote linear application and abstraction. The terms $\pair{M,N}$ and  $\mathtt{let}\ \pair{x, y} = M \  \mathtt{in}\ N$ are the pair constructor and destructor, respectively. The term $\unit$ is the unit.  The sequence term $M;N$ can be thought of as ``$\mathtt{let} \ \unit=M \ \mathtt{in} \ N$'', meaning that the computation of $M$ produces side effects and returns no meaningful value.

The term $\meas\ P \ \{0\rightarrow M \ |\  1\rightarrow N\}$ is used for expressing branching depending on the result of the measurement of a qubit $q$. If the result of the measurement is $0$, then $M$ is executed, otherwise, $N$ is executed.
 The term $\qcase \ P \ \{0\rightarrow M\ |\ 1\rightarrow N\}$ represents the superposition of terms controlled by $P$. The terms $\ket{0}$ and $\ket{1}$ represent qubit states.  $U$ is any single  qubit unitary transformation, i.e., $U\in \mathbb C^{2\times 2}$ s.t. $U^\dagger U=I$, $I$ being the identity. Non-local unitary transformations can be reconstructed, for instance control-U
 can be defined as the following term
 \[\lambda x. \qcase \ x \ \{0\rightarrow I\ |\ 1\rightarrow U\}.\]

By convention, application associates to the left, whereas pairs associate to the right.
We identify terms up to $\alpha$-equivalence in a standard way.
A term with no free variables is said to be \emph{closed}.
We define the following notations inductively, for all $n\geq 2$:
\[\begin{array}{rcl}
\mathtt{let} \ x=M \ \mathtt{in}\ N &:= &(\lambda x . N)M \\
\mathtt{let} \ \pair{x_1,...,x_{n+1}} =M \ \mathtt{in}\ N &:= &\mathtt{let} \ \pair{x_1,y} =M \ \mathtt{in}\ \mathtt{let} \ \pair{x_2,...,x_{n+1}} =y \ \mathtt{in}\ N \\
\lambda \pair{x_1,...,x_n}. M & := &	\lambda y. (\mathtt{let} \ \pair{x_1,...,x_n} = y \ \mathtt{in} \ M) \\
\ket{k_1,...,k_n} & := & \pair{\ket{k_1},...,\ket{k_n}} \quad \text{where }k_1,...,k_n\in\{0,1\}
\end{array}
\]

Lastly, we define the set of \emph{controllable terms} $\mathscr{C}\subseteq \mathscr{T}$ as the set of terms that do not contain measurement.
Controllable terms are the only ones that can appear in the branches of a $\qcase$ statement.
This restriction is formalized by the type system in the following section.

\subsection{Linear Type System}
\label{section:types}

The type system is based on the multiplicative fragment of intuitionistic linear logic~\cite{girard}.
\emph{Types} are given by the following syntax:
\[
 \begin{array}{lll}
  \text{(Types)}\quad A,B \ & ::= \ & \one\ |\ \q \ |\
  A \otimes B \ |\ A \lolli B
 \end{array}
\]
$\one$ is the unit type, and the ground type $\q$ represents a single qubit.
$A\otimes B$ is the linear product type. The type $A \lolli B$ corresponds to linear functions.
By convention, $\lolli$ and $\otimes$ associate to the right.
We write $\q^n$ for $\q \otimes ... \otimes \q$, where $\q$ appears $n\geq 1$ times.

A \emph{typing context} $\Delta$ is a list of pairs of variables and their respective types $x_1:A_1,...,x_n:A_n$, where $x_1,...,x_n$ are pairwise distinct.
Let $\Delta$ and $\Delta'$ be two typing contexts. If $\Delta$ and $\Delta'$ have no common linear variables, then we write $\Delta, \Delta'$ for their union.
\emph{Typing judgments} are expressions of the following form:
\[
 \Delta \vdash M :A
\]
where $M$ is a term, $A$ a type, and $\Delta$ a typing context.

The typing rules are detailed in Figure~\ref{fig:typing_linear}. If a typing judgment can be derived using the typing rules, then it is said to be \emph{valid}. A term $M$ is said to be \emph{well-typed} if there exist $\Delta$ and $A$ such that $\Delta\vdash M:A$ is valid.
Notice that in the meas and qcase rules, the typing context $\Delta'$ is the same for both branching statements.
In particular, if a typing judgment $\Delta:M:A$ is valid, each variable of $\Delta$ is guaranteed to appear exactly once in the term, where occurences in both branches of $\meas$ or $\qcase$ are counted as the same.

\begin{figure}[!t]
\hrulefill
\vspace*{-0.5cm}
\begin{center}
\[
\scalebox{0.9}{
 \begin{prooftree}
  \hypo{\phantom{k\{0,1\}}}
 \infer1[ax]{x:A \vdash x:A}
\end{prooftree}}
 \qquad
 \scalebox{0.9}{
 \begin{prooftree}
 \hypo{\phantom{k\{0,1\}}}
  \infer1[unit]{\vdash U:\q \lolli \q}
 \end{prooftree}}
 \qquad
 \scalebox{0.9}{
 \begin{prooftree}
  \hypo{k\in\{0,1\}}
  \infer1[qbit]{\vdash \ket{k}:\q}
 \end{prooftree}}
 \qquad
 \scalebox{0.9}{
  \begin{prooftree}
  \hypo{\Delta,x:A,y:B,\Delta'\vdash M:C}
  \infer1[X]{\Delta,y:B,x:A,\Delta'\vdash M:C}
 \end{prooftree}}
 \]
 
 \[
 \scalebox{0.9}{
\begin{prooftree}
 \hypo{\Delta,x:A\vdash M:B}
 \infer1[$\lolli$I]{\Delta\vdash \lambda x. M:A\lolli B}
 \end{prooftree}}
 \qquad
 \scalebox{0.9}{
 \begin{prooftree}
  \hypo{\Delta\vdash M:A\lolli B}
  \hypo{\Delta'\vdash N:A}
  \infer2[$\lolli$E]{\Delta,\Delta'\vdash M N:B}
 \end{prooftree}}
 \qquad
 \scalebox{0.9}{
 \begin{prooftree}
  \hypo{\Delta \vdash M:A}
  \hypo{\Delta'\vdash N:B}
  \infer2[$\otimes$I]{\Delta,\Delta' \vdash \pair{M , N}:A\otimes B}
 \end{prooftree}}
 \]

\[
\scalebox{0.9}{
  \begin{prooftree}
\hypo{\Delta \vdash M:A\otimes B}
\hypo{\Delta',x:A,y:B\vdash N:C}
\infer2[$\otimes$E]{\Delta,\Delta' \vdash \mathtt{let}\ \pair{x,y}=M\ \mathtt{in}\ N:C}
\end{prooftree}}
\qquad
\scalebox{0.9}{
 \begin{prooftree}
  \hypo{\phantom{\Delta'}}
  \infer1[$\one$I]{\vdash \unit:\one}
 \end{prooftree}}
 \qquad
 \scalebox{0.9}{
 \begin{prooftree}
  \hypo{\Delta\vdash M: \one}
  \hypo{\Delta' \vdash N: A}
  \infer2[seq]{\Delta,\Delta'\vdash M;N:A}
 \end{prooftree}}
\]

\[
\scalebox{0.9}{
 \begin{prooftree}
 \hypo{\Delta \vdash P:\q}
  \hypo{\Delta' \vdash M:A}
  \hypo{\Delta' \vdash N: A}
  \infer3[meas]{\Delta,\Delta'\vdash \meas\ P\ \{0\rightarrow M\ |\ 1\rightarrow N\}: A}
 \end{prooftree}}
\]

\[
\scalebox{0.9}{
 \begin{prooftree}
  \hypo{\Delta \vdash P:\q}
  \hypo{\Delta' \vdash M:A\lolli \q^n}
  \hypo{\Delta' \vdash N: A\lolli \q^n}
  \hypo{M,N\in\mathscr{C}}
  \infer4[qcase]{\Delta,\Delta'\vdash \qcase\ P\ \{0\rightarrow M\ |\ 1\rightarrow N\}: A\lolli \q^{n+1}}
 \end{prooftree}}
\]
\end{center}
\hrulefill
\caption{Typing rules. Recall that $\mathscr{C}$ is the set of controllable terms, i.e.,~those without measurements.}
\label{fig:typing_linear}
\end{figure}

The measurement statement is intended to discard the control qubit. To see this, notice that the type of $\meas\ P\ \{0\to M\ |\ 1\to N\}$ matches that of $M$ and $N$. On the contrary, the qcase statement preserves the control qubit. This is why the type of $\qcase\ P\ \{0\rightarrow M\ |\ 1\rightarrow N\}$ is $ A\lolli \q^{n+1}$, where $M$ and $N$ are controllable terms of type $A\lolli \q^n$ for some $n\geq 1$.
This type may seem surprising at first glance. However, the naive typing approach, which would consist in assigning $\qcase\ P\ \{0\rightarrow M\ |\ 1\rightarrow N\}$ the type $\q \otimes B$ when $M$ and $N$ have type $B$, raises some {causality} issues: as we will see in Section~\ref{section:operational} (Remark~\ref{rm:qcase}), it would allow us to type terms that are not physical.

In the qcase rule, the requirement that the branching statements $M$ and $N$ be controllable is essential to defining the operational and denotational semantics. We emphasize that this restriction does not prevent us from expressing the coherent control of terms containing measurements. Via an abstraction, we can write terms such as
\[\mathtt{let}\ x=(\lambda y. \meas\ y\ \{0\to \ket{0} | 1\to \ket{1} \}) \ \mathtt{in}\ \qcase\ q\ \{0\to x\ |\ 1\to x\} \]
which is well typed and can be interpreted as the coherent control of a measurement channel.\footnote{To define the reduction of this term, we will see in Section~\ref{section:operational} that the operational semantics relies on a set of \emph{execution terms} whose typing is slightly more permissive.}

\begin{remark}[Linearity as a means of defining coherent control]
We briefly discuss how a linear type system enables us to define coherent control.
It is helpful to distinguish two variants of the qcase primitive appearing in various quantum programming language: the linear qcase, as discussed in this paper, and the nonlinear qcase.
The nonlinear qcase consists of a statement $\qcase \ x\  \{0\to M \ |\ 1 \to N\}$ where the only constraint on $M$ and $N$ is that they be of the same type. When $M$ and $N$ compute the unitary maps $U$ and $V$, respectively, its semantics is given by $(U,V)\mapsto \ket{0}\bra{0}_{x}\otimes U+\ket{1}\bra{1}_{x}\otimes V$. This semantics however cannot be extended to completely positive maps~\cite{badescu}.

The linear qcase consists in the same statement where $M$ and $N$ are additionally required to have exactly the same set of linear variables, and cannot contain measurements  (see Figure~\ref{fig:typing_linear}). Adding the linearity constraint is enough for the process to be well defined on completely positive maps. For instance, this is the case of the quantum switch, where each branching statement contains one copy of each input channel. More generally, processes with indefinite causal orders correspond to this linear variant of qcase. For these reasons, our language will use the linear version of the quantum case, formally enforced by the linear type discipline of Figure~\ref{fig:typing_linear}.

\end{remark}

\begin{remark}[Tractability of type inference]\label{rmk:tit}
As the rules of Figure~\ref{fig:typing_linear} implement a simple-type discipline over a linear lambda-calculus, its type inference is known to be tractable (i.e., computable in polynomial time), as a direct consequence of~\cite{M04}.
 \end{remark}

\subsection{Examples}
\label{section:syntax_examples}
We present various examples to illustrate the expressiveness of the language.
\version{The full typing derivations are given in Appendix~\ref{app:examples}.}{}

\begin{example}[CNOT]\label{ex:cnot}
Non local unitary transformations can be constructed in the language, for instance the control-NOT gate is defined by the following term:
\[
 \mathtt{CNOT} \ := \ \lambda c.\qcase \ c \ \{0\rightarrow I \ |\ 1\rightarrow X\}.
\]
One can derive the typing judgment $\vdash  \mathtt{CNOT}: \q\lolli \q\lolli \q\otimes \q$.
Although only single-qubit unitary gates are built into the language's syntax, this construction demonstrates that this restriction does not actually affect the language's expressiveness. Indeed, CNOT and single-qubit unitary gates form a univeral set of gates~\cite{barenco}.

\end{example}

\begin{example}[Discard]\label{ex:discard}
The discard map is implemented by the following program:
\[
 \mathtt{discard} \ := \ \lambda c. \meas\ c\ \{0\rightarrow \unit\ |\ 1\rightarrow \unit\}
\]
and can be typed as $ \vdash \mathtt{discard}: \q \lolli \one$.
Note that due to linearity, $\mathtt{discard}$ could not have been written as $\lambda c. \unit$.

\end{example}

\begin{example}[Composition]\label{ex:comp}
The term $\mathtt{comp}=\lambda \pair{f, g}.\lambda t. g(ft)$, provided in Figure~\ref{fig:switch}, that simulates the composition can be given the type $(A\lolli B)\otimes (B\lolli C)\lolli A\lolli C$, for any types $A$, $B$, $C$.
\end{example}

\begin{example}[Quantum switch]\label{ex:qs}
\label{ex:switch_syntax}
Using the composition function, the quantum switch
$ \mathtt{switch} \ := \lambda \pair {x,y}. \lambda q. \qcase\ q\ \{0\rightarrow \mathtt{comp} \ \pair{x,y} \ |\ 1\rightarrow \mathtt{comp} \ \pair{y,x}\}$ (also given in Figure~\ref{fig:switch}) can be typed as:
\[
 \vdash \mathtt{switch}: (\q^n\lolli \q^n)\otimes (\q^n\lolli \q^n )\lolli \q\lolli \q^n \lolli \q^{n+1}
\]

\end{example}

\section{Operational Semantics}
\label{section:operational}

\subsection{Design Choices}
\label{section:design}
In this section, we describe the operational semantics for evaluating the terms of the language. The presence of controllable terms makes the definition of the operational semantics challenging, as, for instance, given $ \Meas \ := \ \lambda c. \meas \ c \ \{0 \to \ket{0} \ |\ 1 \to\ket{1}\}$, a  quantum switch on two copies of $ \Meas$: \[ \lambda \pair {x,y}. \lambda q. \qcase\ q\ \{0\rightarrow \mathtt{comp} \ \pair{x,y} \ |\ 1\rightarrow \mathtt{comp} \ \pair{y,x}\}\ \pair{ \Meas, \Meas}\] is well typed but
$\lambda q. \qcase\ q\ \{0\to  \mathtt{comp}\ \pair { \Meas, \Meas}~|~1\to  \mathtt{comp}\ \pair { \Meas, \Meas}\}$ is not. As a consequence, the first term cannot reduce to the second in one or several steps. There is a fundamental reason for this: $\qcase\ q\ \{0\to \mathtt{comp}\ \pair {x,y}~|~1\to \mathtt{comp}\ \pair {y,x}\}$ represents an indefinite causal order: the instrument $x$ is performed first in one branch, last in the other. 
 By contrast, the term $\qcase\ q\ \{0\to  \mathtt{comp}\ \pair { \Meas, \Meas}~|~1\to  \mathtt{comp}\ \pair { \Meas, \Meas}\}$ does not witness which instrument is performed first, and is thus rejected by the typing rules. To preserve distinguishability while reducing the terms,
we equip each measurement device with a unique identifier via  \emph{device references}. 
A second challenge in the definition of the operational semantics is the necessity of reducing both branches of a $\qcase$ in superposition, so as to allow interference. To do so, we allow a term to reduce to a sum of terms (at the topmost level) to account for the superposition phenomenon.
Thus, the operational semantics intuitively proceeds as follows: 
\begin{align*}
&\lambda \pair{x,y}. \lambda q. \qcase\ q\ \{0\to \mathtt{comp}\ \pair {x,y}~|~1\to \mathtt{comp}\ \pair {y,x}\}\ \pair{ \Measd{d}, \Measd{e}} \ H\ket 1\ \ket k\\
&\quad \twoheadrightarrow^* \, \, \frac 1{\sqrt 2} \cdot \pair{ \ket 0, \mathtt{comp}\ \pair {\Measd{d},\Measd{e}}\ \ket k}- \frac 1{\sqrt 2} \cdot  \pair{\ket 1, \mathtt{comp}\ \pair {\Measd{e},\Measd{d}}\ \ket k}
\end{align*}
where $\Measd{i} :=  \lambda c. \meas \ i\triangleright c \ \{0 \to \ket{0} \ |\ 1 \to\ket{1}\}$ and $d,e$ are device identifiers. 
Additionally, the measurement outcomes associated with each device will be recorded in a \emph{memory function}. 
This way, the two occurrences of $\Measd{d}$ (one in each branch of the superposition) share the
device $d$ and hence produce the same measurement outcome, while $\Measd{e}$ carries a different
device $e$ and thus produces an independent outcome.

\subsection{A Language for Program Executions}

In order to define the reduction of terms, we define a language for program executions, which is essentially the original set of terms $\mathscr{T}$ to which we add device references and linear combinations of terms.
These new features are not assumed to be accessible to the programmer, but are rather a tool to define the semantics of the language.

\subsubsection{Adding Device References}

The first step is to add device references.
Let $\device$ be a countably infinite set of \emph{device references}, which we write as $d,e,f$...
The set of \emph{execution terms} $\mathscr{E}$ is defined by the grammar in Figure~\ref{fig:syntax_ex}.
\begin{figure}[!t]
\hrulefill
\begin{align*}
 \begin{array}{lrl}
 \
  \text{(Execution\ terms)}&\quad \mathscr{E} \ni M,N,P \  ::=  &  x \ |\ M N \ |\ \lambda x.M \\
  && \!\!\!\! |\ \ \pair{M,N} \ |\ \mathtt{let}\ \pair{x,y} = M \  \mathtt{in}\ N \ |\ \unit\ |\ M;N \\
  && \!\!\!\! |\ \ U\ | \ \ket{0}\ |\ \ket{1}
  \\
  && \!\!\!\! |\ \ \meas \ d\triangleright P\ \{0\rightarrow M\ |\ 1\rightarrow N\} \\
  && \!\!\!\! |\ \ \qcase \ P\ \{0\rightarrow M\ |\ 1\rightarrow N\}
 \end{array}
\end{align*}
\hrulefill
\caption{Syntax of execution terms}
\label{fig:syntax_ex}
\end{figure}
The only change compared to the set of terms $\mathscr{T}$ (Figure~\ref{fig:syntax}) is the addition of a device reference $d$ to the primitive for measurement. The set of device references appearing in an execution term $M$ will be written as $\text{Dev}(M)$. In order to type execution terms, the following typing rules are defined on terms of $\mathscr{E}$:
\[
\scalebox{0.9}{
\begin{prooftree}
 \hypo{\Delta \vdash P:\q}
  \hypo{\Delta' \vdash M:A}
  \hypo{\Delta' \vdash N: A}
  \infer3[meas]{\Delta,\Delta'\vdash \meas\ d\triangleright P\ \{0\rightarrow M\ |\ 1\rightarrow N\}: A}
 \end{prooftree}}
 \quad
 \scalebox{0.9}{
  \begin{prooftree}
  \hypo{\Delta \vdash P:\q}
  \hypo{\Delta' \vdash M:A\lolli \q^n}
  \hypo{\Delta' \vdash N: A\lolli \q^n}
  \infer3[qcase]{\Delta,\Delta'\vdash \qcase\ P\ \{0\rightarrow M\ |\ 1\rightarrow N\}: A\lolli \q^{n+1}}
 \end{prooftree}}
\]
Namely, the typing rule for measurement with a device reference is similar to the rule (meas) of Figure~\ref{fig:typing_linear}, and in the rule for qcase, the hypothesis requiring that $M$ and $N$ be controllable is dropped. The rules for all other primitives are defined as being the same as those in Figure~\ref{fig:typing_linear}.

As usual, for $M,N\in\mathscr{E}$ and $x\in\mathcal{X}$, $M\{N/x\}$ is the execution term obtained by replacing each occurrence of $x$ in $M$ by $N$.

 To define the operational semantics, a term $M\in\mathscr{T}$ is lifted to an execution term by
tagging its measurements with device references.
A \emph{lifting} of $M$ is a term $\alift{M}\in\mathscr{E}$ obtained from $M$ by assigning a
distinct device reference to each measurement (and leaving $M$ otherwise unchanged).
Equivalently, $\alift{M}$ is any execution term whose measurements carry pairwise distinct
references and which coincides with $M$ once these references are erased.
Liftings always exist and differ only in the choice of references.

Notice that if $M$ is well-typed, then so is $\alift{M}$. However, the converse is false.
This is because we have relaxed the assumption that branches of a $\qcase$ statement are controllable. For example, the execution term
\[
 \qcase\ x \ \{0\rightarrow \lambda t. \Measd{d} \ (Ht) \ |\ 1\rightarrow \lambda t. \Measd{d}  \ (Xt)\},
\]
with $\Measd{d}$ as in Section~\ref{section:design}, is now well typed.

\begin{remark}
In practice, we can assume that the compiler automatically assigns unique device references as needed before executing the program.
\end{remark}

To record the outcome of measurements occurring throughout the computation, we define \emph{memory functions}.  A memory function is a map $\sigma\colon \device\to\{0,1\}_\bot$, where $\{0,1\}_\bot:=\{0,1,\bot\}$ carries the flat order $\preceq$ ($\bot\preceq 0$, $\bot\preceq 1$, with $0$ and $1$ incomparable) and such that the support $\Supp(\sigma):=\{\,d\in\device\mid\sigma(d)\neq\bot\,\}$ is finite. We write $\mathcal{M}$ for the set of memory functions, ordered pointwise by~$\preceq$.
\begin{definition}[Consistency]
Two memory functions $\sigma,\nu$ are \emph{consistent}, written $\sigma{||}\nu$, if they admit a common upper bound in $\mathcal{M}$ --- equivalently, if they agree on $\Supp(\sigma)\cap\Supp(\nu)$. In that case their join $\sigma\sqcup\nu:=\sup\{\sigma,\nu\}$ exists (the union of $\sigma$ and $\nu$ seen as partial maps).
\end{definition}
We lift $\sqcup$ to sets: for pairwise consistent $\Theta,\Xi\subseteq\mathcal{M}$ (i.e.\ $\sigma||\nu$ for all $\sigma\in\Theta$, $\nu\in\Xi$), let $\Theta\sqcup\Xi:=\{\,\sigma\sqcup\nu\mid\sigma\in\Theta,\ \nu\in\Xi\,\}$, and abbreviate $\{\sigma\}\sqcup\Xi$ as $\sigma\sqcup\Xi$.
We write $[\bot]$ for the least memory function, mapping every device reference to $\bot$. For $d\in\device$ and $k\in\{0,1\}$, we write $[d\mapsto k]$ for the function sending $d$ to $k$ and every other device reference to $\bot$.

\subsubsection{Adding Term Superpositions}

In order to define the reduction rules, we must define linear combinations of terms.
We define \emph{configurations}, which allow term superpositions and are based on~\cite{realizability}'s notions of \emph{term distributions}.
The set of configurations, written as $\vec{\mathscr{E}}$, is defined by the following grammar:
\[
 \begin{array}{lrl}
 \text{(Configurations)} &\quad \vec{\mathscr{E}} \ni \vec{M},\vec{N} \ ::= & \vec{0}\ |\ M^{\sigma}\ |\ \vec{M}+\vec{N} \ |\ \alpha \cdot \vec{M}
 \end{array}
\]
where $M\in\mathscr{E}$, $\alpha\in\mathbb{C}$, and $\sigma\in\mathcal{M}$. The purpose of the memory function $\sigma$ is to record all past measurement outcomes.
In the particular case where $\sigma=[\bot]$, we sometimes omit the superscript by writing $M$ instead of  $M^{[\bot]}$ ; hence confusing the configuration with the execution term, by abuse of notation. The configuration or $1 \cdot M^{[\bot]}$ is called an \emph{initial configuration}, i.e., a configuration with amplitude $1$ and least memory configuration. For simplicity, we also write $M$  instead of $1 \cdot M^{[\bot]}$.

\begin{remark}
We emphasize that linear combinations only appear at the topmost level of configurations.
For example, instead of $(\lambda x. x) (\alpha\cdot \ket{0} +\beta\cdot \ket{1})$ we directly write $\alpha\cdot (\lambda x. x) \ket{0} +\beta\cdot (\lambda x. x) \ket{1})$.
\end{remark}

In Figure~\ref{fig:congruence}, we define a congruence relation $\equiv$ over $\vec{\mathscr{E}}$. The rules consist of all the axioms of a vector space except $0\cdot\vec{M}=\vec{0}$. In technical terms, this corresponds to a weak vector space, which is also used in~\citet{realizability}'s calculus. The reason for excluding this rule is that even if a term cancels out, we must still keep the information of its memory function.

\begin{figure}[!t]
\hrulefill
\begin{center}

\vspace{-4mm}
 \[
  \vec{M}+(\vec{N}+\vec{P}) \equiv (\vec{M}+\vec{N})+\vec{P}
  \qquad
  \vec{M}+\vec{0}\equiv \vec{M}
  \qquad \vec{M} + \vec{N} \equiv \vec{N} +\vec{M}
 \]
 \[
  \alpha\cdot (\vec{M}+\vec{N})\equiv\alpha \cdot\vec{M}+\alpha\cdot\vec{N}
  \quad\
  (\alpha+\beta)\cdot\vec{M}\equiv \alpha \cdot\vec{M}+\beta\cdot\vec{M}
  \quad\
  (\alpha \beta)\cdot \vec{M} \equiv \alpha \cdot (\beta\cdot \vec{M})
  \quad\
  1\cdot \vec{M} \equiv \vec{M}
 \]

\end{center}
\hrulefill
\caption{Congruence relation $\equiv$}
\label{fig:congruence}
\end{figure}

We will now consider configurations up to $\equiv$. Then by associativity, we can use the standard notation $\sum_{i=1}^n \alpha_i\cdot M_i^{\sigma_i}$ for the sum $\alpha_1 \cdot M_1^{\sigma_1} + ... + \alpha_n \cdot M_n^{\sigma_n}$. This allows us to define the canonical form of a configuration:
\begin{definition}[Canonical form]
  The canonical form of a configuration $\vec{M}$ is a configuration $\sum_{i=1}^n \alpha_i \cdot M_i^{\sigma_i}\equiv \vec{M}$ such that the $M_i^{\sigma_i}$ are pairwise distinct and the $\alpha_i$ are (possibly 0) complex scalars. The canonical form is unique up to reordering of the terms of the sum.
 \end{definition}

The \emph{support} of a configuration $\vec{M}=\sum_{i=1}^n \alpha_i\cdot M_i^{\sigma_i}$ is defined as the set $\Supp(\vec{M}):=\{M_i^{\sigma_i}\}_{1\leq i\leq n}$.
A configuration $\vec{M}=\sum_{i=1}^n\alpha_i \cdot M_i^{\sigma_i}$ is said to be \emph{closed} if every $M_i$ is closed.

The following definition generalizes the concept of a term being well typed to configurations. In addition to requiring that each individual summand be well typed, we give some consistency conditions on the memory functions:
\begin{definition}[Well-formedness]
Let $\vec{M}$ be a configuration whose canonical form is $\sum_{i=1}^n \alpha_i \cdot M_i ^{\sigma_i}$. $\vec{M}$ is said to be \emph{well formed} if the $\sigma_i$ are pairwise consistent; and for all $i$, $M_i$ is well typed and $\Supp(\sigma_i)\cap \text{Dev}(M_i)=\emptyset$.
\end{definition}

\begin{definition}[Memory of a configuration]
For all well-formed configuration $\vec{M}=\sum_{i=1}^n\alpha_i \cdot M_i^{\sigma_i}$, we define the memory of $\vec{M}$ as
\[
\Mem(\vec{M}): \sup_{1\leq i\leq n} \sigma_i
\]
where the supremum is taken with respect to $\preceq$. 
\end{definition}
$\vec{M}$ being well-formed ensures that the supremum is well defined.
In particular, if $\vec{M}$ and $\vec{N}$ are well-formed configurations satisfying $\Mem(\vec{M})\ ||\ \Mem(\vec{N})$, then $\vec{M}+\vec{N}$ is also well formed.
Given a memory function $\sigma$ and a well-formed configuration $\vec{M}=\sum_{i=1}^n\alpha_i \cdot M_i^{\nu_i}$ such that $\sigma \ ||\ \Mem(\vec{M})$, we define $\vec{M}^{\{\sigma\}}:=\sum_{i=1}^n\alpha_i \cdot M_i^{\sigma\sqcup \nu_i}$.

\subsection{Reduction Rules}

The small-step semantics is a relation ${\twoheadrightarrow}$ over well-formed configurations
$\vec{M}=\sum_{i=1}^n\alpha_i \cdot M_i^{\nu_i}\in \vec{\mathscr{E}}$. It operates by reducing the
individual execution terms $M_i\in\mathscr{E}$ as  initial configurations $1 \cdot M_i^{[\bot]}\equiv M_i^{[\bot]}$.
Each summand $M_i$ carries its own memory function $\nu_i$, whose
support records the measurements performed so far in that summand; these supports may therefore
differ from one summand of $\vec{M}$ to another: the semantics should appropriately update memory functions while preserving consistency of the memory functions.

Before giving the reduction rules, we introduce some preliminary notions and notations.
 The set of values $\mathscr{V}$ consists of execution terms defined by the following grammar:
 \[
 \begin{array}{lll}
  \text{(Values)} \quad \mathscr{V} \ni V,W \ ::= \!\!\!\!& \ \ \
  \lambda x. M \ |\ \unit  \ |\ \pair{V,W} \ |\ U \ |\ \ket{0}\ |\ \ket{1}
 \end{array}
\]

(Evaluation) \emph{contexts} are defined by the following grammar:
\begin{align*}
\text{(Contexts)} \quad  E\ ::=\ & [] \ | \ EV \ |\ ME\ |\ \pair{E,V}\ |\ \pair{M,E}\ |\ \mathtt{let}\ \pair{x,y}=E\ \mathtt{in}\ M\ |\ E;M \\
 &\!\!\!\! |\ \ \meas \ d\triangleright E\ \{0\to M\ |\ 1\to N\} \ |\ \qcase \ E\ \{0\to  M\ |\ 1\to N\}
\end{align*}
where $M,N \in \mathscr{E}$ and $V \in \mathscr{V}$. Let $E[M]$ be the execution term obtained when substituting $M$ to $[]$ in $E$.
Contexts are extended to configurations through the following syntactic sugar:
\[
 E\left[\sum_{i=1}^n \alpha_i \cdot M_i^{\sigma_i}\right] \ :=\ \sum_{i=1}^n \alpha_i \cdot E[M_i]^{\sigma_i}
\]

The reduction relation ${\twoheadrightarrow} \subseteq \vec{\mathscr{E}}\times \vec{\mathscr{E}}$ is defined in Figure~\ref{fig:rel} for well-formed configurations.  We remind you that the notation $M \twoheadrightarrow \vec{M}$ means that $M$ is an initial configuration.
The reduction implements a \emph{call-by-basis} reduction strategy, which is a special case of call-by-value in which the execution of terms is linear by construction~\cite{lineal}.

\begin{figure}[!t]
\hrulefill
\vspace*{-0.5cm}
\begin{center}

\[
\scalebox{0.9}{
 \begin{prooftree}
  \hypo{N\twoheadrightarrow \vec{N}}
  \infer1{MN\twoheadrightarrow M\vec{N}}
 \end{prooftree}}
\quad\
\scalebox{0.9}{
 \begin{prooftree}
  \hypo{M\twoheadrightarrow \vec{M}}
  \infer1{MV\twoheadrightarrow \vec{M}V}
 \end{prooftree}}
\quad\
\scalebox{0.9}{
 \begin{prooftree}
 \hypo{\phantom{\vec{N}}}
  \infer1{(\lambda x. M)V\twoheadrightarrow M\{V/x\}}
 \end{prooftree}}
\quad\
\scalebox{0.9}{
 \begin{prooftree}
  \hypo{N\twoheadrightarrow \vec{N}}
  \infer1{\pair{M,N}\twoheadrightarrow \pair{M,\vec{N}}}
 \end{prooftree}}
\quad\
\scalebox{0.9}{
 \begin{prooftree}
  \hypo{M\twoheadrightarrow \vec{M}}
  \infer1{\pair{M,V}\twoheadrightarrow\pair{\vec{M},V}}
 \end{prooftree}}
\]
\[
\scalebox{0.9}{
 \begin{prooftree}
  \hypo{M\twoheadrightarrow \vec{M}}
  \infer1{\mathtt{let} \ \pair{x,y} = M \ \mathtt{in}\ N \twoheadrightarrow \mathtt{let} \ \pair{x,y}=\vec{M}\ \mathtt{in}\ N}
 \end{prooftree}}
\qquad
\scalebox{0.9}{
 \begin{prooftree}
 \hypo{\phantom{\vec{M}}}
  \infer1{\mathtt{let} \ \pair{x,y} = \pair{V,W} \ \mathtt{in}\ N \twoheadrightarrow N\{V/x,W/y\}}
 \end{prooftree}}
\]
\[
\scalebox{0.9}{
 \begin{prooftree}[center=false]
  \hypo{M\twoheadrightarrow \vec{M}}
  \infer1{M;N \twoheadrightarrow \vec{M};N}
 \end{prooftree}}
\qquad
\scalebox{0.9}{
\begin{prooftree}[center=false]
 \hypo{\phantom{\vec{M}}}
  \infer1{\unit;N\twoheadrightarrow N}
 \end{prooftree}}
\qquad
\scalebox{0.9}{
 \begin{prooftree}[center=false]
 \hypo{k\in\{0,1\}}
 \hypo{U = {\begin{bmatrix}
                  u_{00} & u_{01} \\
                  u_{10} & u_{11}
                 \end{bmatrix}}
}
  \infer2[$(un)$]{U\ket{k}\twoheadrightarrow u_{0k} \cdot\ket{0} + u_{1k} \cdot\ket{1}}
 \end{prooftree}}
\]
\[
\scalebox{0.9}{
 \begin{prooftree}
  \hypo{P\twoheadrightarrow \vec{P}}
  \infer1{\meas\ d\triangleright P\ \{0\rightarrow M_0\ |\ 1\rightarrow M_1\}\twoheadrightarrow \meas\ d\triangleright \vec{P}\  \{0\rightarrow M_0\ |\ 1\rightarrow M_1\}}
 \end{prooftree}}
\]
\\
\[
\scalebox{0.9}{
 \begin{prooftree}
  \hypo{k'\in\{0,1\}}
  \infer1[$(m_k) \quad (\scriptsize k\in\{0,1\})$]{\meas\ d\triangleright \ket{k'}\ \{0\rightarrow M_0\ |\ 1\rightarrow M_1\}\twoheadrightarrow\delta_{k,k'}\cdot M_{k}^{[d\mapsto k]}}
 \end{prooftree}}
\]
\[
\scalebox{0.9}{
 \begin{prooftree}
  \hypo{P\twoheadrightarrow \vec{P}}
  \infer1{\qcase\ P\ \{0\rightarrow M_0\ |\ 1\rightarrow M_1\}\twoheadrightarrow \qcase\ \vec{P}\ \{0\rightarrow M_0\ |\ 1\rightarrow M_1\}}
 \end{prooftree}}
\]
\\
\[
\scalebox{0.9}{
\begin{prooftree}
\hypo{k\in\{0,1\}}
\hypo{t\text{ is fresh}}
 \infer2{\qcase\ \ket{k}\ \{0\rightarrow M_0 \ |\  1\rightarrow M_1\} \twoheadrightarrow \lambda t. \pair{\ket{k},M_k \, t}}
\end{prooftree}}
\]
\\
\[
\scalebox{0.9}{
\begin{prooftree}
\hypo{M \twoheadrightarrow \vec{M}}
\hypo{\sigma \sqcup \Mem(\vec{M})\ ||\ \Mem(\vec{N})}
\hypo{M^\sigma \notin \Supp(\vec{N})}
 \infer3[$(sup)$]{\alpha \cdot M^\sigma + \vec{N}\twoheadrightarrow \alpha \cdot \vec{M}^{\{\sigma\}}+ \vec{N}}
\end{prooftree}}
\]
\end{center}
\hrulefill
\caption{Reduction relation $\twoheadrightarrow$}
\label{fig:rel}
\end{figure}

We now provide a some intuition on the rules of Figure~\ref{fig:rel}.

The unitary rule $(un)$ is the one that introduces sums of terms. Namely, $U\ket{k}$ reduces to a linear combination whose amplitudes are given by matrix components of $U$.

Measurement statements are reduced by first evaluating the control term (here $P$). Then, if the control reaches a qubit state, we can either apply rule ($m_0$) or rule ($m_1$).
Rule $(m_k)$ reads as follows: $k$ is as the outcome of the measurement, and $\ket{k'}$ is the state of the qubit being measured. If $k$ and $k'$ are different, then the result must be $0$, hence the scalar $\delta_{k,k'}$. The memory function records the measurement outcome obtained by the measurement device $d$, which here is $k$. Notice that ($m_0$), ($m_1$) and ($sup$) are the only rules that produce non-trivial memory functions.

Similarly, reductions of the qcase statement involve first reducing the control term (hence the first rule for the qcase). If the control is given by a qubit state $\ket{k}$, this state must be preserved in the output. Therefore in the final term (in the second rule for the qcase), $\ket{k}$ is paired with the corresponding branch $M_k$, within a $\lambda$-abstraction.

In the case of an arbitrary configuration (rule $(sup)$), the relation is defined by reducing one summand at a time. The memory functions of each summand are constructed on the fly by recording outcomes of measurements as they occur.
Here, measurement statements in different summands that are assigned the same device should be understood as representing the same physical measurement. Therefore the outcomes occuring in each summand should be identical. This is formalized by the hypothesis $\sigma \sqcup \Mem(\vec{M})\ ||\ \Mem(\vec{N})$, which ensures that at each step of the reduction, the memory functions of each summand remain consistent.
The hypothesis $M^{\sigma}\notin \Supp(\vec{N})$ indicates that we are reducing one of the summands of the canonical form.
This is to avoid artificially constructing an infinite derivation, for instance by starting with $M\equiv \sfrac{1}{2}\cdot M + \sfrac{1}{4}\cdot M+...$ and reducing the fractions one by one.

The reflexive transitive closure of $\twoheadrightarrow$ will be written as $\twoheadrightarrow^*$.

Notice that the operational semantics is nondeterministic: first, a transition from a configurations is obtained by applying a reduction rule to any one of its summands. Second, if the chosen summand is a measurement, there are two distinct applicable rules, corresponding to each measurement outcome.

\begin{remark}
\label{rem:unitarity}
Concretely, the reduction of a term $M\in\mathscr{T}$ defined by the programmer will have the following form:
\[
 \alift{M}^{[\bot]}
 =\vec{N_0} \twoheadrightarrow \vec{N_1}\twoheadrightarrow \vec{N_2} \twoheadrightarrow \cdots
\]
where $ \alift{M}\in\mathscr{E}$ and $\vec{N_0},\vec{N_1},\vec{N_2}...\in\vec{\mathscr{E}}$.
Treating the original set of terms $\mathscr{T}$ and the set of configurations $\vec{\mathscr{E}}$ as distinct has the following considerable advantage. Sums of terms allow us to express coherent control in a small-step transition system, and simultaneously, starting from a term of $\mathscr{T}$ greatly simplifies checking that the computation is physically meaningful.

In quantum programming languages with linear combinations of terms, verifying unitarity, or more generally physicality, is known to be a major difficulty~\cite{realizability, valiron}.
Here, it is sufficient to check the initial term $M$ is physical, which amounts to checking that it is well typed.\footnote{The notion of ``physically meaningful'' will be clarified in Section~\ref{section:denotational}.}
In this case, the terms $\vec{N_0},\vec{N_1},\vec{N_2},...$ will be considered physical on account of being reachable from $M$.
This will be further justified when we state the soundness theorem (Theorem~\ref{thm:soundness}) in Section~\ref{section:soundness}.
\end{remark}

\begin{remark}
Although device references were introduced in order to define the reduction of processes with indefinite causal orders, they can also be seen as a solution for treating the nonlinear measurement operation within a linear transition system.
For instance, in the reduction
\begin{align*}
\meas \ d\triangleright (H\ket{0}) \ \{0 \to M\ |\ 1\to N\}
&\twoheadrightarrow \frac{1}{\sqrt{2}}\cdot \meas \ d\triangleright \ket{0} \ \{0 \to M\ |\ 1\to N\} \\
&\qquad +  \frac{1}{\sqrt{2}}\cdot \meas \ d\triangleright \ket{1} \ \{0 \to M\ |\ 1\to N\}
\end{align*}
where $H= \sfrac{1}{\sqrt{2}} \begin{psmallmatrix}1&1\\1&-1\end{psmallmatrix}$ is the Hadamard gate,
the outcomes of the measurements in the two summands must be synchronized.
Another solution for expressing measurement within a linear calculus is given in~\cite{meas2017,meas2019}, where the reduction strategy guarantees that measurements are treated consistently.
\end{remark}

\subsection{Examples}

\begin{example}[Quantum switch]
\label{ex:switch_reduction}
 We are now able to give an execution of the quantum switch.  We consider the term $\mathtt{switch} \ \pair{\Meas,\Meas} \ (H\ket{0}) \ \ket{0}$, where $H= \sfrac{1}{\sqrt{2}} \begin{psmallmatrix}1&1\\1&-1\end{psmallmatrix}$ is the Hadamard gate and  the term $\Meas = \lambda c. \meas\ c\ \{0\to\ket{0}\ |\ 1\to \ket{1}\}$ is taken from Section~\ref{section:design}.
 The first step is to assign each measurement to a distinct measurement device (see Section~\ref{section:design}). We will write this as
 \[
 \alift{ \mathtt{switch} \ \pair{\Meas,\Meas} \ (H\ket{0}) \ \ket{0} }\ = \ \mathtt{switch} \ \pair{\Measd{d},\Measd{e}} \ (H\ket{0}) \ \ket{0},
 \]
Then we have the following sequence of reductions:\version{\footnote{The step-by-step sequence of reductions is detailed in Appendix~\ref{app:example_switch}.}}{}
\begin{align*}
 &\mathtt{switch} \ \pair{\Measd{d},\Measd{e}} \ (H\ket{0}) \ \ket{0}\\
 &\quad \twoheadrightarrow^* \ \sfrac{1}{\sqrt{2}}\cdot  \bigl(\lambda q. \qcase\ q\ \{0\rightarrow \mathtt{comp} \ \pair{\Measd{d},\Measd{e}} \ |\ 1\rightarrow \mathtt{comp} \ \pair{\Measd{e},\Measd{d}} \}\bigr) \ \ket{0}\ \ket{0}
 \\
 &\quad \quad \ \ + \sfrac{1}{\sqrt{2}} \cdot \bigl(\lambda q. \qcase\ q\ \{0\rightarrow \mathtt{comp} \ \pair{\Measd{d},\Measd{e}} \ |\ 1\rightarrow \mathtt{comp} \ \pair{\Measd{e},\Measd{d}} \}\bigr)\ \ket{1}\ \ket{0}
 \\
 &\quad \twoheadrightarrow^* \ \sfrac{1}{\sqrt{2}}\cdot   \lambda s. \pair{\ket{0}, \mathtt{comp} \ \pair{\Measd{d},\Measd{e}}\ s}\,\ket{0}
 +\ \sfrac{1}{\sqrt{2}} \cdot\lambda s. \pair{\ket{1}, \mathtt{comp} \ \pair{\Measd{e},\Measd{d}}\ s}\,\ket{0}
\\
 &\quad \twoheadrightarrow^* \ \sfrac{1}{\sqrt{2}}\cdot   \pair{\ket{0}, \Measd{e}(\Measd{d}\ket{0})}
 +\ \sfrac{1}{\sqrt{2}} \cdot\pair{\ket{1}, \Measd{d}(\Measd{e}\ket{0})}
\end{align*}

\end{example}

\begin{remark}\label{rm:qcase}
 In light of this example, we can explain our typing rule for the $\qcase$ primitive (Figure~\ref{fig:typing_linear}) in more detail. One could opt for the following, perhaps simpler, rule:
 \[
\scalebox{0.9}{
 \begin{prooftree}
  \hypo{\Delta \vdash P:\q}
  \hypo{\Delta' \vdash M:B}
  \hypo{\Delta' \vdash N: B}
  \hypo{M,N\in\mathscr{C}}
  \infer4[qcase]{\Delta,\Delta'\vdash \qcase\ P\ \{0\rightarrow M\ |\ 1\rightarrow N\}: \q\otimes B}
 \end{prooftree}}
\]
where the overall type is given by tensoring a qubit (for the output of the control) with the type of the branches.
However this would enable the programmer to define unphysical terms. For example, the quantum switch term would reduce to a superposition of terms of the form $\pair{\ket{k},M}$ where $\ket{k}$ is the state of the control qubit. Then we could write:
\[
 \mathtt{let}\ \pair{v,w}=\pair{\ket{k},M} \ \mathtt{in}\ wv
\]
Essentially, this term plugs the output of the control qubit into the input of the target qubit (see Figure~\ref{fig:switch}).
One can easily show that the corresponding process is unphysical, for instance by choosing as branching statements the Hadamard gate $H$ and the phase gate $S$.

 An essential feature of our language is that typing alone should be enough to guarantee the physicality of a term. Hence we restrict the typing of $\qcase$ by requiring that each branch have a type of the form $A\lolli \q^n$.
\end{remark}

\begin{example}
\label{ex:meas_transfer}
The operational semantics works with the principle of synchronizing measurement outcomes occuring in different summands. Then, what happens if we force conflicting measurements outcomes in different branches of a $\qcase$?
 Consider the following program:
\begin{align*}
 \mathtt{N} \ :=&\ \left(\lambda w. \qcase\ H\ket{0} \ \{ 0 \rightarrow \lambda t. \pair{t,w\,\ket{0}} \ |\ 1\rightarrow \lambda t.\pair{t,w\,\ket{1}}\}\right) \Meas \ \ket{k}
 \end{align*}
 where $\Meas$ is defined as in Section~\ref{section:design}, $H$ is the Hadamard gate, $X$ the Pauli-$X$ gate and $k\in\{0,1\}$. In the 0 branch (resp. the 1 branch), the measured qubit is in state $\ket{0}$ (resp. $\ket{1}$), therefore the measurement outcome should be $\ket{0}$ (resp. $\ket{1}$). We write $\mathtt{N'}:= \alift{\mathtt{N}}^{[\bot]}$. Depending on the measurement outcomes, we have the two alternative reductions:\version{\footnote{See details in Appendix~\ref{app:meas_transfer}.}}{}
 \begin{align*}
  \mathtt{N'} \twoheadrightarrow^* &\ \frac{1}{\sqrt{2}}\cdot \pair{\ket{0},\ket{k},\ket{0}}^{[d\mapsto 0]} + 0\cdot \pair{\ket{1},\ket{k},\ket{0}}^{[d\mapsto 0]}
  \\
  \mathtt{N'} \twoheadrightarrow^* &\ 0 \cdot \pair{\ket{0},\ket{k},\ket{1}}^{[d\mapsto 1]} + \frac{1}{\sqrt{2}} \cdot \pair{\ket{1},\ket{k},\ket{1}}^{[d\mapsto 1]}
 \end{align*}
 Here the control qubit is the first qubit of the triple. Since each term only has one non-zero summand, we can conclude that forcing conflicting measurements equivalent to measuring the control qubit.

\end{example}

\subsection{Properties of the Operational Semantics}

\subsubsection{Safety Properties}

The language satisfies the following standard safety properties.

\begin{restatable}[Substitution]{lemma}{lemsubstitution}
\label{lem:substitution}
Let $M,N\in\mathscr{E}$.
Suppose $\Delta_1,x:A,\Delta_2 \vdash M:B$ and $\Delta_3 \vdash N:A$ are valid typing judgments such that $\Delta_1,\Delta_2$ and $\Delta_3$ have no variables in common.  Then $\Delta_1,\Delta_2,\Delta_3\vdash M\{N/x\}:B$ is a valid typing judgment.
\end{restatable}

The statement is proved by structural induction on the derivation of $\Delta_1,x:A,\Delta_2 \vdash M:B$. Consequently, the language enjoys the usual subject reduction property:
\begin{restatable}[Subject reduction]{proposition}{propsubjectreduction}
\label{prop:subject_reduction}
Let $M\in\mathscr{E}$.
If $\Delta\vdash M:A$ is a valid typing judgment and $M\twoheadrightarrow \sum_i \alpha_i \cdot M_i^{\sigma_i}$, then for all $i$, $\Delta\vdash M_i:A$ is a valid typing judgment.
\end{restatable}
This is proved by induction on the derivation of $M\twoheadrightarrow \sum_i \alpha_i \cdot M_i^{\sigma_i}$.

To prove progress, it is convenient to give the following definition:
\begin{definition}[Final configuration]
Let $\vec{M}=\sum_{i=1}^n \alpha_i\cdot M_i^{\sigma_i}$ be a well-formed configuration.
$\vec{M}$ is said to be a \emph{final configuration} if for all $i\in \{ 1,\ldots,n\}$, $M_i\in\mathscr{V}$.
The set of final configurations will be written as $\vec{\mathscr{V}}$.
\end{definition}

\begin{restatable}[Progress]{lemma}{lemprogress}\label{lem:progress}
\
\begin{itemize}
\item For all well-typed closed term $M\in\mathscr{E}$, either $M\in\mathscr{V}$ or there exists a reduction $M\twoheadrightarrow \vec{M}$.
\item For all well-formed closed configuration $\vec{M}$, either $\vec{M}\in\vec{\mathscr{V}}$ or there exists a reduction $\vec{M}\twoheadrightarrow \vec{M'}$.
\end{itemize}
\end{restatable}

\subsubsection{Uniqueness of the Normal Form}

In this section, we consider configurations $\vec{M}$ that can be reduced to a given final configuration $\vec{V}$.
There are two sources of nondeterminism in the reduction relation $\twoheadrightarrow$: the choice of summand to be reduced, and the choice of measurement outcome. When it comes to the choice of summand,
we show uniqueness of the normal form.
More precisely, suppose we have a configuration $\vec{M}$. Given a fixed set of measurement outcomes for all devices of $\text{Dev}(\vec{M})$, if $\vec{M}$ is a normalizing term, then all reductions from $\vec{M}$ lead to the same final configuration.
This is formalized by the following confluence-like property:
\begin{restatable}[Uniqueness of the normal form]{proposition}{propconfluence}
\label{prop:confluence}
Suppose that $\vec{M}\twoheadrightarrow^* \vec{N}$ and $\vec{M}\twoheadrightarrow^* \vec{V}$, where $\vec{V}\in\vec{\mathscr{V}}$ and $\Mem(\vec{N})\ ||\ \Mem(\vec{V})$. Then we have $\vec{N}\twoheadrightarrow^*\vec{V}$. Diagrammatically,
 \[
 \scalebox{0.9}
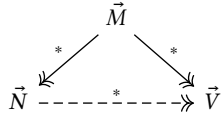
{
  \begin{tikzcd}[ampersand replacement=\&]
   \&\vec{M}
   \ar[dl,twoheadrightarrow,"*"']
   \ar[dr,twoheadrightarrow,"*"]
   \&
   \\
   \vec{N}
   \ar[rr,twoheadrightarrow,dashed,"*"]
   \&\&\vec{V}
\end{tikzcd}}
 \]
\end{restatable}

The consistency condition between the memory functions of $\vec{N}$ and $\vec{V}$ is necessary because the result only applies when measurement outcomes in either branch are not in contradiction.
Therefore, for the same measurement outcomes, the value distribution to which a term can be reduced (when it exists) is unique up to congruence.

\section{Denotational Semantics}
\label{section:denotational}

To define the denotational semantics, we give separate interpretations for the set of controllable terms $\mathscr{C}$ and the set of all terms $\mathscr{T}$.
Controllable terms, which are deterministic, are given a denotational semantics $[\cdot]$ in the category of finite-dimensional Hilbert spaces $\FHilb$ (Section~\ref{section:denotational_controllable}).
The full denotational semantics is based on the Caus construction~\cite{kissinger2019categorical}. Specifically, the interpretation is defined in $\Caus[\CPM]$, where $\CPM$ is the category of completely positive maps (Section~\ref{section:interpretation_full}).

\subsection{Compact Closed Categories}

The denotational semantics will be defined in terms of \emph{symmetric monoidal categories} (SMCs).
An SMC $(\cat,\otimes,I)$ is a category equipped with a parallel composition operation $\otimes$ for objects and morphisms, a unit object $I$ and a swap morphism $\gamma_{A,B}$ for all objects $A,B$ satisfying some consistency equations.
In this paper, SMCs will always be assumed to be strict, that is, satisfying the equalities $(A\otimes B)\otimes C = A\otimes (B\otimes C)$ and $A\otimes I = A = I\otimes A$.
We note that every SMC is equivalent to a strict one~\cite{maclane}.

\begin{definition}[Symmetric monoidal closed category]
 A \emph{symmetric monoidal closed category} (SMCC) is a symmetric monoidal category $(\cat,\otimes,I)$ such that for every object $B$, the functor $-\otimes B : \cat \to\cat$ has a right adjoint $B\lolli -:\cat\to\cat$.
 We will write the corresponding bijections as:
 \[
  \phi_{A,B,C}: \cat(A\otimes B,C)\cong \cat(A,B\lolli C).
 \]

\end{definition}
We will omit the subscript $A,B,C$ of $\phi_{A,B,C}$ when the objects involved are clear from the context.

Compact closed categories are a particular case of SMCCs:
\begin{definition}[Compact closed category]
An SMC $\cat$ is said to be \emph{compact closed} if every object has a \emph{dual object} $A^*$, i.e. there exist morphisms $\eta_A:I \to A^*\otimes A$ and $\epsilon_A:A\otimes A^* \to I$ satisfying:
\[
 (\epsilon_A\otimes 1_A)\circ(1_A\otimes \eta_A) = 1_A \qquad (1_{A^*}\otimes \epsilon_A)\circ(\eta_A\otimes 1_{A^*})=1_{A^*}
\]
\end{definition}

As in~\cite{kissinger2019categorical}, we assume the equalities $I^*=I$,
$A^{**}=A$ and $(A\otimes B)^*=A^*\otimes B^*$ throughout. This is without loss
of generality: in any compact closed category these three identifications are
canonical and coherent~\cite{kelly_laplaza}, so every compact closed category is
monoidally equivalent to one in which they hold strictly.

The morphisms $\epsilon_A$ and $\eta_A$ will be represented diagrammatically with cups and caps, respectively, so that the above equations can be written as:
\[
\scalebox{0.9}{
 \begin{tikzpicture}
	\begin{pgfonlayer}{nodelayer}
		\node [style=none] (0) at (0, 0) {};
		\node [style=none] (1) at (0, 1.25) {};
		\node [style=none] (2) at (0, -1.25) {};
		\node [style=none] (3) at (-1.5, 1.25) {};
		\node [style=none] (4) at (1.5, -1.25) {};
		\node [style=texthere] (5) at (-1.5, 0.8) {$A$};
		\node [style=texthere] (6) at (0, 0.25) {$A^*$};
		\node [style=texthere] (7) at (1.5, -1) {$A$};
	\end{pgfonlayer}
	\begin{pgfonlayer}{edgelayer}
		\draw (3.center) to (1.center);
		\draw [bend left=90, looseness=1.50] (1.center) to (0.center);
		\draw [bend left=270, looseness=1.50] (0.center) to (2.center);
		\draw (2.center) to (4.center);
	\end{pgfonlayer}
\end{tikzpicture}\ = \ \begin{tikzpicture}
	\begin{pgfonlayer}{nodelayer}
		\node [style=none] (0) at (-0.75, 0) {};
		\node [style=none] (1) at (0.75, 0) {};
		\node [style=texthere] (2) at (0.75, 0.25) {$A$};
	\end{pgfonlayer}
	\begin{pgfonlayer}{edgelayer}
		\draw (0.center) to (1.center);
	\end{pgfonlayer}
\end{tikzpicture}
 \qquad \begin{tikzpicture}
	\begin{pgfonlayer}{nodelayer}
		\node [style=none] (0) at (0, 0) {};
		\node [style=none] (1) at (0, 1.25) {};
		\node [style=none] (2) at (1.5, 1.25) {};
		\node [style=none] (3) at (0, -1.25) {};
		\node [style=none] (4) at (-1.5, -1.25) {};
		\node [style=texthere] (5) at (1.5, 0.8) {$A^*$};
		\node [style=texthere] (6) at (0, 0.25) {$A$};
		\node [style=texthere] (7) at (-1.5, -1) {$A^*$};
	\end{pgfonlayer}
	\begin{pgfonlayer}{edgelayer}
		\draw (1.center) to (2.center);
		\draw [bend right=90, looseness=1.50] (1.center) to (0.center);
		\draw [bend left=90, looseness=1.50] (0.center) to (3.center);
		\draw (4.center) to (3.center);
	\end{pgfonlayer}
\end{tikzpicture} \ = \ \begin{tikzpicture}
	\begin{pgfonlayer}{nodelayer}
		\node [style=none] (0) at (-0.75, 0) {};
		\node [style=none] (1) at (0.75, 0) {};
		\node [style=texthere] (2) at (0.75, 0.25) {$A^*$};
	\end{pgfonlayer}
	\begin{pgfonlayer}{edgelayer}
		\draw (0.center) to (1.center);
	\end{pgfonlayer}
\end{tikzpicture}

}
\]

To see that every compact closed category is an SMCC, we choose the functor $B\lolli - $ to be that defined by $B\lolli C:=B^*\otimes C$ for all object $C$, and $B\lolli f:=1_{B^*}\otimes f$ for all morphism $f:C\to C'$; together with the bijections
\[
 \phi_{A,B,C}\left(\scalebox{0.9}{\begin{tikzpicture}
	\begin{pgfonlayer}{nodelayer}
		\node [style=medium box] (0) at (0, 0) {$f$};
		\node [style=none] (1) at (-1.75, 0.75) {};
		\node [style=none] (2) at (-0.75, -0.75) {};
		\node [style=none] (3) at (0.75, 0) {};
		\node [style=none] (4) at (1.75, 0) {};
		\node [style=none] (5) at (-1.75, -0.75) {};
		\node [style=none] (8) at (-0.75, 0.75) {};
		\node [style=texthere] (9) at (-1.75, -0.5) {$B$};
		\node [style=texthere] (10) at (1.75, 0.25) {$C$};
		\node [style=texthere] (11) at (-1.75, 1) {$A$};
	\end{pgfonlayer}
	\begin{pgfonlayer}{edgelayer}
		\draw (5.center) to (2.center);
		\draw (4.center) to (3.center);
		\draw (1.center) to (8.center);
	\end{pgfonlayer}
\end{tikzpicture}}\right) \ = \ \scalebox{0.9}{\begin{tikzpicture}
	\begin{pgfonlayer}{nodelayer}
		\node [style=medium box] (0) at (0, -0.625) {$f$};
		\node [style=none] (1) at (-2.75, 0.125) {};
		\node [style=none] (2) at (-0.75, -1.375) {};
		\node [style=none] (3) at (0.75, -0.625) {};
		\node [style=none] (4) at (1.75, -0.625) {};
		\node [style=none] (5) at (-1.25, -1.375) {};
		\node [style=none] (6) at (-1.25, 1.125) {};
		\node [style=none] (7) at (1.75, 1.125) {};
		\node [style=none] (12) at (-0.75, 0.125) {};
		\node [style=texthere] (13) at (-1.25, -1.125) {$B$};
		\node [style=texthere] (14) at (1.75, -0.375) {$C$};
		\node [style=texthere] (15) at (-2.75, 0.375) {$A$};
		\node [style=texthere] (16) at (1.75, 1.375) {$B^*$};
	\end{pgfonlayer}
	\begin{pgfonlayer}{edgelayer}
		\draw (5.center) to (2.center);
		\draw (6.center) to (7.center);
		\draw (4.center) to (3.center);
		\draw (1.center) to (12.center);
		\draw [bend left=270, looseness=1.25] (6.center) to (5.center);
	\end{pgfonlayer}
\end{tikzpicture}}
 \qquad
  \phi_{A,B,C}^{-1}\left(\scalebox{0.9}{\begin{tikzpicture}
	\begin{pgfonlayer}{nodelayer}
		\node [style=medium box] (0) at (0, 0) {$g$};
		\node [style=none] (1) at (0.75, 0.75) {};
		\node [style=none] (2) at (-0.75, 0) {};
		\node [style=none] (3) at (0.75, -0.75) {};
		\node [style=none] (4) at (1.75, -0.75) {};
		\node [style=none] (5) at (-1.75, 0) {};
		\node [style=none] (6) at (1.75, 0.75) {};
		\node [style=texthere] (7) at (-1.75, 0.25) {$A$};
		\node [style=texthere] (8) at (1.75, -0.5) {$C$};
		\node [style=texthere] (9) at (1.75, 1) {$B^*$};
	\end{pgfonlayer}
	\begin{pgfonlayer}{edgelayer}
		\draw (5.center) to (2.center);
		\draw (4.center) to (3.center);
		\draw (1.center) to (6.center);
	\end{pgfonlayer}
\end{tikzpicture}}\right) \ = \ \scalebox{0.9}{\begin{tikzpicture}
	\begin{pgfonlayer}{nodelayer}
		\node [style=medium box] (0) at (0, 0.5) {$g$};
		\node [style=none] (1) at (0.75, 1.25) {};
		\node [style=none] (2) at (-0.75, 0.5) {};
		\node [style=none] (3) at (0.75, -0.25) {};
		\node [style=none] (4) at (2.75, -0.25) {};
		\node [style=none] (5) at (-1.75, 0.5) {};
		\node [style=none] (6) at (1.25, 1.25) {};
		\node [style=texthere] (7) at (-1.75, 0.75) {$A$};
		\node [style=texthere] (8) at (2.75, 0) {$C$};
		\node [style=texthere] (9) at (1.25, 1.5) {$B^*$};
		\node [style=none] (10) at (1.25, -1.25) {};
		\node [style=none] (11) at (-1.75, -1.25) {};
		\node [style=texthere] (12) at (-1.75, -1) {$B$};
	\end{pgfonlayer}
	\begin{pgfonlayer}{edgelayer}
		\draw (5.center) to (2.center);
		\draw (4.center) to (3.center);
		\draw (1.center) to (6.center);
		\draw (11.center) to (10.center);
		\draw [bend right=90, looseness=1.25] (10.center) to (6.center);
	\end{pgfonlayer}
\end{tikzpicture}}
\]
for all objects $A,B,C$, where $\scalebox{0.9}{\begin{tikzpicture}
	\begin{pgfonlayer}{nodelayer}
		\node [style=none] (0) at (0, -0.75) {};
		\node [style=none] (1) at (0, 0.75) {};
		\node [style=texthere] (2) at (0, 1) {$B^*$};
		\node [style=texthere] (3) at (0, -0.5) {$B$};
	\end{pgfonlayer}
	\begin{pgfonlayer}{edgelayer}
		\draw [bend right=90, looseness=1.75] (0.center) to (1.center);
	\end{pgfonlayer}
\end{tikzpicture}
}\ := \ \scalebox{0.9}{\begin{tikzpicture}
	\begin{pgfonlayer}{nodelayer}
		\node [style=none] (0) at (0.25, -0.75) {};
		\node [style=none] (1) at (0.25, 0.75) {};
		\node [style=none] (2) at (-1, 0.75) {};
		\node [style=none] (3) at (-1, -0.75) {};
		\node [style=texthere] (4) at (-1, 1) {$B^*$};
		\node [style=texthere] (5) at (-1, -0.5) {$B$};
	\end{pgfonlayer}
	\begin{pgfonlayer}{edgelayer}
		\draw [in=-180, out=0] (3.center) to (1.center);
		\draw [in=180, out=0] (2.center) to (0.center);
		\draw [bend right=90, looseness=1.25] (0.center) to (1.center);
	\end{pgfonlayer}
\end{tikzpicture}}$.

\subsection{A Semantics of Controllable Terms}
\label{section:denotational_controllable}

\subsubsection{Higher-Order Linear Maps}
\label{section:Choi}
Controllable terms will be interpreted in $\FHilb$, the category of {finite dimensional} Hilbert spaces and $\mathbb{C}$-linear maps between them. The set of linear maps from a Hilbert space $\mathcal{H}$ to a Hilbert space $\mathcal{K}$ will be written as $\mathcal{L}(\mathcal{H},\mathcal{K})$. $\FHilb$ is a symmetric monoidal category, where $\otimes$ is the usual tensor product, and is compact closed. The dual of an object $\mathcal{H}$ is the Hilbert space $\mathcal{H}^*$ of linear forms on $\mathcal{H}$.
For any $\mathcal{H} \in \FHilb$, $\eta_\mathcal{H}$ and $\epsilon_\mathcal{H}$ are given by:
\[
 \eta_\mathcal{H}: 1\mapsto \sum_i \ket{i}_{\mathcal{H}^*}\otimes\ket{i}_{\mathcal{H}} \qquad
 \epsilon_{\mathcal{H}}: \ket{i}_{\mathcal{H}}\otimes \ket{j}_{\mathcal{H}^*} \mapsto \braket{j}{i}_\mathcal{H} = \braket{i}{j}_{\mathcal{H}^*}
\]
where $(\ket{i}_{\mathcal{H}})_i$ is an orthonormal basis of $\mathcal{H}$ and $(\ket{i}_{\mathcal{H}^*})_i$ is the corresponding dual basis of $\mathcal{H}^*$.
The definitions of $\eta_\mathcal{H}$ and $\epsilon_{\mathcal{H}}$ are independent of the choice of basis.
The compact closed structure of $\FHilb$ enables us to define a semantics for higher order linear maps. This is the same structure as that given by the Choi-Jamiołkowski (CJ) isomorphism~\cite{choi,jamiolkowski}.

\subsubsection{Interpretation in $\FHilb$}
Controllable terms $\mathscr{C}$ are given an interpretation $[\cdot]$:
\begin{itemize}
 \item to each type $A$ we assign an object $[A]$ of $\FHilb$:
 \begin{align*}
[\one] & := I = \mathbb{C} & [A\otimes B] & := [A]\otimes [B]
\\
 [\q] & := \mathbb{C}^2 & [A\lolli B] &:= [A]\lolli [B]
\end{align*}

 \item to each typing context $\Delta = x_1:A_1,...,x_n:A_n$ we assign the object $[\Delta]:= [A_1]\otimes ... \otimes [A_n]$;
 \item to each valid typing judgment $\Delta\vdash M:A$ we assign a morphism $[\Delta]\to[A]$. The interpretation is defined in Figure~\ref{fig:denotational_controllable}, by induction on the derivation of the typing judgments.
 For the qcase rule, we define $\mathfrak{q}(f,g):[\q]\otimes[\Delta']\to [A\lolli \q^n]$ as follows, for $f,g:[\Delta']\to [A\lolli \q^n]$:
 \[
 \mathfrak{q}(f,g) := (\gamma_{[\q],[A]^*}\otimes I_{[\q^n]})(\ket{0}\bra{0}\otimes f + \ket{1}\bra{1}\otimes g).
\]
\end{itemize}

\begin{figure}[!t]
\hrulefill
\vspace*{-0.6cm}
\begin{center}
\[
\scalebox{0.8}{
 \begin{prooftree}
  \hypo{\vphantom{k\in\{0,1\}}}
 \infer1{[x:A \vdash x:A] \ =\ [A]\xrightarrow{1_{[A]}}[A]}
\end{prooftree}}
 \qquad
 \scalebox{0.8}{
 \begin{prooftree}
 \hypo{\vphantom{k\in\{0,1\}}}
  \infer1{[\vdash U:\q \lolli \q]\ =\  I\xrightarrow{(I\otimes U)\circ \eta_{[\q]}}[\q\lolli \q]}
 \end{prooftree}}
 \qquad
 \scalebox{0.8}{
 \begin{prooftree}
  \hypo{k\in\{0,1\}}
  \infer1{[\vdash \ket{k}:\q]\ = \ I \xrightarrow{1\mapsto \ket{k}} [\q]}
 \end{prooftree}}
\]
\vspace{0.1ex}
\[
\scalebox{0.8}{
 \begin{prooftree}
  \hypo{[\Delta,x:A,y:B,\Delta'\vdash M:C] = [\Delta]\otimes[A]\otimes[B]\otimes [\Delta'] \xrightarrow{f} [C]}
  \infer1{[\Delta,y:B,x:A,\Delta'\vdash M:C] = [\Delta]\otimes[B]\otimes[A]\otimes [\Delta'] \xrightarrow{1_{[\Delta]}\otimes \gamma_{[B],[A]}\otimes 1_{[\Delta']}}[\Delta]\otimes[A]\otimes[B]\otimes [\Delta'] \xrightarrow{f} [C]}
 \end{prooftree}}
\]
\vspace{0.1ex}
 \[
 \scalebox{0.8}{
\begin{prooftree}
 \hypo{[\Delta,x:A\vdash M:B] \ = \ [\Delta]\otimes [A] \xrightarrow{f} [B]}
 \infer1{[\Delta\vdash \lambda x. M:A\lolli B] \ = \ [\Delta] \xrightarrow{\phi(f)}[A]\lolli[B]}
 \end{prooftree}}
 \quad
 \scalebox{0.8}{
 \begin{prooftree}
  \hypo{[\Delta\vdash M:A\lolli B] \ = \ [\Delta]\xrightarrow{f} [A\lolli B]
}
\hypo{[\Delta'\vdash N:A] \ = \ [\Delta'] \xrightarrow{g}[A]}
  \infer2{[\Delta,\Delta'\vdash MN:B] \ = \ [\Delta]\otimes [\Delta'] \xrightarrow{f\otimes g} [A\lolli B]\otimes [A] \xrightarrow{\phi^{-1}(1_{A\lolli B})} [B]}
 \end{prooftree}}
 \]
\vspace{0.1ex}
 \[
 \scalebox{0.8}{
 \begin{prooftree}
  \hypo{[\Delta \vdash M:A] \ =\ [\Delta]\xrightarrow{f}[A]}
  \hypo{[\Delta'\vdash N:B] \ =\ [\Delta'] \xrightarrow{g}[B]}
  \infer2{[\Delta,\Delta' \vdash \pair{M , N}:A\otimes B]\ = \ [\Delta]\otimes[\Delta'] \xrightarrow{f\otimes g} [A]\otimes[B] }
 \end{prooftree}}
 \]
\vspace{0.1ex}
\[
\scalebox{0.8}{
  \begin{prooftree}
\hypo{[\Delta \vdash M:A\otimes B]\ = \ [\Delta] \xrightarrow{f}[A\otimes B]
}
\hypo{[\Delta',x:A,y:B\vdash N:C]\ = \ [\Delta']\otimes [A]\otimes [B] \xrightarrow{g}[C]}
\infer2{[\Delta,\Delta' \vdash \mathtt{let}\ \pair{x,y}=M\ \mathtt{in}\ N:C] \ = \ [\Delta]\otimes [\Delta'] \xrightarrow{f\otimes 1_{[\Delta']}} [A\otimes B] \otimes [\Delta'] \xrightarrow{\gamma_{[A\otimes B], [\Delta']}} [\Delta']\otimes [A\otimes B]\xrightarrow{g} [C]}
\end{prooftree}}
\]
\vspace{0.1ex}
\[
\scalebox{0.8}{
 \begin{prooftree}
 \hypo{\vphantom{\xrightarrow{f}}}
  \infer1{[\vdash \unit:\one] \ = \ I\xrightarrow{1_{I}} I}
 \end{prooftree}}
 \qquad
 \scalebox{0.8}{
 \begin{prooftree}
 \hypo{[\Delta\vdash M: \one]\ = \ [\Delta] \xrightarrow{f} I}
  \hypo{[\Delta' \vdash N: A]\ = \ [\Delta']\xrightarrow{g}[A]}
  \infer2{[\Delta,\Delta'\vdash M;N:A] = \ [\Delta]\otimes[\Delta'] \xrightarrow{f\otimes g} [A]}
 \end{prooftree}}
\]
\vspace{0.1ex}
\[
\scalebox{0.8}{
 \begin{prooftree}
  \hypo{[\Delta \vdash P:\q] = [\Delta] \xrightarrow{p}[\q]}
  \hypo{[\Delta' \vdash M:A \lolli \q^n] = [\Delta'] \xrightarrow{f}[A\lolli \q^n]}
  \hypo{[\Delta' \vdash N: A\lolli \q^n] = [\Delta'] \xrightarrow{g} [A\lolli \q^n]}
  \infer3{[\Delta,\Delta'\vdash \qcase\ P\ \{0\rightarrow M\ |\ 1\rightarrow N\}: A\lolli \q^{n+1}]\ = \ [\Delta]\otimes[\Delta'] \xrightarrow{p\otimes 1_{[\Delta']}} [\q]\otimes [\Delta']\xrightarrow{\mathfrak{q}(f,g)}[A\lolli \q^{n+1}]}
 \end{prooftree}}
\]
\end{center}
\hrulefill
\caption{Denotational semantics of controllable terms}
\label{fig:denotational_controllable}
\end{figure}

We briefly explain the semantics for the qcase rule.
Fix linear maps $f,g:[\Delta']\to[A\lolli \q^n]$.
The controlled statements are not $f$ and $g$ themselves, but rather $\phi^{-1}(f),\phi^{-1}(g):[\Delta']\otimes [A] \to[\q^n]$, where the input on $[\Delta']$ can be seen as representing some parameters for the free variables of the branching statements $M$ and $N$.
Therefore, the qcase statement is given by the linear map below:
\[
 \ket{0}\bra{0}\otimes \phi^{-1}(f) + \ket{1}\bra{1} \otimes \phi^{-1}(g)\ :\ [\q]\otimes [\Delta']\otimes [A] \to [\q^{n+1}],
\]
which we will represent as
\[
 \mathfrak{q}(f,g)\ := \phi\bigl(\ket{0}\bra{0}\otimes \phi^{-1}(f) + \ket{1}\bra{1} \otimes \phi^{-1}(g)\bigr)\ :\ [\q]\otimes [\Delta'] \to [A]\lolli [\q^{n+1}].
\]
This expression can be simplified to
\[
 \mathfrak{q}(f,g)\ = (\gamma_{[\q],[A]^*}\otimes I_{[\q^n]})\circ (\ket{0}\bra{0}\otimes f + \ket{1}\bra{1} \otimes g).
\]

\begin{restatable}{proposition}{proppurewelldefined}
\label{prop:pure_well_defined}
The interpretation $[\cdot]$ is well-defined.
\end{restatable}
Due to the exchange rule, the derivation of a judgment is generally not unique.
 The proof of well-definedness involves showing that the interpretation is independent of the choice of derivation.

\subsection{A Causal Semantics}
\label{section:interpretation_full}

\subsubsection{Higher-Order Quantum Maps and Causal Categories}
\label{section:causal}

The full fragment will be interpreted in the category $\Caus[\CPM]$.
First,
$\CPM$ is the category of finite-dimensional Hilbert spaces and completely positive maps. We recall that a map $\mathcal{C}\in\mathcal{L}(\mathcal{L}(\mathcal{H}),\mathcal{L}(\mathcal{K}))$ is completely positive if for all auxiliary Hilbert space $\mathcal{E}$ and all positive matrix $\rho\in\mathcal{L}(\mathcal{H}\otimes \mathcal{E})$, the matrix $(\mathcal{C}\otimes \mathcal{I}_{\mathcal{E}})(\rho)\in \mathcal{L}(\mathcal{K}\otimes\mathcal{E})$ is also positive.
Like $\FHilb$, the category $\CPM$ is a symmetric monoidal category, with $\otimes$ the usual tensor product. Moreover, $\CPM$ is compact closed, with dual objects defined as dual vector spaces, and cups and caps defined by
\[
 \eta'_\mathcal{H}: 1\mapsto \sum_{i,j} \ket{i}_{\mathcal{H}^*}{}_{\mathcal{H}^*}\bra{j}\otimes\ket{i}_{\mathcal{H}}{}_{\mathcal{H}}\bra{j}
 \qquad
 \epsilon_{\mathcal{H}}': \rho \mapsto \left(\sum_i {}_{\mathcal{H}} \bra{i}\otimes {}_{\mathcal{H}^*}\bra{i}\right)\,\rho\,\left(\sum_i \ket{i}_{\mathcal{H}}\otimes\ket{i}_{\mathcal{H}^*}\right)
\]
Further, there exists a identity-on-objects functor $D:\FHilb \to \CPM$ that ``doubles'' linear maps to construct completely positive ones, that is $D(f) = [\rho \mapsto f\rho f^\dag]$.
$D$ preserves the compact closed structure of $\FHilb$. In particular, we have $D(f\otimes g)=D(f)\otimes D(g)$, $D(\eta_{\mathcal{H}})=\eta'_{\mathcal{H}}$ and $D(\epsilon_{\mathcal{H}})=\epsilon'_{\mathcal{H}}$.

Similarly to $\FHilb$, the compact closed structure in $\CPM$ corresponds the Choi-Jamiołkowski isomorphism in quantum theory\footnote{Specifically the variant given in~\cite{common_causes,causal_models} which explicitly uses dual vector spaces.}, which
plays an important role in representing higher-order quantum maps.
Quantum supermaps are defined as higher-order linear maps that send quantum channels to quantum channels, including when tensored with an identity map~\cite{chiribella}. Via the CJ isomorphism, a quantum supermap can be transformed into a linear map that sends CJ representations of quantum channels to CJ representations of quantum channels. It was shown that such a linear map is completely positive~\cite[Theorem 1]{chiribella}.
In turn, its CJ representation is a positive-semidefinite matrix. By this process, all higher-order quantum maps can be represented as positive semidefinite matrices, and therefore in the category $\CPM$.

Although it is remarkable that compact closure collapses all higher-order quantum maps to order 1, the disadvantage is that there is no way to keep track of whether maps are trace preserving (or trace preserving preserving, and further variations).
A solution to this problem has already been studied in~\cite{kissinger2019categorical}: causal categories.
We give a short summary of the construction.

A \emph{precausal category} is defined as a compact closed category {$\cat$} satisfying four additional axioms.
Notably, $\cat$ is required to have discarding maps $\mathfrak{d}_A \in\cat(A, I)$ for each object $A$, which are compatible with the monoidal structure.
E.g., $\CPM$ is a precausal category, whose discard maps are traceout functions.
Given a precausal category $\cat$,  the causal category $\Caus[\cat]$ is constructed as follows:

\begin{definition}{\cite{kissinger2019categorical}}
 Let $\cat$ be a precausal category. The dual of a set $c\subseteq \cat(I,A)$ is defined as
 \[
  c^*:= \left\{ \pi\in\cat(I,A^*)\ \middle|\  \forall \rho \in c:\epsilon_A \circ (\pi\otimes \rho) \ = 1_I\right\}
 \]
A set $c\in\cat(I,A)$ is \emph{closed} if $c=c^{**}$ and \emph{flat} if there exist invertible $\lambda,\mu\in\cat(I,I)$ such that $\lambda \mathfrak{d}_A\in c$ and $\mu \mathfrak{d}^*_A\in c^*$, where $\mathfrak{d}_A^*$ is the transpose of $\mathfrak{d}_A$, i.e $\mathfrak{d}_A^*:=(\mathfrak{d}_{A^*}\otimes 1_A)\circ\eta_A$.\footnote{$\lambda,\mu$ being invertible means that $\lambda,\mu$ are isomorphisms; $\lambda \mathfrak{d}_A$ and $\mu \mathfrak{d}^*_A$ should be read as $\lambda \otimes \mathfrak{d}_A$ and $\mu \otimes \mathfrak{d}^*_A$.}
$\Caus[\cat]$ is defined as the category whose objects are pairs $\mathbf{A}=(A,c_\mathbf{A})$ where $A$ is an object of $\cat$ and $c_\mathbf{A}\subseteq \cat(I,A)$ is closed and flat. Morphisms $A\stackrel{f}{\to}B$ in $\Caus[\cat]$ are morphisms $A\stackrel{f}{\to}B$ in $\cat$ such that $\rho \in c_\mathbf{A} \Rightarrow f\circ \rho \in c_\mathbf{B}$.
\end{definition}

Tensor products are defined in $\Caus[\cat]$ as follows. Given objects $\mathbf{A}=(A,c_\mathbf{A})$ and $\mathbf{B}=(B,c_\mathbf{B})$ of $\Caus[\cat]$, we define
\[
 c_\mathbf{A}\otimes c_\mathbf{B} := \{\rho_1\otimes \rho_2 \ |\ \rho_1\in c_\mathbf{A} \text{ and }\rho_2\in c_\mathbf{B}\}.
\]
Then, $\mathbf{A}\otimes\mathbf{B}:=(A\otimes B,c_{\mathbf{A}\otimes \mathbf{B}})$ where $c_{\mathbf{A}\otimes \mathbf{B}}:=(c_\mathbf{A}\otimes c_\mathbf{B})^{**}$. Equipped with this tensor product, $\Caus[\cat]$ is shown to be a symmetric monoidal category. Moreover, duals are defined as $\mathbf{A}^*:=(A^*,c_{\mathbf{A}^*})$ with $c_{\mathbf{A}^*}:=c_\mathbf{A}^*$.
With $\mathbf{A}\lolli\mathbf{B}:=(\mathbf{A}\otimes \mathbf{B}^*)^*$, $\Caus[\cat]$ is shown to be $*$-autonomous (and therefore a symmetric monoidal closed category). However, $\Caus[\cat]$ is not compact closed.

The `par' operation $\parr$ is defined as $\mathbf{A}\parr \mathbf{B} := (\mathbf{A}^*\otimes\mathbf{B}^*)^*$. In particular, it is shown that $c_{\mathbf{A}\otimes \mathbf{B}}\subseteq c_{\mathbf{A} \parr \mathbf{B}}$. Intuitively, $\parr$ constructs a larger joint system than $\otimes$.

\subsubsection{Interpretation in $\Caus{[}\CPM{]}$}

We can now define the denotational semantics in the category $\Caus[\CPM]$.
Similarly to the pure semantics $[\cdot]$, the semantics $\sem{\cdot}$ is defined by the following:
\begin{itemize}
 \item to each type $A$ we assign an object $\sem{A}$ of $\Caus[\CPM]$. Objects of $\Caus[\CPM]$ have the form $\mathbf{A}=(\mathcal{H},c_{\mathbf{A}})$ where $\mathcal{H}$ is a Hilbert space and $c_{\mathbf{A}}$ is a set of completely positive maps from $\mathcal{L}(\mathbb{C})$ to $\mathcal{\mathcal{L}}(\mathcal{H})$. 
 By a slight abuse of notation, we will identify the completely positive map $\rho\in \mathcal{L}(\mathcal{L}(\mathbb{C}),\mathcal{L}(\mathcal{H}))$ with the positive semi-definite matrix $\rho(1)\in \mathcal{L}(\mathcal{H})$. We define:
 \begin{align*}
\sem{\one} & := \mathbf{I} = (\mathbb{C},\{1\})
&
\sem{A\otimes B} & := \sem{A}\otimes \sem{B}
\\
 \sem{\q} & :=(\mathbb{C}^2,St(\mathbb{C}^2))
 &
 \sem{A\lolli B} &:= \sem{A}\lolli \sem{B}
\end{align*}
where $St(\mathbb{C}^2)$ is the set of positive semidefinite matrices in $\mathcal{L}(\mathbb{C}^2)$ of trace 1;
 \item to each typing context $\Delta = x_1:A_1,...,x_n:A_n$ we assign the object $\sem{\Delta}:= \sem{A_1}\otimes ... \otimes \sem{A_n}$;
 \item to each valid typing judgment $\Delta\vdash M:A$ we assign a morphism $\sem{\Delta}\to\sem{A}$.
The semantics is defined by induction on the derivation of the typing judgment, and consists of the same cases as controllable terms, plus measurement. \version{The full definition is given in Appendix~\ref{app:denotational}, Figure~\ref{fig:denotational}, with most cases being similar to the pure semantics. The cases with notable differences are the following:}{Most cases are similar to the pure semantics—the cases with notable differences are the following:}
\[
\scalebox{0.8}{
 \begin{prooftree}
  \infer0{\sem{\vdash U:\q \lolli \q}\ =\  \mathbf{I}\xrightarrow{D(I\otimes U)\circ \eta'_{[q]}}\sem{\q\lolli \q}}
 \end{prooftree}}
 \qquad
 \scalebox{0.8}{
 \begin{prooftree}
  \hypo{k\in\{0,1\}}
  \infer1{\sem{\vdash \ket{k}:\q}\ = \ I \xrightarrow{1\mapsto \ket{k}\bra{k}} \sem{\q}}
 \end{prooftree}}
\]
\[
\scalebox{0.8}{
 \begin{prooftree}
 \hypo{\sem{\Delta \vdash P:\q}\ = \ \sem{\Delta} \xrightarrow{p} \sem{\q}}
  \hypo{\sem{\Delta' \vdash M:A}\ = \ \sem{\Delta'}\xrightarrow{f}\sem{A}}
  \hypo{\sem{\Delta' \vdash N: A} \ = \ \sem{\Delta'} \xrightarrow{g} \sem{A}}
  \infer3{\sem{\Delta,\Delta'\vdash \meas\ P\ \{0\rightarrow M\ |\ 1\rightarrow N\}: A}\ = \ \sem{\Delta}\otimes\sem{\Delta'}\xrightarrow{p\otimes 1_{\sem{\Delta'}}} \sem{\q}\otimes \sem{\Delta'} \xrightarrow{\mathfrak{M}(f,g)}\sem{A}}
 \end{prooftree}}
\]
\[
\scalebox{0.8}{
 \begin{prooftree}
  \hypo{\sem{\Delta \vdash P:\q} = \sem{\Delta} \xrightarrow{p}\sem{\q}}
  \hypo{[\Delta' \vdash M:A\lolli \q^n] = [\Delta'] \xrightarrow{f}[A\lolli \q^n]}
  \hypo{[\Delta' \vdash N: A\lolli \q^n] = [\Delta'] \xrightarrow{g} [A\lolli\q^n]}
  \infer3{\sem{\Delta,\Delta'\vdash \qcase\ P\ \{0\rightarrow M\ |\ 1\rightarrow N\}: A\lolli \q^{n+1}}\ = \ \sem{\Delta}\otimes\sem{\Delta'} \xrightarrow{p\otimes 1_{\sem{\Delta'}}} \sem{\q}\otimes \sem{\Delta'}\xrightarrow{D(\mathfrak{q}(f,g))}\sem{A\lolli\q^{n+1}}}
 \end{prooftree}}
\]
For the measurement rule, the morphism $\mathfrak{M}(f,g)$ is defined for all $f,g:\sem{\Delta'}\to \sem{A}$ as:
\[
 \mathfrak{M}(f,g) = \mathcal{P}_0\otimes f+ \mathcal{P}_1\otimes g : \sem{\q} \otimes \sem{\Delta'} \to \sem{A}
\]
where $\mathcal{P}_k =\bra{k}(\cdot)\ket{k}$ is the projection on the computational basis state $\ket{k}$ ($k\in\{0,1\}$).\footnote{Note that the $\mathcal{P}_k$ are not morphisms of $\Caus[\CPM]$.}
The qcase rule is based on the linear semantics of the branching statements (which are necessarily controllable).
\end{itemize}

We give some intuition regarding the interpretation of measurement statement. Suppose $A=\q\lolli \q$, that is, the branching statements are channels from one qubit to one qubit.
Let $f,g:\sem{\Delta'}\to \sem{\q\lolli \q}$.
We emphasize that $f$ and $g$ are not channels representing the branching statements.
Rather, the branching statements are given by $\phi^{-1}(f)$ and $\phi^{-1}(g):\sem{\Delta'}\otimes\sem{\q} \to \sem{\q}$, where the input on $\sem{\Delta'}$ can be seen as parametrizing the branching statements $M$ and $N$.
Therefore, the measurement statement is given by $\mathcal{P}_0 \otimes \phi^{-1}(f) + \mathcal{P}_1 \otimes \phi^{-1}(g)$.
Taking
\[
 \mathfrak{M}(f,g)\ := \ \phi\bigl(\mathcal{P}_0 \otimes \phi^{-1}(f) + \mathcal{P}_1 \otimes \phi^{-1}(g)\bigr) : \sem{\q}\otimes \sem{\Delta'} \to \sem{\q}\lolli \sem{\q},
\]
we obtain the final expression, which is simply $\mathfrak{M}(f,g) = \mathcal{P}_0 \otimes f + \mathcal{P}_1 \otimes g$. This formula can clearly be generalized to an arbitrary type $A$.

To show that the denotational semantics is well-defined, we check that the morphisms are indeed in $\Caus[\CPM]$. As a first step, doubling the interpretation $[\cdot]$ yields morphisms of $\Caus[\CPM]$:
\begin{restatable}{lemma}{lemdouble}
\label{lem:double}
Let $M$ be a controllable term and $\Delta \vdash M:A$ a valid typing judgment. If $[\Delta \vdash M:A]=[\Delta] \xrightarrow{f} [A]$, then $\sem{\Delta} \xrightarrow{D(f)} \sem{A}$ is a morphism of $\Caus[\CPM]$.
\end{restatable}

This enables us to prove well-definedness of the semantics:
\begin{restatable}{proposition}{propmixedwelldefined}
 \label{prop:mixed_well_defined}
 The interpretation $\sem{\cdot}$ is well defined.
\end{restatable}

The fact that the interpretation of typing judgments is well defined not only in $\CPM$ but also in the causal category $\Caus[\CPM]$ guarantees
the physicality of the language. The ground type $\q$ is interpreted as the object $(\mathbb{C}^2,St(\mathbb{C}^2))$ which is a first order system by the definition of~\cite[Section 5]{kissinger2019categorical}.
Therefore by construction, the interpretation {of} each type corresponds to a physically meaningful collection of higher-order quantum maps (for instance quantum channels, quantum supermap or quantum supermaps on no-signalling channels).
Hence, typing alone is enough to decide whether a given term implements a physically meaningful operation.

We conclude this section with some remarks on tensor product types.
For qubit systems, the type $\q\otimes \q$ has the interpretation $(\mathbb{C}^4,St(\mathbb{C}^4))$, which corresponds to all possible states of 2 qubits.
For second order systems, the tensor product actually does not correspond to the entire joint system. For instance, the type $(q\lolli q)\otimes (q\lolli q)$ corresponds to channels from 2 qubits to 2 qubits that are \emph{no-signalling}. That is, the first input qubit cannot influence the second output qubit and vice-versa.
No signalling channels include, but are not limited to, product channels.
The set of \emph{all} channels from 2 qubits to 2 qubits is instead given by the type $(q\otimes q)\lolli (q\otimes q)$.

As an example, the type of the composition term $\mathtt{comp}$ restricts to $(\q\lolli \q)\otimes(\q\lolli \q)\lolli \q\lolli \q$. Hence its input must be a no-signalling channel from 2 qubits to 2 qubits.
Remarkably, we are able to apply composition to inputs that cannot be decomposed into products of channels.

\begin{example}
 Let $\mathtt{U}$ and $\mathtt{V}$ be terms of type $\q^2\lolli \q^2$ and $\q^3\lolli \q^3$, respectively, which we will assume to implement 2- and 3-qubit unitary gates. Then the term
\[
  \mathtt{share}\ := \ \mathtt{let}\ \pair{t,x} = \mathtt{CNOT}\ (H\ket{0})\ \ket{0}
  \ \mathtt{in}\
  \pair{\lambda s. \mathtt{U} \, \pair{s,t} , \lambda \pair{y,z}. \mathtt{V}\,\pair{x,y,z}}
\]
 has type $(\q\lolli \q^2)\otimes (\q^2\lolli \q^3)$.
First, we create a pair of entangled qubits stored in the variables $t$ and $s$. Second, we form a pair in which $\mathtt{U}$ and $\mathtt{V}$ each get one qubit of the entangled pair. Therefore, the overall term can be interpreted as two quantum channels with shared entanglement.

Then, the term $\mathtt{comp}\ \mathtt{share}$ is a well typed (of type $\q\lolli \q^3$).

\end{example}

\section{Soundness}
\label{section:soundness}

\subsection{Valuations}

Since we have defined an operational semantics over the set of configurations $\vec{\mathscr{E}}$ containing linear combinations of terms and device references, in order to state the soundness of the denotational semantics with respect to the operational semantics, we must also extend the definition of the denotational semantics to $\vec{\mathscr{E}}$.
The idea of this semantics is to choose a fixed set of measurement outcomes and give a pure interpretation assuming those are the outcomes that occur. Formally, the chosen set of outcomes is described by a valuation.

Given $M\in \mathscr{E}$, the set $\Omega(M)$ of \emph{valuations} of $M$ is the subset of $\mathcal{M}$ defined inductively by:

\begin{minipage}{0.4\textwidth}
 \begin{align*}
 \Omega(x) &:= [\bot]
 \\
 \Omega(MN)&:= \Omega(M)\sqcup \Omega(N)
 \\
 \Omega(\lambda x. M) &:= \Omega(M)
 \\
 \Omega(\pair{M,N}) &:= \Omega(M)\sqcup \Omega(N)
 \\
 \Omega(\mathtt{let}\ \pair{x,y}=M\ \in\ N)&:= \Omega(M)\sqcup \Omega(N)
 \end{align*}
\end{minipage}
\begin{minipage}{0.6\textwidth}
\begin{align*}
 \Omega(\unit)&:= [\bot]
 \\
 \Omega(M;N)&:=\Omega(M)\sqcup \Omega(N)
 \\
 \Omega(U)&:=[\bot]
 \\
 \Omega(\ket{k})&:=[\bot]
\end{align*}
\end{minipage}
\vspace*{-0.15cm}
\begin{align*}
 \Omega(\meas\ d\triangleright P\ \{0\to M\ |\ 1\to N\})&:= \Omega(P)\sqcup \bigl([d\mapsto 0] \sqcup \Omega(M) \uplus [d\mapsto 1]\sqcup \Omega(N)\bigr)
 \\
 \Omega(\qcase\ P\ \{0\to M\ |\ 1\to N\}) &:= \Omega(P)\sqcup \Omega(M)\sqcup \Omega(N)
\end{align*}
Therefore, elements of $\Omega(M)$ are memory functions that provide exactly enough measurement outcomes to fully determine the term to which $M$ will evaluate.
In particular, we have $\Supp(\nu)\subseteq \text{Dev}(M)$ for all $\nu\in\Omega(M)$.
We further define the set of \emph{extended valuations} of $M$:
\[
 \Omega_+(M):= \left\{\sigma\in\mathcal{M}\ \middle|\
 \exists \nu\in\Omega(M):\nu\preceq \sigma
 \right\}.
\]
Then, given a well-formed configuration $\vec{M}=\sum_i \alpha_i\cdot M_i^{\sigma_i}\in\vec{\mathscr{E}}$, the set of valuations of $\vec{M}$ is defined as
\[
\Omega(\vec{M}) :=\bigcap_i \sigma_i \sqcup \Omega(M_i)
\]
Note that $\sigma_i$ and elements of $\Omega(M_i)$ have disjoint supports, guaranteeing that $\Omega(\vec{M})$ is well defined.

\subsection{Interpretation of Configurations}

The goal of this section is to define an interpretation in $\CPM$ of configurations of $\vec{\mathscr{E}}$, which will be written as $\sem{\cdot}$.
The first step is to define an interpretation in $\FHilb$ of terms $M\in\mathscr{E}$, which is parametrized by a compatible valuation $\nu\in \Omega_+(M)$. We will write this interpretation as $[\cdot]_\nu$. On types and typing contexts, $[\cdot]_\nu$ coincides with $[\cdot]$. For a typing judgment $\Delta\vdash M:A$ with $M\in\mathscr{E}$ and a valuation $\nu\in\Omega_+(M)$, we inductively define the interpretation $[\Delta\vdash M:A]_\nu$. Compared to $[\cdot]$, the main addition is the rule for measurement, where the valuation is used to determine measurement outcomes:
\[
\scalebox{0.9}{
 \begin{prooftree}
 \hypo{[\Delta \vdash P:\q]_{\nu}\ = \ [\Delta] \xrightarrow{p} [\q]}
  \hypo{[\Delta' \vdash M_i:A]_{\nu}\ = \ [\Delta']\xrightarrow{f_i}[A]}
  \hypo{\nu(d)=i}
  \infer3{[\Delta,\Delta'\vdash \meas\ d\triangleright P\ \{0\rightarrow M_0\ |\ 1\rightarrow M_1\}: A]_\nu \ = \ [\Delta]\otimes[\Delta']\xrightarrow{p\otimes 1_{\sem{\Delta'}}} [\q]\otimes [\Delta'] \xrightarrow{\mathfrak{m}_i(f_i)}[A]}
 \end{prooftree}}
\]
where $\mathfrak{m}_i(f):= \bra{i}\otimes f$.
For terms without measurement, the definition is similar to that of $[\cdot]$. \version{The full definition is given in Appendix~\ref{app:soundness}.}{}

\begin{restatable}{lemma}{lemnuwelldefined}
\label{lem:nu_well_defined}
 For all valid typing judgment $\Delta \vdash M:A$ where $M\in\mathscr{E}$ and valuation $\nu\in \Omega_+(M)$, $[\Delta\vdash M:A]_\nu$ is well defined.
\end{restatable}

We extend the definition to configurations as follows. Let $\vec{M} = \sum_i \alpha_i\cdot M_i^{\sigma_i}\in\vec{\mathscr{E}}$ be a well-formed configuration for which there exist $\Delta,A$ such that $\Delta \vdash M_i: A$ is valid for all $i$. Then, we define
\[
 [\Delta;\vec{M};A]_{\nu}\ :=\ \sum_i \alpha_i[\Delta\vdash M_i:A]_{\nu}, \quad \text {for }\nu \in\Omega\bigl(\vec{M}\bigr).
\]
$[\cdot]_{\nu}$ is well defined—indeed, it is invariant with respect to congruence. 

We can now introduce the interpretation of a configuration restricted to valuations compatible with a fixed memory function. Given a  well-formed term $\vec{M}\in \vec{\mathscr{E}}$, let $\theta:\device\to\{0,1\}_\bot$ be a memory function such that $\Mem(\vec{M}) \ ||\ \theta$. We define:
\[
 \sem{\Delta;\vec{M};A}_{\theta} \ :=\ \sum_{ \nu\in\Omega(\vec{M}):\ \nu \ || \ \theta} [\Delta;\vec{M};A]_\nu(\cdot)[\Delta;\vec{M};A]_\nu^\dag
\]
 In the particular case where $\theta= [\bot]$, we omit the subscript and simply write $\sem{\Delta;\vec{M};A}$. Then, $\sem{\cdot}$ gives a semantics of the extended language in $\CPM$.
Generally, it does not yield morphisms of $\Caus[\CPM]$.
Nonetheless, we have the following property: the interpretation restricted to terms of $\mathscr{E}$ coincides with the original definition on $\mathscr{T}$, which justifies our using the same notation $\sem{\cdot}$:
\begin{restatable}{lemma}{lemcoincide}
\label{lem:coincide}
 Let $M\in \mathscr{T}$ be a well-typed term. For appropriate $\Delta$ and $A$, we have:
 \[
  \sem{\Delta \vdash M:A}=\sem{\Delta;\alift{M};A}
 \]
\end{restatable}

\subsection{Main Result}

We now prove a soundness theorem stating that if one configuration can be reduced to another, then their denotations are the same.

\begin{restatable}[Soundness]{theorem}{thmsoundness}
 \label{thm:soundness}
Let $\vec{M}=\sum_i\alpha_i\cdot M_i^{\sigma_i}$ be a configuration for which $\Delta \vdash M_i :A$ is valid for all $i$, for some $\Delta$ and $A$. Consider a reduction $\vec{M}\twoheadrightarrow \vec{M'}$. Then:
\[
 \sem{\Delta;\vec{M};A}_{\Mem(\vec{M})}=\sem{\Delta;\vec{M'};A}
\]
\end{restatable}

A direct consequence is that
if a configuration is reachable from a well-typed term of $\mathscr{T}$ (via $\alift{\cdot}$ and $\twoheadrightarrow$), then it must be physically meaningful (see Remark~\ref{rem:unitarity}). This implies by-design that no unitarity checks are needed to ensure physically in the operational semantics.

\section{Expressivity}

We study the expressivity of the language at first order and second order, i.e., at the level of quantum channels and of quantum supermaps, respectively.

We begin with first order. Recall that every quantum channel is a completely positive trace preserving map. We show that each quantum channel over qubits can be expressed by a term of the language:
\begin{restatable}[Universality for quantum channels]{proposition}{propchannels}
\label{prop:channels}
Let $\chan\in \CPM(\mathbb{C}^{2^n},\mathbb{C}^{2^m})$ be a quantum channel from $n$ qubits to $m$ qubits. Then there exists a term $M$ such that $\vdash M : \q^n \lolli \q^m$ is derivable and $\sem{M} = (1_{\sem{\q^n}}\otimes \chan)\circ \eta'_{[\q^n]}$.
\end{restatable}

The result on quantum channels is expected. Our main expressivity result concerns the level above, i.e. quantum supermaps.
We show that
the language implements a subclass of so-called quantum circuits with quantum control (\textbf{QC-QC}s), which we will refer to as \textbf{QC-QC}s \emph{with memory}.
Based on~\citet{grenoble_process}'s definition of \textbf{QC-QC}s, we give the following definition.
An $N$-linear map from $\CPM(\mathcal{H}_I,\mathcal{H}_O)^N$ to $\CPM(\mathcal{H}_P,\mathcal{H}_F)$ (where $\mathcal{H}_I,\mathcal{H}_O,\mathcal{H}_P,\mathcal{H}_F$ are finite-dimensional Hilbert spaces) is said to be a \textbf{QC-QC} with memory if it is computed by a circuit of the following shape:
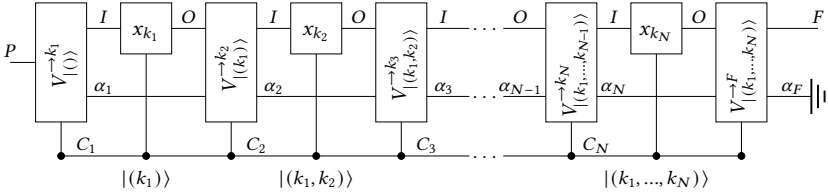
\[\scalebox{0.9}{\begin{tikzpicture}
	\begin{pgfonlayer}{nodelayer}
		\node [style=none] (3) at (-0.25, 1.75) {};
		\node [style=none] (4) at (1.25, 1.75) {};
		\node [style=none] (5) at (2.25, 1.75) {};
		\node [style=none] (6) at (3.75, 1.75) {};
		\node [style=none] (7) at (6.25, -0.25) {};
		\node [style=none] (8) at (2.25, -0.25) {};
		\node [style=none] (16) at (-3.25, -2) {};
		\node [style=none] (17) at (8.75, -2) {};
		\node [style=none] (18) at (4.75, 1.75) {};
		\node [style=none] (19) at (6.25, 1.75) {};
		\node [style=none] (20) at (-1.25, 1.75) {};
		\node [style=none] (21) at (-2.75, 1.75) {};
		\node [style=none] (22) at (7.25, 1.75) {};
		\node [style=none] (23) at (8.75, 1.75) {};
		\node [style=none] (24) at (7.25, -0.25) {};
		\node [style=none] (25) at (8.75, -0.25) {};
		\node [style=none] (26) at (-3.75, 0.75) {};
		\node [style=none] (27) at (-4.75, 0.75) {};
		\node [style=none] (28) at (11.25, 1.75) {};
		\node [style=none] (29) at (11.25, -0.25) {};
		\node [style=none] (30) at (9.75, -0.25) {};
		\node [style=none] (31) at (9.75, 1.75) {};
		\node [style=none] (32) at (12.25, 1.75) {};
		\node [style=none] (33) at (13.75, 1.75) {};
		\node [style=none] (34) at (12.25, -0.25) {};
		\node [style=none] (35) at (16.75, -0.25) {};
		\node [style=none] (36) at (14.75, 1.75) {};
		\node [style=none] (37) at (16.25, 1.75) {};
		\node [style=none] (38) at (1.25, -0.25) {};
		\node [style=none] (39) at (-2.75, -0.25) {};
		\node [style=none] (40) at (-3.25, -0.5) {};
		\node [style=none] (41) at (1.75, -0.5) {};
		\node [style=none] (42) at (1.75, -2) {};
		\node [style=none] (43) at (6.75, -0.5) {};
		\node [style=none] (44) at (6.75, -2) {};
		\node [style=none] (45) at (11.75, -0.5) {};
		\node [style=none] (46) at (11.75, -2) {};
		\node [style=none] (47) at (9.75, -2) {};
		\node [style=none] (48) at (16.75, -2) {};
		\node [style=texthere] (49) at (9.25, 1.675) {$\cdots$};
		\node [style=texthere] (50) at (9.25, -0.325) {$\cdots$};
		\node [style=texthere] (51) at (9.25, -2.075) {$\cdots$};
		\node [style=none] (52) at (17.25, 1.75) {};
		\node [style=none] (53) at (19, 1.75) {};
		\node [style=dot] (54) at (16.75, -2) {};
		\node [style=dot] (55) at (11.75, -2) {};
		\node [style=dot] (56) at (-3.25, -2) {};
		\node [style=dot] (57) at (1.75, -2) {};
		\node [style=dot] (58) at (6.75, -2) {};
		\node [style=dot] (59) at (-0.75, -2) {};
		\node [style=dot] (60) at (4.25, -2) {};
		\node [style=dot] (61) at (14.25, -2) {};
		\node [style=none] (62) at (14.25, 1.25) {};
		\node [style=none] (63) at (-0.75, 1.25) {};
		\node [style=none] (64) at (4.25, 1.25) {};
		\node [style=texthere] (65) at (-0.75, -2.75) {$|(k_1)\rangle$};
		\node [style=texthere] (66) at (4.25, -2.75) {$|(k_1,k_2)\rangle$};
		\node [style=texthere] (67) at (14.25, -2.75) {$|(k_1,...,k_N)\rangle$};
		\node [style=upground] (68) at (19, -0.25) {};
		\node [style=none] (69) at (18.825, -0.25) {};
		\node [style=big box] (74) at (-3.25, 0.75) {\rotatebox[origin=c]{90}{$V_{|()\rangle}^{\to k_1}$}};
		\node [style=small box] (75) at (-0.75, 1.75) {$x_{k_1}$};
		\node [style=small box] (76) at (4.25, 1.75) {$x_{k_2}$};
		\node [style=small box] (77) at (14.25, 1.75) {$x_{k_N}$};
		\node [style=texthere] (78) at (-4.75, 1) {$P$};
		\node [style=texthere] (79) at (-2, 2) {$I$};
		\node [style=texthere] (80) at (0.5, 2) {$O$};
		\node [style=texthere] (81) at (-2, 0) {$\alpha_1$};
		\node [style=texthere] (82) at (3, 2) {$I$};
		\node [style=texthere] (83) at (3, 0) {$\alpha_2$};
		\node [style=texthere] (85) at (5.5, 2) {$O$};
		\node [style=texthere] (86) at (8, 2) {$I$};
		\node [style=texthere] (87) at (8, 0) {$\alpha_3$};
		\node [style=texthere] (88) at (10.25, 2) {$O$};
		\node [style=texthere] (89) at (10.25, 0) {$\alpha_{N-1}$};
		\node [style=texthere] (90) at (13, 2) {$I$};
		\node [style=texthere] (91) at (13, 0) {$\alpha_N$};
		\node [style=texthere] (92) at (15.5, 2) {$O$};
		\node [style=texthere] (93) at (19, 2) {$F$};
		\node [style=texthere] (94) at (18.25, 0) {$\alpha_F$};
		\node [style=texthere] (95) at (-2.5, -1.75) {$C_1$};
		\node [style=texthere] (96) at (2.5, -1.75) {$C_2$};
		\node [style=texthere] (97) at (7.5, -1.75) {$C_3$};
		\node [style=texthere] (98) at (12.5, -1.75) {$C_N$};
		\node [style=big box] (99) at (6.75, 0.75) {\rotatebox[origin=c]{90}{$V_{|(k_1,k_2)\rangle}^{\to k_3}$}};
		\node [style=big box] (100) at (1.75, 0.75) {\rotatebox[origin=c]{90}{$V_{|(k_1)\rangle}^{\to k_2}$}};
		\node [style=big box] (101) at (11.75, 0.75) {\rotatebox[origin=c]{90}{$V_{\!|(k_1,...,k_{N\!-\!1})\rangle}^{\to k_N}$}};
		\node [style=big box] (102) at (16.75, 0.75) {\rotatebox[origin=c]{90}{$V_{|(k_1,...,k_N)\rangle}^{\to F}$}};
	\end{pgfonlayer}
	\begin{pgfonlayer}{edgelayer}
		\draw (3.center) to (4.center);
		\draw (5.center) to (6.center);
		\draw (8.center) to (7.center);
		\draw (16.center) to (17.center);
		\draw (18.center) to (19.center);
		\draw (22.center) to (23.center);
		\draw (24.center) to (25.center);
		\draw (21.center) to (20.center);
		\draw (27.center) to (26.center);
		\draw (31.center) to (28.center);
		\draw (30.center) to (29.center);
		\draw (34.center) to (35.center);
		\draw (32.center) to (33.center);
		\draw (36.center) to (37.center);
		\draw (39.center) to (38.center);
		\draw (43.center) to (44.center);
		\draw (47.center) to (46.center);
		\draw (46.center) to (45.center);
		\draw (46.center) to (48.center);
		\draw (48.center) to (35.center);
		\draw (40.center) to (16.center);
		\draw (41.center) to (42.center);
		\draw (52.center) to (53.center);
		\draw (63.center) to (59);
		\draw (64.center) to (60);
		\draw (62.center) to (61);
		\draw (35.center) to (69.center);
	\end{pgfonlayer}
\end{tikzpicture}}\]
For readability, we label a wire for Hilbert $\mathcal{H}_X$ as $X$. $\mathcal{H}_{C_n}$ is defined as the finite-dimensional Hilbert space whose computational basis consists of the states $\ket{(k_1,...,k_n)}$ where $k_1,...,k_n$ are pairwise distinct elements of $\{ 1,\ldots, N\}$.
Each of the $N$ open slots,  which are given by variables $x_{1},...,x_N$, is for one input channel.
In this model, every open slot is controlled by a memory register (the bottom wire) that decides which input channel is plugged to which open slot. For instance, if the state of the control register of the left-most controlled open slot is $\ket{(1)}$ then the first input channel $x_1$ is applied, if the state is $\ket{(2)}$ then the second channel $x_2$ is applied, and if the state is $\frac{\ket{(1)}+\ket{(2)}}{\sqrt 2}$, then  $x_1$ and $x_2$ are applied in superposition.

The bottom register also plays the role of a memory: the $n^\text{th}$ controlled open slot is controlled by a state of the form $\ket{(k_1,...,k_n)}$, such that $x_{k_1}, \ldots, x_{k_{n-1}}$ correspond to the input channels applied so far in this branch of the superposition, and $x_{k_n}$ is the input channel which is plugged in the $n^\text{th}$ open slot. Such a memory mechanism guarantees the linearity on the input channels: at the end of the circuit, for a given basis state of the memory register, each input channel has been applied once and only once. Indefinite causal order can thus be implemented when the state of memory register is in superposition.
Each of these controlled open slot is intertwined with isometries $V_{\ket{(k_1,...,k_n)}}^{\to k_{n+1}}$ that initialise and update the memory register in an appropriate way, depending on the higher order quantum operation one wants to implement.\footnote{Additionally, the output spaces of the $V_{\ket{(k_1,...,k_N)}}^{\to F}$ are assumed to be orthogonal.}

 \version{A more detailed presentation of \textbf{QC-QC}s with memory is given in Appendix~\ref{app:QCQC}.}{}

We show that the language is universal for \textbf{QC-QC}s with memory:
\begin{restatable}[Universality for \textbf{QC-QC}s with memory]{proposition}{propqcqc}
\label{prop:qcqc}
Let $\smap$ be a \textup{\textbf{QC-QC}} with memory whose input consists of $N$ channels in $\CPM(\mathbb{C}^{2^n},\mathbb{C}^{2^m})$ and whose output is a channel in $\CPM(\mathbb{C}^{2^k},\mathbb{C}^{2^l})$. Then there exists a term $M$ of type $(\q^n \lolli \q^m)^{\otimes N}\lolli \q^k \lolli \q^l$ whose interpretation coincides with $\smap$.
\end{restatable}
The constructive, bottom-up approach of the definition directly gives us the implementation of a specific \textbf{QC-QC}s with memory.
Therefore after defining a generalization of the $\qcase$ construct to $n$ branching statements, it is straightforward to construct a term that implements any \textbf{QC-QC} with memory. \version{The construction is detailed in Appendix~\ref{app:expressivity}.}{}

\begin{remark}
\textbf{QC-QC}s with memory are special instances of general \textbf{QC-QC}s. In the general case the bottom register stores which input channels have been applied so far but not in which order (the basis states of this register are then sets rather than lists). We leave for further development the implementation
of general \textbf{QC-QC}s.
\end{remark}

\section{An Extension for Non-Linearity and Recursion}
\label{section:extension}

We have defined a linear programming language which supports both measurements and processes with indefinite causal orders.
Although linearity is the key feature that enables us combine these two features, requiring all programs to be strictly linear is a considerable restriction. Indeed, we cannot implement any quantum algorithm whose input is an oracle that is queried more than once.
To address this limitation, we define an extension of the language that includes some nonlinear features, and give its operational semantics. Notably, this extension has a $\letrec$ construct for recursive processes, because multiple calls to a function are now allowed.

\subsection{Syntax and Type System}

The extension of the syntax is fairly straightforward: we add nonlinear versions of the terms of the lambda calculus, and terms for recursion.
In addition to the set $\mathcal{X}$ of linear variables, we define a countably infinite set $\mathcal{V}$ of \emph{duplicable variables}, written as $u,v,w...$.
The set of terms $\mathscr{T}^*\supseteq \mathscr{T}$ is defined by the grammar in Figure~\ref{fig:syntax_gen}.
\begin{figure}[!t]
\hrulefill
\begin{align*}
 \begin{array}{lrl}
 \
 \text{(Extended terms)}&\quad \mathscr{T}^* \ni M,N,P \  ::=  &  x \ |\ M N \ |\ \lambda x.M \ |\ u\ |\ M^*N \ | \ \lambda^* u. M \\
  && \!\!\!\! |\ \ \pair{M,N} \ |\ \mathtt{let}\ \pair{x,y} = M \  \mathtt{in}\ N \ |\ \unit\ |\ M;N \\
  && \!\!\!\! |\ \ U\ | \ \ket{0}\ |\ \ket{1}
  \\
  && \!\!\!\! |\ \ \meas \ P\ \{0\rightarrow M\ |\ 1\rightarrow N\} \\
  && \!\!\!\! |\ \ \qcase \ P\ \{0\rightarrow M\ |\ 1\rightarrow N\} \\
 && \!\!\!\! |\ \ \letrec \ f\ x=M \ |\ \letrec^* \ f\ u=M
 \end{array}
\end{align*}
\hrulefill
\caption{Extended syntax}
\label{fig:syntax_gen}
\end{figure}
The terms $u,M^*N,\lambda^*u.M$ denote nonlinear variables, application and $\lambda$-abstraction.
Whereas the term $M$ in $\lambda x. M$ must contain \emph{exactly} one occurrence of $x$ (counting occurrences in different branches of $\meas$ or $\qcase$ as the same), the term $M$ in $\lambda^* u. M$ must contain \emph{at least} one occurrence of $u$.
The term $\letrec\ f\ x=M$ recursively defines $f$ as a linear function of $x$; the term $\letrec^*\ f\ u=M$ recursively defines $f$ as a nonlinear function of $u$. In both cases, the function symbol $f$ itself is a duplicable variable.

We define a set of \emph{controllable terms} $\mathscr{C}^*\subseteq \mathscr{T}^*$ as the set of terms that do not contain duplicable variables or measurement.

The type system is extended to include a nonlinear function type $\Rightarrow$:
\[
 \begin{array}{lll}
 \text{(Extended types)}\quad A,B \ & ::= \ & \one\ |\ \q \ |\
  A \otimes B \ |\ A \lolli B \ |\ A \Rightarrow B
 \end{array}
\]
We now distinguish linear and unrestricted typing contexts.
A \emph{linear typing context} is a list $\Delta = x_1:A_1,...,x_n:A_n$ where $x_1,...,x_n$ are pairwise distinct linear variables and $A_1,...,A_n$ their respective types.
An \emph{unrestricted typing context} is a list $\Gamma = u_1:A_1,...,u_n:A_n$ where $u_1,...,u_n$ are pairwise distinct duplicable variables and $A_1,...,A_n$ their respective types.
Then, \emph{typing judgments} are expressions of the form
\[
 \Gamma; \Delta \vdash M :A
\]
where $M$ is a term, $A$ a type, $\Gamma$ an unrestricted typing context and $\Delta$ a linear typing context.
The typing rules are obtained by extending the linear type system in the expected way.\version{\footnote{See Appendix~\ref{app:extension_typing}.}}{}

\subsection{Operational Semantics}

The main difficulty of the operational semantics is keeping track of which measurement outcomes must be synchronized. In the linear fragment, this is achieved by assigning to each measurement a unique device reference $d$ and recording its outcome via a memory function.

However, the possible duplication of terms in the nonlinear extension invalidates the device reference approach.
To see this,  consider the execution term $\Measd{d}$ of Section~\ref{section:design}
On the one hand, by reducing the term
\begin{equation}
\label{eq:linear_duplication}
(\lambda y. \qcase \ H\ket{0} \ \{0\to y\ |\ 1\to y\})\  \Measd{d},
\end{equation}
we will obtain two copies of $ \Measd{d}$, whose measurement outcomes must be synchronized.
But on the other hand, the term
\[
 (\lambda^*u. \pair{u\ket{0},u\ket{1}})\ \Measd{d}
\]
should be read as creating two independent occurences of the process $\Measd{d}$, which are then applied to distinct inputs. This is essentially the same as creating a second measurement device; their measurement outcomes should therefore also be independent.

To resolve this issue, we modify the treatment of device references so that nonlinearity can be interpreted as duplicating the devices themselves. We define the set of \emph{extended device references} $\device^*$ as sequences $d\cdot n_1 \cdot ...\cdot n_k$ where $d\in\device$ is a device reference, $k\geq 0$ and $n_1,...,n_k\geq 1$. In particular, $\device\subseteq \device^*$. Then, an \emph{extended memory function} is a function from $\device^*$ to $\{0,1\}_\bot$ such that only a finite number of elements of $\device$ are mapped to $\{0,1\}$.

In order to define the operational semantics, we define execution terms and terms superpositions similarly to the linear language, replacing device references and memory functions with their extended counterparts. \version{A more detailed explanation is provided in Appendix~\ref{app:extension_op}.}{} For all pair of execution terms $M,N$ and $u\in\mathcal{V}$, $M\db{N/u}$ is defined as the execution term where the $n$-th occurrence of $u$ in $M$ is replaced by $N$ in which every extended device reference $d^*$ is replaced by $d^* \cdot n$. For example:
\begin{align*}
&\pair{\lambda x. M\,u,u}\db{\meas \ d^*\triangleright H\ket{0}\  \{0\to \ket{0} \ |\ 1\to \ket{1}\}/u}\\
&\quad := \pair{\lambda x. M\,\meas \ (d^*\cdot 1) \triangleright H\ket{0}\ \{0\to \ket{0} \ |\ 1\to \ket{1}\},\meas \ (d^*\cdot 2)\triangleright H\ket{0}\ \{0\to \ket{0} \ |\ 1\to \ket{1}\}}\
\end{align*}

Therefore, we can extend the reduction relation $\twoheadrightarrow$, where the main novelty is the rules for nonlinear abstraction and recursion\version{ (the full definition is given in Appendix~\ref{app:extension_op})}{}:
\[
\scalebox{1}{
 \begin{prooftree}
  \infer0{(\lambda^* u. M)^*V\twoheadrightarrow M\db{V/u}}
 \end{prooftree}}
\qquad
\scalebox{1}{
\begin{prooftree}
 \infer0{(\letrec\ f\ x=M )V \twoheadrightarrow M\{V/x\}\db{\letrec \ f\ x=M/f}}
\end{prooftree}}
\]
\[
\scalebox{1}{
\begin{prooftree}
 \infer0{(\letrec^*\ f\ u=M)^*V \twoheadrightarrow M\db{V/u}\db{\letrec^* \ f\ u=M/f}}
\end{prooftree}}
\]

Lastly, we can check that the standard properties of the typing and operational semantics (substitution—Lemma~\ref{lem:substitution}, subject reduction—Proposition~\ref{prop:subject_reduction}, progress—Lemma~\ref{lem:progress} and uniqueness of the normal form—Proposition~\ref{prop:confluence}) still hold.

\subsection{Examples}

\begin{example}[No-cloning]

Due to the linear evaluation strategy, the term $\lambda^*u.\pair{u,u}$ acts by copying a qubit along the computational basis, not cloning it. To illustrate this, suppose the term is applied to $U\ket{0}$, which we assume reduces to $\alpha\ket{0}+\beta\ket{1}$:
\begin{align*}
 (\lambda^*u.\pair{u,u} )(U\ket{0}) & \twoheadrightarrow \alpha \cdot (\lambda^*u.\pair{u,u} )\ket{0} +\beta\cdot (\lambda^*u.\pair{u,u} )\ket{1} \\
 & \twoheadrightarrow \alpha \cdot \ket{0,0} +\beta\cdot (\lambda^*u.\pair{u,u} )\ket{1}  \twoheadrightarrow \alpha\cdot \ket{0,0} +\beta \cdot\ket{1,1}
\end{align*}
In contrast, cloning would have given the following result:
\begin{align*}
 \pair{U\ket{0},U\ket{0}}&\twoheadrightarrow \alpha\cdot \pair{U\ket{0},\ket{0}}+\beta\cdot \pair{U\ket{0},\ket{1}} \twoheadrightarrow \alpha^2\cdot \ket{0,0}+\alpha\beta\cdot \ket{1,0}+\beta\alpha \cdot\ket{0,1} + \beta^2\cdot\ket{1,1}
\end{align*}

\end{example}

\begin{example}[Repeat-Until-Success]
\label{ex:rus}
In~\cite{repeat_until_success}, the authors define ``Repeat-Until-Success'' circuits, which are non-deterministic circuits that minimize the $T$ count required to approximate single-qubit unitary gates.
Repeat-Until-Success circuits can be implemented in our language by recursive terms of the following form:
\begin{align*}
 \letrec \ f \ z = \ & \mathtt{let}\ \pair{x,y}= \mathtt{U}\ \pair{\ket{0},z} \ \mathtt{in}
 \ \meas \ x \ \{0\to \pair{y,f} \ |\ 1\ \to f\,y\}
\end{align*}
where $\mathtt{U}$ is a term of type $\q\otimes\q \lolli \q\otimes \q$ implementing a 2-qubit unitary gate. The program returns a pair consisting of the output qubit together with a copy of the recursive function $f$.

\end{example}

\begin{acks}
This work is supported by the the \emph{Plan France 2030} through the PEPR integrated project EPiQ ANR-22-PETQ-0007 and the HQI platform ANR-22-PNCQ-0002; and by the European projects Quantum Flagship NEASQC, European  MSCA Staff Exchanges Qcomical HORIZON-MSCA-2023-SE-01. The project is also supported by the \emph{Maison du Quantique} MaQuEst. 
\end{acks}

\printbibliography

\version{}{\end{document}}

\begin{appendices}

\section{Additional Details on Examples}
\label{app:examples}

\subsection{Typing Derivations for Section~\ref{section:syntax_examples}}

Here, we give the typing derivations for the examples of Section~\ref{section:syntax_examples}.

$\mathtt{CNOT}$ is of type $\q\lolli\q \lolli \q\otimes \q$:
\[
\scalebox{0.8}{
 \begin{prooftree}
 \infer0{c:\q\vdash c:\q}
 \infer0{\vdash I:\q\lolli \q}
 \infer0{\vdash X:\q\lolli \q}
 \hypo{I,X\in\mathscr{C}}
 \infer4{c:\q\vdash \qcase \ c \ \{0\rightarrow I \ |\ 1\rightarrow X\}: \q\lolli \q\otimes \q}
 \infer1{\vdash  \lambda c.\qcase \ c \ \{0\rightarrow I \ |\ 1\rightarrow X\}: \q\lolli \q\lolli \q\otimes \q}
 \end{prooftree}}
\]

For composition, we have the following typing derivation for any types $A$, $B$, $C$, where exchange rules are left implicit:
\[
\scalebox{0.8}{
 \begin{prooftree}
 \infer0{z:(A\lolli B)\otimes (B\lolli C)\vdash z:(A\lolli B)\otimes (B\lolli C)}
 \infer0{g:B\lolli C\vdash g: B\lolli C}
 \infer0{f:A\lolli B\vdash f: A\lolli B}
 \infer0{t:A\vdash t: A}
 \infer2{f:A\lolli B,t:A\vdash ft: B}
 \infer2{f:A\lolli B, g:B\lolli C,t:A\vdash g(ft): C}
 \infer1{f:A\lolli B, g:B\lolli C\vdash \lambda t. g(ft): A\lolli C}
 \infer2{z:(A\lolli B)\otimes (B\lolli C)\vdash \mathtt{let} \ \pair{f, g} =z \ \mathtt{in} \ \lambda t. g(ft): A\lolli C}
\infer1{\vdash \lambda z. \mathtt{let} \ \pair{f,g} =z \ \mathtt{in} \ \lambda t. g(f t):(A\lolli B)\otimes (B\lolli C)\lolli A\lolli C}
 \end{prooftree}}
\]

The following derivation shows that $\mathtt{switch}$ is of type $(\q^n\lolli \q^n)\otimes (\q^n\lolli \q^n )\lolli \q\lolli \q^n \lolli \q^{n+1}$, for $n\geq 1$. Again, exchange rules are left implicit.
\[
\scalebox{0.6}{
 \begin{prooftree}
 \infer0{z :  (\q^n\lolli \q^n)\otimes (\q^n\lolli \q^n )\vdash z:  (\q^n\lolli \q^n)\otimes (\q^n\lolli \q^n )}
  \infer0{q:\q \vdash q:\q}
 \hypo{(*)}
 \infer1{x:\q^n\lolli \q^n , y: \q^n\lolli \q^n \vdash \mathtt{comp} \ \pair{x,y}: \q^n\lolli \q^n}
 \hypo{(*')}
 \infer1{x:\q^n\lolli \q^n , y: \q^n\lolli \q^n \vdash \mathtt{comp} \ \pair{y,x}:\q^n\lolli \q^n}
 \infer3{x:\q^n\lolli \q^n , y: \q^n\lolli \q^n, q: \q \vdash \qcase\ q\ \{0\rightarrow \mathtt{comp} \ \pair{x,y} \ |\ 1\rightarrow \mathtt{comp} \ \pair{y,x} \}: \q^n \lolli \q^{n+1}}
 \infer1{x:\q^n\lolli \q^n , y: \q^n\lolli \q^n \vdash \lambda q. \qcase\ q\ \{0\rightarrow \mathtt{comp} \ \pair{x,y} \ |\ 1\rightarrow \mathtt{comp} \ \pair{y,x} \}: \q\lolli \q^n \lolli \q^{n+1}}
 \infer2{z :  (\q^n\lolli \q^n)\otimes (\q^n\lolli \q^n )\vdash \mathtt{let}\ \pair{x,y}=z \ \mathtt{in}\ \lambda q. \qcase\ q\ \{0\rightarrow \mathtt{comp} \ \pair{x,y} \ |\ 1\rightarrow \mathtt{comp} \ \pair{y,x} \}: \q\lolli \q^n \lolli \q^{n+1}}
  \infer1{\vdash \lambda z. \mathtt{let}\ \pair{x,y}=z \ \mathtt{in}\ \lambda q. \qcase\ q\ \{0\rightarrow \mathtt{comp} \ \pair{x,y} \ |\ 1\rightarrow \mathtt{comp} \ \pair{y,x} \}: (\q^n\lolli \q^n)\otimes (\q^n\lolli \q^n )\lolli \q\lolli \q^n \lolli \q^{n+1}}
 \end{prooftree}}
\]
The derivations $(*)$ and $(*')$ are symmetrical. We give the derivation $(*)$:
\[
 \scalebox{0.8}{
 \begin{prooftree}
  \infer0{\vdash \mathtt{comp}: (\q^n\lolli \q^n)\otimes (\q^n\lolli \q^n) \lolli \q^n\lolli \q^n}
  \infer0{x:\q^n\lolli \q^n \vdash x:\q^n\lolli \q^n}
  \infer0{y: \q^n\lolli \q^n \vdash y:\q^n\lolli \q^n}
 \infer2{x:\q^n\lolli \q^n , y: \q^n\lolli \q^n \vdash \pair{x,y}:(\q^n\lolli \q^n)\otimes (\q^n\lolli \q^n)}
 \infer2{x:\q^n\lolli \q^n , y: \q^n\lolli \q^n \vdash \mathtt{comp} \ \pair{x,y}: \q^n\lolli \q^n}
 \end{prooftree}
}
\]
where the typing of composition is given by the previous example.

\subsection{Detailed Reduction of Example~\ref{ex:switch_reduction}}
\label{app:example_switch}

With $\Measd{d}$ and $\Measd{e}$ defined as in Example~\ref{ex:switch_reduction}, we have the following sequence of reductions:
\begin{align*}
 &\mathtt{switch} \ \pair{\Measd{d},\Measd{e}} \ (H\ket{0})\\
 &\quad = \ \bigl(\lambda z. \mathtt{let}\ \pair{x,y}=z \ \mathtt{in}\ \lambda q. \qcase\ q\ \{0\rightarrow \mathtt{comp} \ \pair{x,y} \ |\ 1\rightarrow \mathtt{comp} \ \pair{y,x} \}\bigr)\ \pair{\Measd{d},\Measd{e}} \ (H\ket{0})
 \\
 &\quad \twoheadrightarrow \ \sfrac{1}{\sqrt{2}}\cdot  \bigl(\lambda z. \mathtt{let}\ \pair{x,y}=z \ \mathtt{in}\ \lambda q. \qcase\ q\ \{0\rightarrow \mathtt{comp} \ \pair{x,y} \ |\ 1\rightarrow \mathtt{comp} \ \pair{y,x} \}\bigr)\ \pair{\Measd{d},\Measd{e}} \ \ket{0}
 \\
 &\quad \quad\ \ + \sfrac{1}{\sqrt{2}} \cdot \bigl(\lambda z. \mathtt{let}\ \pair{x,y}=z \ \mathtt{in}\ \lambda q. \qcase\ q\ \{0\rightarrow \mathtt{comp} \ \pair{x,y} \ |\ 1\rightarrow \mathtt{comp} \ \pair{y,x} \}\bigr)\ \pair{\Measd{d},\Measd{e}} \ \ket{1}
 \\
 & \quad \twoheadrightarrow^2 \ \sfrac{1}{\sqrt{2}}\cdot  \bigl(\mathtt{let}\ \pair{x,y}=\pair{\Measd{d},\Measd{e}} \ \mathtt{in}\ \lambda q. \qcase\ q\ \{0\rightarrow \mathtt{comp} \ \pair{x,y} \ |\ 1\rightarrow \mathtt{comp} \ \pair{y,x} \}\bigr) \ \ket{0}
 \\
 &\quad \quad \ \ + \sfrac{1}{\sqrt{2}} \cdot \bigl(\mathtt{let}\ \pair{x,y}=\pair{\Measd{d},\Measd{e}} \ \mathtt{in}\ \lambda q. \qcase\ q\ \{0\rightarrow \mathtt{comp} \ \pair{x,y} \ |\ 1\rightarrow \mathtt{comp} \ \pair{y,x} \}\bigr)\ \ket{1}
 \\
 &\quad \twoheadrightarrow^2 \ \sfrac{1}{\sqrt{2}}\cdot  \bigl(\lambda q. \qcase\ q\ \{0\rightarrow \mathtt{comp} \ \pair{\Measd{d},\Measd{e}} \ |\ 1\rightarrow \mathtt{comp} \ \pair{\Measd{e},\Measd{d}} \}\bigr) \ \ket{0}
 \\
 &\quad \quad \ \ + \sfrac{1}{\sqrt{2}} \cdot \bigl(\lambda q. \qcase\ q\ \{0\rightarrow \mathtt{comp} \ \pair{\Measd{d},\Measd{e}} \ |\ 1\rightarrow \mathtt{comp} \ \pair{\Measd{e},\Measd{d}} \}\bigr)\ \ket{1}
 \\
 &\quad \twoheadrightarrow^2 \ \sfrac{1}{\sqrt{2}}\cdot   \qcase\ \ket{0}\ \{0\rightarrow \mathtt{comp} \ \pair{\Measd{d},\Measd{e}} \ |\ 1\rightarrow \mathtt{comp} \ \pair{\Measd{e},\Measd{d}} \}
 \\
 &\quad \quad\ \ + \sfrac{1}{\sqrt{2}} \cdot \qcase\ \ket{1}\ \{0\rightarrow \mathtt{comp} \ \pair{\Measd{d},\Measd{e}} \ |\ 1\rightarrow \mathtt{comp} \pair{\Measd{e},\Measd{d}} \}
 \\
 &\quad \twoheadrightarrow^2 \ \sfrac{1}{\sqrt{2}}\cdot   \lambda s. \pair{\ket{0}, \mathtt{comp} \ \pair{\Measd{d},\Measd{e}}\ s}
 +\ \sfrac{1}{\sqrt{2}} \cdot\lambda s. \pair{\ket{1}, \mathtt{comp} \ \pair{\Measd{e},\Measd{d}}\ s}
\end{align*}
Therefore, if we take the target qubit to be initially in state $\ket{0}$, we obtain the sequence of reductions:
\begin{align*}
 &\mathtt{switch} \ \pair{\Measd{d},\Measd{e}} \ (H\ket{0}) \ \ket{0}
\\
 &\quad \twoheadrightarrow^* \ \sfrac{1}{\sqrt{2}}\cdot   (\lambda s. \pair{\ket{0}, \mathtt{comp} \ \pair{\Measd{d},\Measd{e}}\ s})\,\ket{0}
 +\ \sfrac{1}{\sqrt{2}} \cdot(\lambda s. \pair{\ket{1}, \mathtt{comp} \ \pair{\Measd{e},\Measd{d}}\ s})\,\ket{0}
 \\
 &\quad = \ \sfrac{1}{\sqrt{2}}\cdot   (\lambda s. \pair{\ket{0}, (\lambda z. \mathtt{let} \ \pair{x, y} =z \ \mathtt{in} \ \lambda t. y(xt)) \ \pair{\Measd{d},\Measd{e}}\ s})\,\ket{0}
 \\
 &\quad \quad \ \ +\ \sfrac{1}{\sqrt{2}} \cdot(\lambda s.\pair{\ket{1}, (\lambda z. \mathtt{let} \ \pair{x, y} =z \ \mathtt{in} \ \lambda t. y(xt)) \ \pair{\Measd{e},\Measd{d}}\ s})\,\ket{0}
 \\
 &\quad \twoheadrightarrow^2 \ \sfrac{1}{\sqrt{2}}\cdot \pair{\ket{0}, (\lambda z. \mathtt{let} \ \pair{x, y} =z \ \mathtt{in} \ \lambda t. y(xt)) \ \pair{\Measd{d},\Measd{e}}\ \ket{0}}
 \\
 &\quad \quad \ \ +\ \sfrac{1}{\sqrt{2}} \cdot \pair{\ket{1}, (\lambda z. \mathtt{let} \ \pair{x, y} =z \ \mathtt{in} \ \lambda t. y(xt)) \ \pair{\Measd{e},\Measd{d}}\ \ket{0}}
 \\
 &\quad \twoheadrightarrow^2 \ \sfrac{1}{\sqrt{2}}\cdot   \pair{\ket{0}, (\mathtt{let} \ \pair{x, y} =\pair{\Measd{d},\Measd{e}} \ \mathtt{in} \ \lambda t. y(xt)) \ket{0}}
 \\
&\quad \quad \ \ +\ \sfrac{1}{\sqrt{2}} \cdot\pair{\ket{1}, (\mathtt{let} \ \pair{x, y} =\pair{\Measd{e},\Measd{d}} \ \mathtt{in} \ \lambda t. y(xt))\ket{0}}
 \\
&\quad \twoheadrightarrow^2 \ \sfrac{1}{\sqrt{2}}\cdot   \pair{\ket{0}, (\lambda t. \Measd{e}(\Measd{d}t))\ket{0}}
 +\ \sfrac{1}{\sqrt{2}} \cdot\pair{\ket{1}, (\lambda t. \Measd{d}(\Measd{e}t))\ket{0}}
 \\
 &\quad \twoheadrightarrow^2 \ \sfrac{1}{\sqrt{2}}\cdot   \pair{\ket{0}, \Measd{e}(\Measd{d}\ket{0})}
 +\ \sfrac{1}{\sqrt{2}} \cdot\pair{\ket{1}, \Measd{d}(\Measd{e}\ket{0})}
\end{align*}

\subsection{Detailed Reduction of Example~\ref{ex:meas_transfer}}
\label{app:meas_transfer}

We have the following sequence of reductions:
\begin{align*}
 \mathtt{N'} =& \left(\lambda w. \qcase\ H\ket{0} \ \{ 0 \rightarrow \lambda t. \pair{t,w\,\ket{0}} \ |\ 1\rightarrow \lambda t.\pair{t,w\,\ket{1}}\}\right)\Measd{d} \ \ket{k}
 \\
 \twoheadrightarrow & \ \qcase\ H\ket{0} \ \{ 0 \rightarrow \lambda t. \pair{t,\Measd{d}\,\ket{0}} \ |\ 1\rightarrow \lambda t.\pair{t,\Measd{d}\,\ket{1}}\} \ \ket{k}
 \\
 \twoheadrightarrow &\ \frac{1}{\sqrt{2}}\cdot \qcase\ \ket{0} \ \{ 0 \rightarrow \lambda t. \pair{t,\Measd{d}\,\ket{0}} \ |\ 1\rightarrow \lambda t.\pair{t,\Measd{d}\,\ket{1}}\} \ \ket{k} \\
 & + \frac{1}{\sqrt{2}}\cdot \qcase\ \ket{1} \ \{ 0 \rightarrow \lambda t. \pair{t,\Measd{d}\,\ket{0}} \ |\ 1\rightarrow \lambda t.\pair{t,\Measd{d}\,\ket{1}}\} \ \ket{k}
 \\
 \twoheadrightarrow^2 &\ \frac{1}{\sqrt{2}}\cdot \lambda s. \pair{\ket{0},\lambda t. \pair{t,\Measd{d}\,\ket{0}} s} \ket{k}
 + \frac{1}{\sqrt{2}}\cdot \lambda s. \pair{\ket{1},\lambda t. \pair{t,\Measd{d}\,\ket{1}} s}\ket{k}
 \\
\twoheadrightarrow^2 &\ \frac{1}{\sqrt{2}}\cdot \pair{\ket{0},\lambda t. \pair{t,\Measd{d}\,\ket{0}} \ket{k}}
 + \frac{1}{\sqrt{2}}\cdot \pair{\ket{1},\lambda t. \pair{t,\Measd{d}\,\ket{1}}\ket{k}}
 \\
 \twoheadrightarrow^2 &\ \frac{1}{\sqrt{2}}\cdot \pair{\ket{0},\ket{k} ,\Measd{d}\,\ket{0}}
 + \frac{1}{\sqrt{2}}\cdot \pair{\ket{1},\ket{k},\Measd{d}\,\ket{1}}
 \ =:\ \mathtt{P}
\end{align*}
The term $P$ consists of a superposition with one copy of the measurement in each summand. If we apply rule ($m_0$) in the first summand, we obtain:
\begin{align*}
 \mathtt{P} \ =& \ \frac{1}{\sqrt{2}}\cdot \pair{\ket{0},\ket{k} ,\Measd{d}\,\ket{0}}
 + \frac{1}{\sqrt{2}}\cdot \pair{\ket{1},\ket{k},\Measd{d}\,\ket{1}}
 \\
 \twoheadrightarrow^* &\ \frac{1}{\sqrt{2}}\cdot \pair{\ket{0},\ket{k},\ket{0}}^{[d\mapsto 0]} + \frac{1}{\sqrt{2}}\cdot \pair{\ket{1},\ket{k},\Measd{d}\,\ket{1}}^{[\bot]}
 \\
 \twoheadrightarrow^* &\ \frac{1}{\sqrt{2}}\cdot \pair{\ket{0},\ket{k},\ket{0}}^{[d\mapsto 0]} + 0\cdot \pair{\ket{1},\ket{k},\ket{0}}^{[d\mapsto 0]}
\end{align*}
Conversely if we apply $(m_1)$, we obtain:
\begin{align*}
 P \ =& \ \frac{1}{\sqrt{2}}\cdot \pair{\ket{0},\ket{k} ,\Measd{d}\,\ket{0}}
 + \frac{1}{\sqrt{2}}\cdot \pair{\ket{1},\ket{k},\Measd{d}\,\ket{1}}
 \\
 \twoheadrightarrow^* &\  0 \cdot \pair{\ket{0},\ket{k},\ket{1}}^{[d\mapsto 1]} + \frac{1}{\sqrt{2}}\cdot \pair{\ket{1},\ket{k},\Measd{d}\,\ket{1}}^{[\bot]}
 \\
 \twoheadrightarrow^* & \ 0 \cdot \pair{\ket{0},\ket{k},\ket{1}}^{[d\mapsto 1]} + \frac{1}{\sqrt{2}} \cdot \pair{\ket{1},\ket{k},\ket{1}}^{[d\mapsto 1]}
\end{align*}

 \section{Proofs of Properties of the Operational Semantics}

 \subsection{Safety Properties}

 \lemsubstitution*

\begin{proof}
 We prove the statement by structural induction on the derivation of $\Delta,x:A\vdash M:B$.
  We examine the last typing rule of the derivation. For instance, suppose the last rule is $\lolli$I:
  \[
  \scalebox{0.9}{
  \begin{prooftree}
 \hypo{\Delta_1,x:A,\Delta_2,y:B\vdash M:C}
 \infer1[$\lolli$I]{\Delta_1,x:A,\Delta_2 \vdash \lambda y. M:B\lolli C}
 \end{prooftree}}
 \]
 Without loss of generality, we assume that the variable $y$ is not in $\Delta_3$. By induction hypothesis, $\Delta_1,\Delta_2,y:B,\Delta_3\vdash M\{N/x\}:C$ is a valid typing judgment. Therefore we can derive:
\[
\scalebox{0.9}{
\begin{prooftree}
 \hypo{\Delta_1,\Delta_2,y:B,\Delta_3\vdash M\{N/x\}:C}
 \infer[double]1[$X$]{\Delta_1,\Delta_2,\Delta_3,y:B\vdash M\{N/x\}:C}
 \infer1[$\lolli$I]{\Delta_1,\Delta_2,\Delta_3\vdash \lambda y. M\{N/x\}:B\lolli C}
 \end{prooftree}}
 \]
where the double inference line represents several applications of the $X$ rule. Since $\lambda y. M\{N/x\} = (\lambda y. M)\{N/x\}$, this proves the statement for the $\lolli$I case.
 The proofs for the other applicable rules are similar.
\end{proof}

\propsubjectreduction*

\begin{proof}
The property can be proved by induction on the derivation of $M\twoheadrightarrow \sum_{i=1}^n \alpha_i \cdot M_i^{\sigma_i}$. We examine each derivation rule.
We detail the proof for the case of the linear abstraction rule,
$\begin{prooftree}
\infer0{(\lambda x. M)V\twoheadrightarrow M\{V/x\}}
\end{prooftree}$. Suppose that we have the following derivation:
 \[
 \scalebox{0.9}{
 \begin{prooftree}
  \hypo{\begin{array}{c}\pi \\ \vdots \\ \Lambda,x:A\vdash M:B\end{array}}
  \infer1[$\lolli$I]{\Lambda\vdash \lambda x. M : A\lolli B}
  \infer[double]1[$X$]{\Delta\vdash \lambda x. M : A\lolli B}
  \hypo{\pi'}
  \infer[no rule]1{\raisebox{1mm}{\vdots}}
  \infer[no rule]1{\Delta'\vdash V:A}
  \infer2[$\lolli$E]{\Delta,\Delta'\vdash (\lambda x. M)V :B}
  \infer[double]1[$X$]{\Gamma\vdash (\lambda x. M)V :B}
 \end{prooftree}}
\]
where the typing contexts $\Lambda$ and $\Gamma$ are respectively permutations of $\Delta$ and of $\Delta,\Delta'$.
Then by Lemma~\ref{lem:substitution}, $\Lambda,\Delta'\vdash M\{V/x\}:B$ is a valid typing judgment, and we have:
\[
 \begin{prooftree}
  \hypo{\Lambda,\Delta'\vdash M\{V/x\}:B}
  \infer[double]1[$X$]{\Delta,\Delta'\vdash M\{V/x\}:B}
  \infer[double]1[$X$]{\Gamma\vdash M\{V/x\}:B}
  \end{prooftree}
\]
which completes the proof for the abstraction case.

The pair decomposition case requires two applications of the substitution lemma;
the remaining rules of Figure~\ref{fig:rel} are straightforward to check.
\end{proof}

\lemprogress*
\begin{proof}
The first point can be shown by induction on the structure of $M$.

To prove the second point, suppose $\vec{M}\equiv \alpha\cdot P^{\sigma}+ \vec{N}$ is a well-formed closed configuration with $P\notin \mathscr{V}$ and $P^\sigma \notin \Supp(\vec{N})$. Since $P$ is closed, there exists a reduction $P\twoheadrightarrow \vec{P}$. If $\sigma \sqcup \Mem(\vec{P})\ ||\ \Mem(\vec{N})$, then we have a reduction $\vec{M}\twoheadrightarrow \alpha\cdot \vec{P}^{\{\sigma\}}+\vec{N}$.
Otherwise, the reduction $P\twoheadrightarrow \vec{P}$ was necessarily obtained by applying rule $(m_k)$ for $k\in\{0,1\}$, followed by some congruence rules. By instead applying rule $(m_{1-k})$, we obtain a reduction $P\twoheadrightarrow \vec{P}$ with $\sigma \sqcup \Mem(\vec{P})\ ||\ \Mem(\vec{N})$. Therefore we have a reduction $\vec{M}\twoheadrightarrow \alpha\cdot \vec{P}^{\{\sigma\}}+\vec{N}$, which concludes the proof.
\end{proof}

 \subsection{Uniqueness of the Normal Form}

The goal of this appendix is to prove Proposition~\ref{prop:confluence}. The proof applies to both the operational semantics defined in Section~\ref{section:operational} and its extension defined in Section~\ref{section:extension}. Fix a configuration $\vec{M}$ and a target final configuration for which we know that there exists a reduction $\vec{M}\twoheadrightarrow^*\vec{V}$. We will show that the following simple algorithm  constructs a specific reduction $\vec{M}\twoheadrightarrow^*\vec{V}$: we begin with the initial configuration $\vec{M}$ in its canonical form, choose a summand that is not a value, and reduce it consistently with the target memory function of $\vec{V}$. In this case, there is exactly one possible reduction. We repeat the process until only values remain, in which case we have obtained the desired final configuration $\vec{V}$ (up to congruence). We will show that the algorithm produces a reduction with a bounded number of steps, and that every possible reduction $\vec{M}\twoheadrightarrow^* \vec{V}$ corresponds to an execution of this algorithm.

\begin{lemma}
\label{lem:unique_transition}
Let $\sigma,\nu:\device\to\{0,1\}_\bot$ be a pair of memory functions satisfying $\sigma\preceq \nu$. For all well-typed term $M\in\mathscr{E}$, there exists at most one transition $M\twoheadrightarrow \vec{M}$ such that $\Mem(\vec{M}^{\{\sigma\}})\preceq \nu$.
\end{lemma}
\begin{proof}
By induction on the structure of $M$.
\end{proof}

Fix $\vec{V}\in\vec{\mathscr{V}}$ and let $\nu=\Mem(\vec{V})$.
We inductively define the length function $L_\nu$ on i) configurations $\vec{M}$ such that $\vec{M}\twoheadrightarrow^* \vec{V}$, and ii) pairs $(M,\sigma)$ such that $M^\sigma$ is in the support of such an $\vec{M}$. Notice that we necessarily have $\Mem(\vec{M})\preceq \nu$ and $\sigma\preceq \nu$.
 \begin{align*}
 \begin{split}
  L_\nu(\vec{M}) & := \sum_{i=1}^n L_\nu(M_i,\sigma_i) \quad\ \ \text{where }\sum_{i=1}^n\alpha_i\cdot M_i^{T_i}\text{ is the canonical form of }\vec{M}
  \\
  L_\nu(M,\sigma) &:= \left\{
  \begin{array}{ll}
1+L_\nu(\vec{M})& \ \ \,\text{if }\vec{M} \text{ is the unique output as defined in Lemma~\ref{lem:unique_transition},}
\\
0 &\ \ \,\text{if $M\in\mathscr{V}$.}
\end{array}
\right.
\end{split}
 \end{align*}

 It is not immediately clear that $L_\nu$ is well defined. In particular, the unique output defined by Lemma~\ref{lem:unique_transition} does not always exist. For this reason, it is necessary to restrict the domain of $L_\nu$ to configurations that reduce to $\vec{V}$.
\begin{lemma}
 Let $\vec{V}$ be a final configuration and $\nu= \Mem(\vec{V})$. Then $L_\nu$ is well defined:
 \begin{itemize}
  \item on all configuration $\vec{M}$ such that $\vec{M}\twoheadrightarrow^* \vec{V}$;
  \item on all pair $(M,\sigma)$ such that $M^\sigma$ is in the support of such an $\vec{M}$.
 \end{itemize}
\end{lemma}
\begin{proof}
We show by induction on $n\in\mathbb{N}$ that if there exists a reduction $\vec{M}\twoheadrightarrow^n\vec{V}$ in $n$ steps, then $L_\nu(\vec{M})$ is well defined.
This is sufficient to prove both points of the lemma. Indeed, if we write the canonical form of $\vec{M}$ as $\sum_{i=1}^n \alpha_i \cdot M_i^{\sigma_i}$, $L_{\nu}(\vec{M})$ being well defined implies that all the $L_{\nu}(M_i,\sigma_i)$ are also well defined.

For $n=0$ the result is clear, using $L_\nu(\vec{V})=0$. Suppose the result holds for a fixed $n\in\mathbb{N}$ and consider a reduction $\vec{M}\twoheadrightarrow \vec{M}\twoheadrightarrow^n \vec{V}$.
We write the canonical form of $\vec{M}$ as $\sum_{i=1}^n\alpha_i \cdot M_i^{\sigma_i}$, and assume that the reduction $\vec{M}\twoheadrightarrow \vec{M}$ occurs by reducing $M_{i_0}\twoheadrightarrow \sum_{j=1}^m\beta_j\cdot N_j^{\mu_j}$ ($1\leq i_0\leq n$).
By induction hypothesis, $L_\nu(\vec{M})$ is well-defined.
Then necessarily,
\begin{itemize}
 \item for all $i \in \llbracket 1,n \rrbracket \setminus \{i_0\}$, $L_\nu(M_i,\sigma_i)$ is well defined;
 \item for all $j\in \llbracket 1,m \rrbracket$, $L_\nu(N_j,\sigma_{i_0}\sqcup \mu_j)$ is well defined.
\end{itemize}
Moreover, $M_{i_0}\twoheadrightarrow \sum_{j\in J}\beta_j\cdot N_j^{\mu_j}$ is necessarily the unique reduction given by Lemma~\ref{lem:unique_transition}. Consequently, $L_\nu(M_{i_0},\sigma_{i_0})$, and therefore $L_\nu(\vec{M})$, are well defined.
\end{proof}

\begin{lemma}
 \label{lem:decrease}
 If we have the following reductions:
 \[
 \scalebox{0.9}
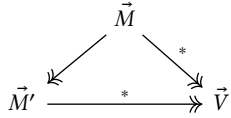
{
  \begin{tikzcd}[ampersand replacement=\&]
   \&\vec{M}
   \ar[dl,twoheadrightarrow]
   \ar[dr,twoheadrightarrow,"*"]
   \&
   \\
   \vec{M'}
   \ar[rr,twoheadrightarrow,"*"]
   \&\&\vec{V}
\end{tikzcd}}
 \]
then $L_\nu(M')\leq L_\nu(M)-1$, where $\nu:=\Mem(\vec{V})$.
\end{lemma}
\begin{proof}
 Let $\sum_{i=1}^n\alpha_i \cdot M_i^{\sigma_i}$ be the canonical form of $\vec{M}$. The reduction $\vec{M}\twoheadrightarrow\vec{M'}$ is obtained by reducing one of the summands of the canonical form, say $M_{i_0}\twoheadrightarrow \sum_{j=1}^n\beta_j\cdot N_j^{\mu_j}$. Then we have $\vec{M}\twoheadrightarrow \alpha_{i_0}\cdot \sum_{j}\beta_j\cdot N_j^{\sigma_{i_0}\sqcup \mu_j}+ \sum_{i\neq i_0}\alpha_i\cdot M_i^{\sigma_i} \equiv \vec{M'}$. Therefore:
\[
 L_\nu(\vec{M'})\leq L_\nu(\vec{M})-1
 \]
 (This is an inequality because some of the $N_j^{\sigma_{i_0}\sqcup \mu_j}$ may be equal to some of the remaining $M_i^{\sigma_i}$.)
\end{proof}

\begin{lemma}
\label{lem:bounded_length}
 Let $\vec{M}\in\vec{\mathscr{E}}$ and $\vec{V}\in\vec{\mathscr{V}}$ such that $\vec{M}\twoheadrightarrow^*\vec{V}$. Every reduction $\vec{M}\twoheadrightarrow^* \vec{V}$ has at most $L_\nu(\vec{M})$ steps, where $\nu:=\Mem(\vec{V})$.
\end{lemma}
\begin{proof}
Fix $\vec{V}\in\mathscr{\vec{V}}$. We prove by induction on $n\in\mathbb{N}$ that for all $\vec{M}\in\vec{\mathscr{E}}$ such that $\vec{M}\twoheadrightarrow^*\vec{V}$ and $L_\nu(\vec{M})=n$, every reduction $\vec{M}\twoheadrightarrow^* \vec{V}$ has length at most $n$. If $L_\nu(\vec{M})=0$, we must have $\vec{M}\equiv\vec{V}$ and the result is clear. Fix $n\geq 0$ and suppose the result holds for all $0\leq k \leq n$. Let $\vec{M}$ be a configuration such that $\vec{M}\twoheadrightarrow^*\vec{V}$ and $L_\nu(\vec{M})=n+1$.
We write a reduction $\vec{M}\twoheadrightarrow^* \vec{V}$ as $\vec{M}\twoheadrightarrow\vec{M'}\twoheadrightarrow^*\vec{V}$, for $\vec{M'}\in\vec{\mathscr{E}}$.
By Lemma~\ref{lem:decrease}, $L_\nu(\vec{M'})\leq L_\nu(\vec{M})-1\leq n$.
By induction hypothesis, every reduction $\vec{M'}\twoheadrightarrow^* \vec{V}$ has length at most $n$, and therefore every reduction $\vec{M}\twoheadrightarrow^*\vec{V}$ whose first step is $\vec{M}\twoheadrightarrow \vec{M'}$ has length at most $n+1$. Since this holds for any $\vec{M'}$,
every reduction $\vec{M}\twoheadrightarrow^*\vec{V}$ has length at most $n+1$.
\end{proof}

We are now able to prove Proposition~\ref{prop:confluence}:

\propconfluence*

\begin{proof}
Fix $\vec{V}\in\vec{\mathscr{V}}$ and let $\nu:=\Mem(\vec{V})$.
We prove by induction on $n\in\mathbb{N}$ that for all $\vec{M},\vec{N}$ such that $\vec{M}\twoheadrightarrow^* \vec{V}$, $L_\nu(\vec{M})=n$, $\vec{M}\twoheadrightarrow^* \vec{N}$ and $\Mem(\vec{N})\ ||\ \Mem(\vec{V})$, we have $\vec{N}\twoheadrightarrow^*\vec{V}$. For $n=0$, we have $\vec{M}\equiv\vec{V}$ and the result is clear.

Let $n\geq 0$ and suppose the result holds for all $0\leq k\leq n$.
To show the result for $n+1$, we fix a configuration $\vec{M}$ such that $\vec{M}\twoheadrightarrow^* \vec{V}$ and $L_\nu(\vec{M})=n+1$, and prove by induction on $m\in\mathbb{N}$ that for all $\vec{N}$ such that $\vec{M}\twoheadrightarrow^m\vec{N}$ and $\Mem(\vec{N})\ ||\ \Mem(\vec{V})$, we have $\vec{N}\twoheadrightarrow^*\vec{V}$ and $L_\nu(\vec{N})\leq n+1$. If $m=0$ the result is clear. Now suppose the result holds for a fixed $m$ and consider $\vec{M}\twoheadrightarrow^{m+1} \vec{N}$ with $\Mem(\vec{N})\ ||\ \Mem(\vec{V})$.
We write the reduction as $\vec{M}\twoheadrightarrow^m \vec{P}\twoheadrightarrow \vec{N}$.
Since $\Mem(\vec{N})\ ||\ \Mem(\vec{V})$ and $\Mem(\vec{P})\preceq \Mem(\vec{N})$, we necessarily have $\Mem(\vec{P})\ ||\ \Mem(\vec{V})$.
By induction hypothesis (inner induction), there exists a reduction $\vec{P}\twoheadrightarrow^*\vec{V}$ and $L_\nu(\vec{P})\leq n+1$. If $L_\nu(\vec{P})< n+1$, we obtain the result by induction hypothesis (outer induction).

Otherwise, we write the reduction as $\vec{P}\twoheadrightarrow\vec{Q}\twoheadrightarrow^*\vec{V}$.
We write the canonical form of $\vec{P}$ as $\sum_{i=1}^n\alpha_i\cdot P_i^{\sigma_i}$. Suppose that the reduction $\vec{P}\twoheadrightarrow \vec{N}$ occurs by reducing $P_{i_0}\twoheadrightarrow \sum_{j\in J}\beta_j \cdot N_j^{\mu_j}$ and the reduction $\vec{P}\twoheadrightarrow \vec{Q}$ occurs by reducing $P_{i_1}\twoheadrightarrow \sum_{k\in K}\gamma_k \cdot Q_k^{\theta_k}$ ($i_0,i_1\in \llbracket 1,n \rrbracket$). If $i_0=i_1$ then $\vec{N}=\vec{Q}$ and the result holds.
Otherwise, we make the following observation. One of the $N_j^{\sigma_{i_0}\sqcup \mu_j}$ may be equal to $P_{i_1}^{\sigma_{i_1}}$, and similarly one of the $Q_k^{\sigma_{i_1}\sqcup \theta_k}$ may be equal to $P_{i_0}^{\sigma_{i_0}}$; but these two events cannot simultaneously be true. Indeed, if both are true, then $\sigma_{i_0}=\sigma_{i_1}$ and $\mu_j=\theta_k=[\bot]$, and there are no other possible reductions from $P_{i_0}$ and $P_{i_1}$. Therefore, no sequence of reductions from $\vec{P}$ would ever reach a final configuration.
Then by reducing all occurrences of $P_{i_0}$ and $P_{i_1}$, we obtain three possible cases, represented diagrammatically as follows:
 \[
 \scalebox{0.9}
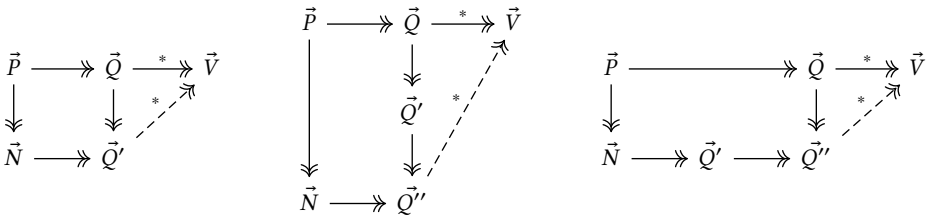
{
 \begin{tikzcd}[ampersand replacement=\&]
  \vec{P}
  \ar[r,twoheadrightarrow]
  \ar[d,twoheadrightarrow]
  \&
  \vec{Q}
  \ar[r,twoheadrightarrow,"*"]
  \ar[d,twoheadrightarrow]
  \&\vec{V}
  \\
  \vec{N}
  \ar[r,twoheadrightarrow]
  \&\vec{Q'}
  \ar[ru,twoheadrightarrow,dashed, "*"]\&
 \end{tikzcd}}
\qquad
\scalebox{0.9}{
  \begin{tikzcd}[ampersand replacement=\&]
   \vec{P}
   \ar[dd,twoheadrightarrow]
   \ar[r,twoheadrightarrow]
   \&\vec{Q}
   \ar[r,twoheadrightarrow,"*"]
   \ar[d,twoheadrightarrow]
   \&\vec{V}
   \\
   \& \vec{Q'}
   \ar[d,twoheadrightarrow]
   \&
   \\
   \vec{N}
   \ar[r,twoheadrightarrow]\& \vec{Q''}
   \ar[ruu,twoheadrightarrow,"*",dashed]
   \&
  \end{tikzcd}}
  \qquad
  \scalebox{0.9}{
  \begin{tikzcd}[ampersand replacement=\&]
   \vec{P}
   \ar[d,twoheadrightarrow]
   \ar[rr,twoheadrightarrow]
   \&\&\vec{Q}
   \ar[r,twoheadrightarrow,"*"]
   \ar[d,twoheadrightarrow]
   \&\vec{V}
   \\
   \vec{N}
   \ar[r,twoheadrightarrow]
   \& \vec{Q'} \ar[r,twoheadrightarrow]
   \&\vec{Q''}\ar[ru,twoheadrightarrow, dashed, "*"]
   \&
  \end{tikzcd}}
 \]
 where $\vec{Q'}$ and $\vec{Q''}$ can be chosen such that $\Mem(\vec{Q'})\ ||\ \Mem(\vec{V})$ and $\Mem(\vec{Q''})\ ||\ \Mem(\vec{V})$. By Lemma~\ref{lem:decrease}, $L_{\nu}(\vec{Q})\leq L_{\nu}(\vec{P})-1=n$.
 Then, the induction hypothesis (outer induction) guarantees the existence of the dashed arrows. Lemma~\ref{lem:decrease} also ensures that $L_\nu(\vec{N})\leq L_{\nu}(\vec{P})-1 = n$, which concludes the proof.
\end{proof}

We have the following corollary:
\begin{corollary}
 \label{cor:termination}
Suppose that
\begin{align*}
 \vec{M}&\twoheadrightarrow^* \vec{V} &\text{where }\vec{V}\text{ is a final configuration}
 \\
 \vec{M}&\twoheadrightarrow \vec{M}_1 \twoheadrightarrow \vec{M}_2 \cdots &\text{is an infinite sequence of reductions}
\end{align*}
Then there exists $n\geq 1$ such that $\Mem(\vec{V})$ and $\Mem(\vec{M_n})$ are not consistent.
\end{corollary}
\begin{proof}
By contradiction, suppose for all $n\geq1$, $\Mem(\vec{M_n})\ ||\ \Mem(\vec{V})$. By Proposition~\ref{prop:confluence}, there exist reductions
 \[
 \scalebox{0.9}
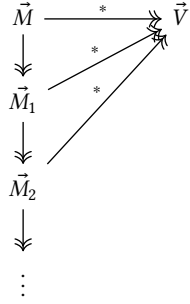
{
  \begin{tikzcd}[ampersand replacement=\&]
   \vec{M}
   \ar[d,twoheadrightarrow]
   \ar[rr,twoheadrightarrow,"*"]
   \&\& \vec{V} \\
   \vec{M_1}
   \ar[d,twoheadrightarrow]
   \ar[rru,twoheadrightarrow,"*"]\&\&
   \\
   \vec{M_2}
   \ar[d,twoheadrightarrow]
   \ar[rruu,twoheadrightarrow,"*"]\&\&
   \\
   \vdots
  \end{tikzcd}}
 \]
Therefore there exist reductions $\vec{M}\twoheadrightarrow^* \vec{V}$ of arbitrary length, which contradicts Lemma~\ref{lem:bounded_length}.
\end{proof}

In short, normalization is necessarily strong.
This means that whether a sequence of reductions from a configuration $\vec{M}$ terminates accurately represents whether the corresponding program execution terminates.

\section{Denotational Semantics: Full Definition and Proofs}
\label{app:denotational}

In this appendix, we give the proofs of results stated in Section~\ref{section:denotational} as well as the full definition of $\sem{\cdot}$.

\proppurewelldefined*
\begin{proof}
 One can show by induction on the derivation of the typing judgment $\Delta\vdash M:A$ that i) its interpretation $[\Delta\vdash M:A]$ is well defined, and ii) for all typing context $\Gamma$ obtained by permuting the variables of $\Delta$, the interpretation $[\Gamma\vdash M:A]$ is well defined and can be obtained by precomposing $[\Delta\vdash M:A]$ with the appropriate symmetry morphisms.
\end{proof}

The full definition of the denotational semantics $\sem{\cdot}$ is given in Figure~\ref{fig:denotational}.
The morphism $\mathfrak{M}(f,g)$ is defined as follows for all $f,g:\sem{\Delta'}\to \sem{A}$:
\[
 \mathfrak{M}(f,g) = \mathcal{P}_0\otimes f+ \mathcal{P}_1\otimes g : \sem{\q} \otimes \sem{\Delta'} \to \sem{A}
\]
where $\mathcal{P}_k =\bra{k}(\cdot)\ket{k}$ is the projection on the computational basis state $\ket{k}$ ($k\in\{0,1\}$).
The morphism $\mathfrak{q}(f,g):[\q]\otimes[\Delta']\to [A\lolli \q^n]$ is defined as follows for all $f,g:[\Delta']\to [A\lolli \q^n]$:
 \[
 \mathfrak{q}(f,g) := (\gamma_{[\q],[A]^*}\otimes I_{[\q^n]})(\ket{0}\bra{0}\otimes f + \ket{1}\bra{1}\otimes g).
\]

\begin{figure}[!ht]
\hrulefill
\begin{center}
\[
\scalebox{0.8}{
 \begin{prooftree}
  \hypo{\vphantom{k\in\{0,1\}}}
 \infer1{\sem{x:A \vdash x:A} \ =\ \sem{A}\xrightarrow{1_{\sem{A}}}\sem{A}}
\end{prooftree}}
 \qquad
 \scalebox{0.8}{
 \begin{prooftree}
  \hypo{\vphantom{k\in\{0,1\}}}
  \infer1{\sem{\vdash U:\q \lolli \q}\ =\  \mathbf{I}\xrightarrow{D(I\otimes U)\circ \eta'_{[q]}}\sem{\q\lolli \q}}
 \end{prooftree}}
 \qquad
 \scalebox{0.8}{
 \begin{prooftree}
  \hypo{k\in\{0,1\}}
  \infer1{\sem{\vdash \ket{k}:\q}\ = \ \mathbf{I} \xrightarrow{1\mapsto \ket{k}\bra{k}} \sem{\q}}
 \end{prooftree}}
\]
\vspace{0.5ex}
\[
\scalebox{0.8}{
 \begin{prooftree}
  \hypo{\sem{\Delta,x:A,y:B,\Delta'\vdash M:C} = \sem{\Delta}\otimes\sem{A}\otimes\sem{B}\otimes \sem{\Delta'} \xrightarrow{f} \sem{C}}
  \infer1{\sem{\Delta,y:B,x:A,\Delta'\vdash M:C} = \sem{\Delta}\otimes\sem{B}\otimes\sem{A}\otimes \sem{\Delta'} \xrightarrow{1_{\sem{\Delta}}\otimes \gamma_{\sem{B},\sem{A}}\otimes 1_{\sem{\Delta'}}}\sem{\Delta}\otimes\sem{A}\otimes\sem{B}\otimes \sem{\Delta'} \xrightarrow{f} \sem{C}}
 \end{prooftree}}
\]
 \vspace{0.5ex}
\[
\scalebox{0.8}{
\begin{prooftree}
 \hypo{\sem{\Delta,x:A\vdash M:B} \ = \ \sem{\Delta}\otimes\sem{A} \xrightarrow{f} \sem{B}}
 \infer1{\sem{\Delta\vdash \lambda x. M:A\lolli B} \ = \ \sem{\Delta} \xrightarrow{\phi(f)}\sem{A}\lolli\sem{B}}
 \end{prooftree}}
 \quad
 \scalebox{0.8}{
 \begin{prooftree}
  \hypo{\sem{\Delta\vdash M:A\lolli B} \ = \ \sem{\Delta}\xrightarrow{f} \sem{A\lolli B}
}
\hypo{\sem{\Delta'\vdash N:A} \ = \ \sem{\Delta'} \xrightarrow{g}\sem{A}}
  \infer2{\sem{\Delta,\Delta'\vdash M N:B} \ = \ \sem{\Delta}\otimes \sem{\Delta'} \xrightarrow{f\otimes g} \sem{A\lolli B}\otimes \sem{A} \xrightarrow{\phi^{-1}(1_{A\lolli B})} \sem{B}}
 \end{prooftree}}
 \]
 \vspace{0.5ex}
\[
\scalebox{0.8}{
 \begin{prooftree}
  \hypo{\sem{\Delta \vdash M:A} \ =\ \sem{\Delta}\xrightarrow{f}\sem{A}}
  \hypo{\sem{\Delta'\vdash N:B} \ =\ \sem{\Delta'} \xrightarrow{g}\sem{B}}
  \infer2{\sem{\Delta,\Delta' \vdash \pair{M , N}:A\otimes B}\ = \ \sem{\Delta}\otimes\sem{\Delta'} \xrightarrow{f\otimes g} \sem{A}\otimes\sem{B} }
 \end{prooftree}}
 \]
\vspace{0.5ex}
\[
\scalebox{0.8}{
  \begin{prooftree}
\hypo{\sem{\Delta \vdash M:A\otimes B}\ = \ \sem{\Delta} \xrightarrow{f}\sem{A\otimes B}
}
\hypo{\sem{\Delta',x:A,y:B\vdash N:C}\ = \ \sem{\Delta'}\otimes \sem{A}\otimes \sem{B} \xrightarrow{g}\sem{C}}
\infer2{\sem{\Delta,\Delta' \vdash \mathtt{let}\ \pair{x,y}=M\ \mathtt{in}\ N:C} \ = \ \sem{\Delta}\otimes\sem{\Delta'} \xrightarrow{f\otimes 1_{\sem{\Delta'}}} \sem{A\otimes B} \otimes\sem{\Delta'} \xrightarrow{\gamma_{\sem{A\otimes B},\sem{\Delta'}}} \sem{\Delta'}\otimes \sem{A\otimes B}\xrightarrow{g} \sem{C}}
\end{prooftree}}
\]
\vspace{0.5ex}
\[
\scalebox{0.8}{
 \begin{prooftree}
  \hypo{\vphantom{\xrightarrow{f}}}
  \infer1{\sem{\vdash \unit:\one} \ = \ \mathbf{I}\xrightarrow{1_{\mathbf{I}}} \mathbf{I}}
 \end{prooftree}}
 \qquad
 \scalebox{0.8}{
 \begin{prooftree}
  \hypo{\sem{\Delta\vdash M: \one}\ = \ \sem{\Delta} \xrightarrow{f} \mathbf{I}}
  \hypo{\sem{\Delta' \vdash N: A}\ = \ \sem{\Delta'}\xrightarrow{g}\sem{A}}
  \infer2{\sem{\Delta,\Delta'\vdash M;N:A}\ = \ \sem{\Delta}\otimes\sem{\Delta'} \xrightarrow{f\otimes g} \sem{A}}
 \end{prooftree}}
\]
\vspace{0.5ex}
\[
\scalebox{0.8}{
 \begin{prooftree}
 \hypo{\sem{\Delta \vdash P:\q}\ = \ \sem{\Delta} \xrightarrow{p} \sem{\q}}
  \hypo{\sem{\Delta' \vdash M:A}\ = \ \sem{\Delta'}\xrightarrow{f}\sem{A}}
  \hypo{\sem{\Delta' \vdash N: A} \ = \ \sem{\Delta'} \xrightarrow{g} \sem{A}}
  \infer3{\sem{\Delta,\Delta'\vdash \meas\ P\ \{0\rightarrow M\ |\ 1\rightarrow N\}: A}\ = \ \sem{\Delta}\otimes\sem{\Delta'}\xrightarrow{p\otimes 1_{\sem{\Delta'}}} \sem{\q}\otimes \sem{\Delta'} \xrightarrow{\mathfrak{M}(f,g)}\sem{A}}
 \end{prooftree}}
\]
\vspace{0.5ex}
\[
\scalebox{0.8}{
 \begin{prooftree}
  \hypo{\sem{\Delta \vdash P:\q} = \sem{\Delta} \xrightarrow{p}\sem{\q}}
  \hypo{[\Delta' \vdash M:A\lolli \q^n] = [\Delta'] \xrightarrow{f}[A\lolli \q^n]}
  \hypo{[\Delta' \vdash N: A\lolli \q^n] = [\Delta'] \xrightarrow{g} [A\lolli\q^n]}
  \infer3{\sem{\Delta,\Delta'\vdash \qcase\ P\ \{0\rightarrow M\ |\ 1\rightarrow N\}: A\lolli \q^{n+1}} = \sem{\Delta}\otimes\sem{\Delta'} \xrightarrow{p\otimes 1_{\sem{\Delta'}}} \sem{\q}\otimes \sem{\Delta'}\xrightarrow{D(\mathfrak{q}(f,g))}\sem{A\lolli\q^{n+1}}}
 \end{prooftree}}
\]

\end{center}

\hrulefill
\caption{Denotational semantics}
\label{fig:denotational}
\end{figure}

We now give proofs leading up to the fact that $\sem{\cdot}$ is well defined.

\begin{lemma}
\label{lem:double_qcase}
 Let $f,g:[\Lambda]\to [A\lolli \q^n]$ be morphisms of $\FHilb$ such that $D(f),D(g):\sem{\Lambda}\to \sem{A\lolli \q^n}$ are morphisms of $\Caus[\CPM]$. Then $D(\mathfrak{q}(f,g)):\sem{\q}\otimes\sem{\Lambda}\to\sem{A\lolli \q^{n+1}}$ is a morphism of $\Caus[\CPM]$.
\end{lemma}
\begin{proof}
The proof relies on a few properties of causal categories.
By~\cite[Definition 5.1]{kissinger2019categorical}, the object $\sem{\q}$ is first order.  Since first-order systems are closed under $\otimes$, so is $\sem{\q^{n+1}}$.
Let $\pi\in c_{\sem{A\lolli \q^{n+1}}}^*$. By~\cite[Lemma 6.1]{kissinger2019categorical}, there exists $\pi'\in c_{\sem{A}}$ such that $\pi = \pi'\otimes \mathfrak{d}_{[\q^{n+1}]}^*$, where $\mathfrak{d}_{[\q^{n+1}]}$ is given by the trace function $Tr_{[\q^{n+1}]}$.
We aim to show that $\pi \circ D(\mathfrak{q}(f,g))\in c_{\sem{\q}\otimes \sem{\Lambda}}^*=(c_{\sem{\q}}\otimes c_{\sem{\Lambda}})^*$. Let $\rho_1\in c_{\sem{\q}}$ and $\rho_2\in c_{\sem{\Lambda}}$. We have:
\begin{align*}
 &\pi \circ  D(\mathfrak{q}(f,g)) \circ (\rho_1\otimes \rho_2)
 \\
 &\quad= (\pi'\otimes Tr_{[\q^{n+1}]})\left(\gamma_{\sem{\q},\sem{A}^*}\left( (\ket{0}\bra{0}\otimes f+\ket{1}\bra{1}\otimes g)(\rho_1\otimes \rho_2)(\ket{0}\bra{0}\otimes f^\dag+\ket{1}\bra{1}\otimes g^\dag)\right)\right)
 \\
 &\quad = \bra{0}\rho_1\ket{0}(\pi'\otimes Tr_{[\q^{n}]})\left( f\rho_2 f^\dag\right)+ \bra{1}\rho_1\ket{1}(\pi'\otimes Tr_{[\q^{n}]})\left( g\rho_2 g^\dag\right)
 \\
 &\quad = \bra{0}\rho_1\ket{0}(\pi'\otimes Tr_{[\q^{n}]})\circ D(f)\circ \rho_2 + \bra{1}\rho_1\ket{1}(\pi'\otimes Tr_{[\q^{n}]})\circ D(g) \circ \rho_2
 \\
 &\quad = \bra{0}\rho_1\ket{0}+ \bra{1}\rho_1\ket{1}
 \\
 &\quad =1
\end{align*}
where we used the fact that $c_{\sem{A}}\otimes c_{\sem{\q^n}}^*\subseteq c_{\sem{A\lolli \q^n}}^*$. Therefore $D(\mathfrak{q}(f,g))$ is a morphism of $\Caus[\CPM]$.

\end{proof}

\lemdouble*
\begin{proof}
 $[\Delta] \xrightarrow{D(f)} [A]$ is a morphism of $\CPM$. To show that $\sem{\Delta} \xrightarrow{D(f)} \sem{A}$ is a morphism of $\Caus[\CPM]$, we must show that for all $\rho\in c_{\sem{\Delta}}$, we have $D(f)\circ \rho =f\rho f^\dag \in c_{\sem{A}}$.
We prove this statement by induction on the derivation of $\Delta\vdash M:A$. We examine the last rule in the derivation. For most of the rules, the statement can be proved using the fact that $\Caus[\CPM]$ is a symmetric monoidal closed category, and that the bijection $\phi$ is the same in $\CPM$ and $\Caus[\CPM]$.

For the unitary rule, showing that $\mathbf{I}\xrightarrow{D(I\otimes U)\circ \eta'_{[q]}}\sem{\q\lolli \q}$ is a morphism of $\Caus[\CPM]$ is the same as showing that $D(I\otimes U)\circ \eta'_{[q]}\in c_{\sem{\q\lolli \q}}= c_{\sem{\q}\lolli \sem{\q}}$. We have $c_{\sem{\q}}= St(\mathbb{C}^2)$ and $c_{\sem{\q}^*}=\{Tr\}$. Then the result follows from~\cite[Lemma 4.9]{kissinger2019categorical}, using the fact that for all $\rho\in c_{\sem{\q}}$, $Tr(U\rho U^\dag)=1$.

For the qcase rule, the statement follows from Lemma~\ref{lem:double_qcase}.
\end{proof}

\propmixedwelldefined*
\begin{proof}
To prove that the interpretation is well defined, we show that the interpretation of each valid typing judgment $\Delta \vdash M:A$ is well defined first as a morphism of $\CPM$, then as a morphism of $\Caus[\CPM]$.

To prove that the interpretation is well defined in $\CPM$, one can show by induction that the interpretation of a typing judgment is independent of the particular derivation of that judgment, similarly to the proof of Proposition~\ref{prop:pure_well_defined}.

Then, we prove by induction of the derivation of $\Delta\vdash M:A$ that its interpretation is in $\Caus[\CPM]$. When examining the last rule in the derivation, in most cases the statement can be shown using the fact that $\Caus[\CPM]$ is a symmetric monoidal closed category. We treat the unitary, measurement and qcase rules in more detail.

For the unitary rule, following from~\cite[Lemma 4.9]{kissinger2019categorical}, we have $D(I\otimes U)\circ \eta'_{[q]}\in c_{\sem{\q\lolli \q}}$. Therefore $\mathbf{I}\xrightarrow{D(I\otimes U)\circ \eta'_{[q]}} \sem{\q\lolli\q}$ is a morphism of $\Caus[\CPM]$.

For the measurement rule, we must show that if $f,g:\sem{\Delta'}\to\sem{A}$ are morphisms of $\Caus[\CPM]$, then so is $\mathfrak{M}(f,g)$. Let $\pi \in c_{\sem{A}}^*$. To see that $\pi \circ\mathfrak{M}(f,g)\in c_{\sem{\q}\otimes\sem{\Delta'}}^*=(c_{\sem{\q}}\otimes c_{\sem{\Delta'}})^*$, we have for all $\rho_1\in c_{\sem{\q}}$ and $\rho_2\in c_{\sem{\Delta'}}$:
\begin{align*}
 \pi \circ \mathfrak{M}(f,g) \circ (\rho_1\otimes \rho_2) &= \pi\bigl(\mathcal{P}_0(\rho_1)\otimes f(\rho_2)+ \mathcal{P}_1(\rho_1)\otimes g(\rho_2) \bigr)\\
 &= \bra{0}\rho_1\ket{0} \pi \circ f \circ \rho_2 + \bra{1}\rho_1\ket{1} \pi \circ g \circ \rho_2
 \\
 &= \bra{0}\rho_1\ket{0} + \bra{1}\rho_1\ket{1}
 \\
 &=1
\end{align*}
Therefore $\mathfrak{M}(f,g): \sem{\q}\otimes \sem{\Delta'} \to \sem{A}$ is a morphism of $\Caus[\CPM]$.

For the qcase rule, suppose $\Delta'\vdash M:A\lolli \q^n$ and $\Delta'\vdash M:A\lolli \q^n$ are valid typing judgments with $M,N\in\mathscr{C}$, whose interpretations we write as $[\Delta' \vdash M:A\lolli \q^n] = [\Delta'] \xrightarrow{f}[A\lolli \q^n]$ and $[\Delta' \vdash N: A\lolli \q^n] = [\Delta'] \xrightarrow{g} [A\lolli\q^n]$. By Lemma~\ref{lem:double}, $D(f):\sem{\Delta'} \to \sem{A\lolli \q^n}$ and $D(g):\sem{\Delta'} \to\sem{A\lolli \q^n}$ are morphism of $\Caus[\CPM]$. By Lemma~\ref{lem:double_qcase}, $D(\mathfrak{q}(f,g)):\sem{\q}\otimes \sem{\Delta'} \to \sem{A \lolli \q^{n+1}}$ is a morphism of $\Caus[\CPM]$, which concludes the proof.
\end{proof}

\section{Soundness}
\label{app:soundness}

The interpretation of terms of $\mathscr{E}$ is inductively defined in Figure~\ref{fig:denotational_execution}, where $\mathfrak{m}_i(f):= \bra{i}\otimes f$.

\begin{figure}[!ht]
\hrulefill
\begin{center}
\[
\scalebox{0.8}{
 \begin{prooftree}
  \hypo{\vphantom{k\in\{0,1\}}}
 \infer1{[x:A \vdash x:A]_\nu \ =\ [A]\xrightarrow{1_{[A]}}[A]}
\end{prooftree}}
 \qquad
 \scalebox{0.8}{
 \begin{prooftree}
   \hypo{\vphantom{k\in\{0,1\}}}
  \infer1{[\vdash U:\q \lolli \q]_\nu\ =\  I\xrightarrow{(I\otimes U)\circ \eta_{[\q]}}[\q\lolli \q]}
 \end{prooftree}}
 \qquad
 \scalebox{0.8}{
 \begin{prooftree}
  \hypo{k\in\{0,1\}}
  \infer1{[\vdash \ket{k}:\q]_\nu\ = \ I \xrightarrow{1\mapsto \ket{k}} [\q]}
 \end{prooftree}}
\]
\vspace{0.5ex}
\[
\scalebox{0.8}{
 \begin{prooftree}
  \hypo{[\Delta,x:A,y:B,\Delta'\vdash M:C]_\nu = [\Delta]\otimes[A]\otimes[B]\otimes [\Delta'] \xrightarrow{f} [C]}
  \infer1{[\Delta,y:B,x:A,\Delta'\vdash M:C]_\nu = [\Delta]\otimes[B]\otimes[A]\otimes [\Delta'] \xrightarrow{1_{[\Delta]}\otimes \gamma_{[B],[A]}\otimes 1_{[\Delta']}}[\Delta]\otimes[A]\otimes[B]\otimes [\Delta'] \xrightarrow{f} [C]}
 \end{prooftree}}
\]
\vspace{0.5ex}
 \[
 \scalebox{0.8}{
\begin{prooftree}
 \hypo{[\Delta,x:A\vdash M:B]_\nu = [\Delta]\otimes [A] \xrightarrow{f} [B]}
 \infer1{[\Delta\vdash \lambda x. M:A\lolli B]_\nu = [\Delta] \xrightarrow{\phi(f)}[A]\lolli[B]}
 \end{prooftree}}
 \quad
 \scalebox{0.8}{
 \begin{prooftree}
  \hypo{[\Delta\vdash M:A\lolli B]_{\nu} = [\Delta]\xrightarrow{f} [A\lolli B]}
\hypo{[\Delta'\vdash N:A]_{\nu} = [\Delta'] \xrightarrow{g}[A]}
  \infer2{[\Delta,\Delta'\vdash M N:B]_{\nu} = [\Delta]\otimes [\Delta'] \xrightarrow{f\otimes g} [A\lolli B]\otimes [A] \xrightarrow{\phi^{-1}(1_{A\lolli B})} [B]}
 \end{prooftree}}
 \]
\vspace{0.5ex}
\[
\scalebox{0.8}{
 \begin{prooftree}
  \hypo{[\Delta \vdash M:A]_{\nu} \ =\ [\Delta]\xrightarrow{f}[A]}
  \hypo{[\Delta'\vdash N:B]_{\nu} \ =\ [\Delta'] \xrightarrow{g}[B]}
  \infer2{[\Delta,\Delta' \vdash \pair{M , N}:A\otimes B]_\nu\ = \ [\Delta]\otimes[\Delta'] \xrightarrow{f\otimes g} [A]\otimes[B] }
 \end{prooftree}}
 \]
\vspace{0.5ex}
\[
\scalebox{0.8}{
  \begin{prooftree}
\hypo{[\Delta \vdash M:A\otimes B]_{\nu}\ = \ [\Delta] \xrightarrow{f}[A\otimes B]
}
\hypo{[\Delta',x:A,y:B\vdash N:C]_{\nu}\ = \ [\Delta']\otimes [A]\otimes [B] \xrightarrow{g}[C]}
\infer2{[\Delta,\Delta' \vdash \mathtt{let}\ \pair{x,y}=M\ \mathtt{in}\ N:C]_\nu \ = \ [\Delta]\otimes [\Delta'] \xrightarrow{f\otimes 1_{[\Delta']}} [A\otimes B] \otimes [\Delta'] \xrightarrow{\gamma_{[A\otimes B], [\Delta']}} [\Delta']\otimes [A\otimes B]\xrightarrow{g} [C]}
\end{prooftree}}
\]
\vspace{0.5ex}
\[
\scalebox{0.8}{
 \begin{prooftree}
  \hypo{\vphantom{\xrightarrow{f}}}
  \infer1{[\vdash \unit:\one]_\nu \ = \ I\xrightarrow{1_{I}} I}
 \end{prooftree}
 \qquad
 \begin{prooftree}
 \hypo{[\Delta\vdash M: \one]_{\nu}\ = \ [\Delta] \xrightarrow{f} I}
  \hypo{[\Delta' \vdash N: A]_{\nu}\ = \ [\Delta']\xrightarrow{g}[A]}
  \infer2{[\Delta,\Delta'\vdash M;N:A]_\nu = \ [\Delta]\otimes[\Delta'] \xrightarrow{f\otimes g} [A]}
 \end{prooftree}}
\]
\vspace{0.5ex}
\[
\scalebox{0.8}{
 \begin{prooftree}
 \hypo{[\Delta \vdash P:\q]_{\nu}\ = \ [\Delta] \xrightarrow{p} [\q]}
  \hypo{[\Delta' \vdash M_i:A]_{\nu}\ = \ [\Delta']\xrightarrow{f_i}[A]}
\hypo{\nu(d)=i}
  \infer3{[\Delta,\Delta'\vdash \meas\ d\triangleright P\ \{0\rightarrow M_0\ |\ 1\rightarrow M_1\}: A]_\nu \ = \ [\Delta]\otimes[\Delta']\xrightarrow{p\otimes 1_{\sem{\Delta'}}} [\q]\otimes [\Delta'] \xrightarrow{\mathfrak{m}_i(f_i)}[A]}
 \end{prooftree}}
\]
\vspace{0.5ex}
\[
\scalebox{0.8}{
 \begin{prooftree}
  \hypo{[\Delta \vdash P:\q]_{\nu} = [\Delta] \xrightarrow{p}[\q]}
  \hypo{[\Delta' \vdash M:A\lolli \q^n]_{\nu} = [\Delta'] \xrightarrow{f}[A \lolli \q^{n}]}
  \hypo{[\Delta' \vdash N: A\lolli \q^n]_{\nu} = [\Delta'] \xrightarrow{g} [A\lolli \q^n]}
  \infer3{[\Delta,\Delta'\vdash \qcase\ P\ \{0\rightarrow M\ |\ 1\rightarrow N\}: A\lolli \q^{n+1}]_\nu\ = \ [\Delta]\otimes[\Delta'] \xrightarrow{p\otimes 1_{[\Delta']}} [\q]\otimes [\Delta']\xrightarrow{\mathfrak{q}(f,g)}[A\lolli \q^{n+1}]}
 \end{prooftree}}
\]

\end{center}
\hrulefill
\caption{Denotational semantics of $\mathscr{E}$}
\label{fig:denotational_execution}
\end{figure}

\lemnuwelldefined*
\begin{proof}
We show by induction on the derivation of $\Delta\vdash M:A$ that its interpretation is independent of the particular derivation of the typing judgment, similarly to the proof of Proposition~\ref{prop:pure_well_defined}. In particular, we use the fact that $\nu$ is also an extended valuation for the appropriate subterms.
\end{proof}

\lemcoincide*
\begin{proof}
 We prove the result by induction on the derivation of $\Delta \vdash M:A$.
 For the measurement rule, we partition $\Omega(\meas\ d\triangleright P\ \{0\to M \ |\ 1\to N\})$ according to the image of $d$.
 For the qcase rule, we use the fact that $\Omega(M)=\Omega(N)=[\bot]$.
\end{proof}

\thmsoundness*
\begin{proof}
We write $\vec{M}= \alpha \cdot N^\sigma + \vec{P}$, where $ N^\sigma \notin \Supp(\vec{P})$, and assume that the reduction $\vec{M} \twoheadrightarrow \vec{M'}$ is obtained by reducing $N \twoheadrightarrow \vec{N'}$. That is, we have $\vec{M'}=\alpha\cdot \vec{N'}^{\{\sigma\}}+\vec{P}$.
With $\mu:=\Mem(\vec{N'})$, we have:
\begin{align*}
 \sem{\vec{M}}_{\Mem(\vec{M'})}&=\sum_{\nu\in \Omega(\vec{M}):\  \nu\ ||\ \mu}[\vec{M}]_\nu (\cdot)[\vec{M}]_\nu^\dag
 &
 \sem{\vec{M'}} &= \sum_{\nu\in\Omega(\vec{M'})} [\vec{M'}]_\nu (\cdot)[\vec{M'}]_\nu^\dag
\end{align*}
where to simplify the notations we omit $\Delta$ and $A$.
First, we observe that $\bigl\{\nu\in \Omega(\vec{M})\ \big|\ \nu\ ||\ \mu\bigr\}=\Omega(\vec{M'})$.
Then it suffices to prove that for all $\nu$ in this set, we have $[\vec{M}]_\nu = [\vec{M'}]_\nu$.

To this end, we prove by induction on the derivation $N\twoheadrightarrow\vec{N'}$ that $[N]_\nu = [\vec{N'}]_\nu$ for all $\nu\in\Omega_+(N)$ such that $\nu\ ||\ \Mem(\vec{N'})$.
This is sufficient to prove the result. Indeed, notice that $\Omega(\vec{M'})\subseteq \bigl\{\nu\in \Omega_+(N)\ \big|\ \nu\ ||\ \Mem(\vec{N'})\bigr\}$, and for $\nu\in\Omega(\vec{M'})$:
\begin{align*}
 [\vec{M}]_{\nu} &= \alpha[N]_\nu+[\vec{P}]_\nu
 &
 [\vec{M'}]_{\nu}&=\alpha[\vec{N'}]_\nu+[\vec{P}]_\nu
\end{align*}

We detail the case of rule $(m_k)$. Suppose $N = \meas \ d\triangleright \ket{k'}\ \{0\to N_0 \ |\ 1\to N_1\}$ and $\vec{N'}=\delta_{k,k'}\cdot N_k^{[d\to k]}$. For $\nu\in\Omega_+(N)$ satisfying $\nu\ ||\ \Mem(\vec{N'})$, we necessarily have $\nu(d)=k$. By definition:
\begin{align*}
 [\vec{N'}]_\nu &= \delta_{k,k'}[N_k]_{\nu}
 \\
 [N]_\nu &= \mathfrak{m}_k([N_k]_\nu) \circ (1\mapsto \ket{k'}) = \delta_{k,k'}[N_k]_{\nu}
\end{align*}
which is the desired result.
\end{proof}

\section{Expressivity}
\label{app:expressivity}

\subsection{Quantum Channels}

\propchannels*
\begin{proof}
 By~\citet{stinespring}'s dilation theorem, every quantum channel can be decomposed into a unitary channel applied to the input together with an environment system, followed by a partial trace.
 To implement a unitary channel, we invoke the univerality of CNOT and single-qubit unitary gates~\cite{barenco}. Example~\ref{ex:cnot} shows how to implement the CNOT gate using a $\qcase$.
The partial trace can be implemented by applying the discard map (Example~\ref{ex:discard}) to every qubit of the environment system.
\end{proof}

\subsection{\textbf{QC-QC}s with Memory}
\label{app:QCQC}

We provide some more details on the construction of \textbf{QC-QC}s with memory and prove that the language is universal for this subclass. The $N$-linear map from $\CPM(\mathcal{H}_I,\mathcal{H}_O)^N$ to $\CPM(\mathcal{H}_P,\mathcal{H}_F)$ that is computed by the circuit
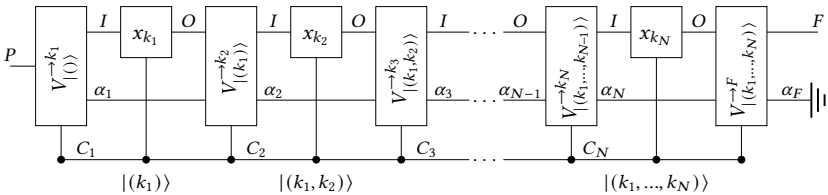
\[
\scalebox{0.9}{\begin{tikzpicture}
	\begin{pgfonlayer}{nodelayer}
		\node [style=none] (3) at (-0.25, 1.75) {};
		\node [style=none] (4) at (1.25, 1.75) {};
		\node [style=none] (5) at (2.25, 1.75) {};
		\node [style=none] (6) at (3.75, 1.75) {};
		\node [style=none] (7) at (6.25, -0.25) {};
		\node [style=none] (8) at (2.25, -0.25) {};
		\node [style=none] (16) at (-3.25, -2) {};
		\node [style=none] (17) at (8.75, -2) {};
		\node [style=none] (18) at (4.75, 1.75) {};
		\node [style=none] (19) at (6.25, 1.75) {};
		\node [style=none] (20) at (-1.25, 1.75) {};
		\node [style=none] (21) at (-2.75, 1.75) {};
		\node [style=none] (22) at (7.25, 1.75) {};
		\node [style=none] (23) at (8.75, 1.75) {};
		\node [style=none] (24) at (7.25, -0.25) {};
		\node [style=none] (25) at (8.75, -0.25) {};
		\node [style=none] (26) at (-3.75, 0.75) {};
		\node [style=none] (27) at (-4.75, 0.75) {};
		\node [style=none] (28) at (11.25, 1.75) {};
		\node [style=none] (29) at (11.25, -0.25) {};
		\node [style=none] (30) at (9.75, -0.25) {};
		\node [style=none] (31) at (9.75, 1.75) {};
		\node [style=none] (32) at (12.25, 1.75) {};
		\node [style=none] (33) at (13.75, 1.75) {};
		\node [style=none] (34) at (12.25, -0.25) {};
		\node [style=none] (35) at (16.75, -0.25) {};
		\node [style=none] (36) at (14.75, 1.75) {};
		\node [style=none] (37) at (16.25, 1.75) {};
		\node [style=none] (38) at (1.25, -0.25) {};
		\node [style=none] (39) at (-2.75, -0.25) {};
		\node [style=none] (40) at (-3.25, -0.5) {};
		\node [style=none] (41) at (1.75, -0.5) {};
		\node [style=none] (42) at (1.75, -2) {};
		\node [style=none] (43) at (6.75, -0.5) {};
		\node [style=none] (44) at (6.75, -2) {};
		\node [style=none] (45) at (11.75, -0.5) {};
		\node [style=none] (46) at (11.75, -2) {};
		\node [style=none] (47) at (9.75, -2) {};
		\node [style=none] (48) at (16.75, -2) {};
		\node [style=texthere] (49) at (9.25, 1.675) {$\cdots$};
		\node [style=texthere] (50) at (9.25, -0.325) {$\cdots$};
		\node [style=texthere] (51) at (9.25, -2.075) {$\cdots$};
		\node [style=none] (52) at (17.25, 1.75) {};
		\node [style=none] (53) at (19, 1.75) {};
		\node [style=dot] (54) at (16.75, -2) {};
		\node [style=dot] (55) at (11.75, -2) {};
		\node [style=dot] (56) at (-3.25, -2) {};
		\node [style=dot] (57) at (1.75, -2) {};
		\node [style=dot] (58) at (6.75, -2) {};
		\node [style=dot] (59) at (-0.75, -2) {};
		\node [style=dot] (60) at (4.25, -2) {};
		\node [style=dot] (61) at (14.25, -2) {};
		\node [style=none] (62) at (14.25, 1.25) {};
		\node [style=none] (63) at (-0.75, 1.25) {};
		\node [style=none] (64) at (4.25, 1.25) {};
		\node [style=texthere] (65) at (-0.75, -2.75) {$|(k_1)\rangle$};
		\node [style=texthere] (66) at (4.25, -2.75) {$|(k_1,k_2)\rangle$};
		\node [style=texthere] (67) at (14.25, -2.75) {$|(k_1,...,k_N)\rangle$};
		\node [style=upground] (68) at (19, -0.25) {};
		\node [style=none] (69) at (18.825, -0.25) {};
		\node [style=big box] (74) at (-3.25, 0.75) {\rotatebox[origin=c]{90}{$V_{|()\rangle}^{\to k_1}$}};
		\node [style=small box] (75) at (-0.75, 1.75) {$x_{k_1}$};
		\node [style=small box] (76) at (4.25, 1.75) {$x_{k_2}$};
		\node [style=small box] (77) at (14.25, 1.75) {$x_{k_N}$};
		\node [style=texthere] (78) at (-4.75, 1) {$P$};
		\node [style=texthere] (79) at (-2, 2) {$I$};
		\node [style=texthere] (80) at (0.5, 2) {$O$};
		\node [style=texthere] (81) at (-2, 0) {$\alpha_1$};
		\node [style=texthere] (82) at (3, 2) {$I$};
		\node [style=texthere] (83) at (3, 0) {$\alpha_2$};
		\node [style=texthere] (85) at (5.5, 2) {$O$};
		\node [style=texthere] (86) at (8, 2) {$I$};
		\node [style=texthere] (87) at (8, 0) {$\alpha_3$};
		\node [style=texthere] (88) at (10.25, 2) {$O$};
		\node [style=texthere] (89) at (10.25, 0) {$\alpha_{N-1}$};
		\node [style=texthere] (90) at (13, 2) {$I$};
		\node [style=texthere] (91) at (13, 0) {$\alpha_N$};
		\node [style=texthere] (92) at (15.5, 2) {$O$};
		\node [style=texthere] (93) at (19, 2) {$F$};
		\node [style=texthere] (94) at (18.25, 0) {$\alpha_F$};
		\node [style=texthere] (95) at (-2.5, -1.75) {$C_1$};
		\node [style=texthere] (96) at (2.5, -1.75) {$C_2$};
		\node [style=texthere] (97) at (7.5, -1.75) {$C_3$};
		\node [style=texthere] (98) at (12.5, -1.75) {$C_N$};
		\node [style=big box] (99) at (6.75, 0.75) {\rotatebox[origin=c]{90}{$V_{|(k_1,k_2)\rangle}^{\to k_3}$}};
		\node [style=big box] (100) at (1.75, 0.75) {\rotatebox[origin=c]{90}{$V_{|(k_1)\rangle}^{\to k_2}$}};
		\node [style=big box] (101) at (11.75, 0.75) {\rotatebox[origin=c]{90}{$V_{\!|(k_1,...,k_{N\!-\!1})\rangle}^{\to k_N}$}};
		\node [style=big box] (102) at (16.75, 0.75) {\rotatebox[origin=c]{90}{$V_{|(k_1,...,k_N)\rangle}^{\to F}$}};
	\end{pgfonlayer}
	\begin{pgfonlayer}{edgelayer}
		\draw (3.center) to (4.center);
		\draw (5.center) to (6.center);
		\draw (8.center) to (7.center);
		\draw (16.center) to (17.center);
		\draw (18.center) to (19.center);
		\draw (22.center) to (23.center);
		\draw (24.center) to (25.center);
		\draw (21.center) to (20.center);
		\draw (27.center) to (26.center);
		\draw (31.center) to (28.center);
		\draw (30.center) to (29.center);
		\draw (34.center) to (35.center);
		\draw (32.center) to (33.center);
		\draw (36.center) to (37.center);
		\draw (39.center) to (38.center);
		\draw (43.center) to (44.center);
		\draw (47.center) to (46.center);
		\draw (46.center) to (45.center);
		\draw (46.center) to (48.center);
		\draw (48.center) to (35.center);
		\draw (40.center) to (16.center);
		\draw (41.center) to (42.center);
		\draw (52.center) to (53.center);
		\draw (63.center) to (59);
		\draw (64.center) to (60);
		\draw (62.center) to (61);
		\draw (35.center) to (69.center);
	\end{pgfonlayer}
\end{tikzpicture}}
\]
is formally defined as follows.

The $V_{\ket{(k_1,...,k_n)}}^{\to k_{n+1}}$ ($0\leq n \leq N-1$) and $V_{\ket{(k_1,...,k_N)}}^{\to F}$ represent isometries. We further assume that the output spaces of the $V_{\ket{(k_1,...,k_N)}}^{\to F}$ in the final time-step are orthogonal. The corresponding controlled operations, which we will call $\tilde{V_1},...,\tilde{V_{N+1}}$, are defined as follows:
\begin{align*}
 \tilde{V}_1 &:= \sum_{k_1} V_{\ket{()}}^{\to k_1}\otimes \ket{(k_1)} \in\mathcal{L}(\mathcal{H}_P,\mathcal{H}_I\otimes \mathcal{H}_{\alpha_1}\otimes \mathcal{H}_{C_1})
 \\
 \tilde{V}_n &:=\sum_{(k_1,...,k_{n})} V_{\ket{(k_1,...,k_{n-1})}}^{\to k_{n}}\otimes \ket{(k_1,...k_{n})}\bra{(k_1,...,k_{n-1})} \in\mathcal{L}(\mathcal{H}_I\otimes \mathcal{H}_{\alpha_{n-1}}\otimes \mathcal{H}_{C_{n-1}},\mathcal{H}_I\otimes \mathcal{H}_{\alpha_{n}}\otimes \mathcal{H}_{C_{n}})
 \\
 &(\text{for }2\leq n\leq N)
 \\
 \tilde{V}_{N+1} &:=\sum_{(k_1,...,k_{N})} V_{\ket{(k_1,...,k_{N})}}^{\to F}\otimes \bra{(k_1,...,k_{N})} \in\mathcal{L}(\mathcal{H}_I\otimes \mathcal{H}_{\alpha_{N}}\otimes \mathcal{H}_{C_{N}},\mathcal{H}_F\otimes \mathcal{H}_{\alpha_{F}})
\end{align*}
where when taking sums over $(k_1,..,k_n)$ or $(k_1,...,k_N)$, the indices are assumed to be pairwise distinct.
Note that $\tilde{V}_1$, the $\tilde{V}_n$ and $\tilde{V}_{N+1}$ are also isometries. For this property to hold, it is necessary that the $V_{\ket{(k_1,...,k_N)}}^{\to F}$ have orthogonal output spaces.

Suppose that $x_1,...,x_N$ are substituted for inputs $\mathcal{A}_1,...,\mathcal{A}_N\in \CPM(\mathcal{H}_I,\mathcal{H}_O)$ such that $\mathcal{A}_1 = D(A_1),...,\mathcal{A}_N=D(A_N)$ for some $A_1,...,A_n\in \FHilb(\mathcal{H}_I,\mathcal{H}_O)$.
The corresponding controlled operations occuring at each time-step, which we call $\tilde{A}_1,...,\tilde{A}_N$, are defined as follows:
\begin{align*}
\tilde{A}_n := \sum_{(k_1,...,k_n)}A_{k_n}\otimes I_{\mathcal{H}_{\alpha_n}}\otimes \ket{(k_1,...,k_n)}\bra{(k_1,...,k_n)} \in \mathcal{L}(\mathcal{H}_I\otimes\mathcal{H}_{\alpha_n}\otimes \mathcal{H}_{C_n},\mathcal{H}_O\otimes\mathcal{H}_{\alpha_n}\otimes \mathcal{H}_{C_n})
\end{align*}

Then, the ciruict computes the completely positive map $\mathcal{A}_{out}$ by composing the $\tilde{V_n}$ and $\tilde{A_n}$ as follows:
\[
 \mathcal{A}_{\text{out}} \ := \ Tr_{\alpha_F} \circ D(\tilde{V}_{N+1}\circ \tilde{A}_{N}\circ\tilde{V}_{N}\circ ... \circ \tilde{V}_2\circ \tilde{A}_1 \circ \tilde{V}_1) \in \CPM(\mathcal{H}_P,\mathcal{H}_F)
\]

We obtain the output circuit for general $\mathcal{A}_n$ (i.e. not necessarily of the form $\mathcal{A}_n=D(A_n)$) by linearity.

To show that every \textbf{QC-QC} with memory can be implemented in the language, we begin by recursively defining the statement $\qcase_n$ with $n$ branching statements:
 \begin{align*}
 &\qcase_1\ P \ \{ 0 \to M\} :=\  \lambda t. \pair{P,M\,t}
 \\
 &\qcase_2 \ P \ \{ 0\to M_0 \ |\ 1\to M_1\}:=\ \qcase \ P \ \{ 0\to M_0 \ |\ 1\to M_1\}
 \\
& \qcase_n \ P\ \{0\to M_0 \ |\ ...\ |\ n-1 \to M_{n-1}\}
\\
& \quad := \mathtt{let}\ \pair{x,y}=P \ \mathtt{in}\\
 & \qquad\ \  \qcase \ x \ \{0\to \qcase_{\alift{\frac{n}{2}}}\ y \ \{0 \to M_0 \ |\ ...\ |\ \Big\lceil {\frac{n}{2}}\Big\rceil \to M_{\alift{\frac{n}{2}}}\}
 \\
 &\qquad \ \qquad \qquad |\ 1\to \qcase_{\lfloor \frac{n}{2}\rfloor}\ y\ \{0\to M_{\alift{\frac{n}{2}}+1}\ |\ ...\ |\ \Big\lfloor \frac{n}{2} \Big\rfloor \to M_n\}
 \\
 &\qquad \ \qquad \qquad \} &\text{for }n\geq 3.
\end{align*}
To simplify the notations, we will write $\qcase_n \ P \ \{k \to M_k\}$ for $\qcase_n \ P\ \{0\to M_0 \ |\ ...\ |\ n-1 \to M_{n-1}\}$.

In order to construct a program implementing the above circuit, we define the following isometries:
\begin{align*}
 W_{(k_1,...,k_{n-1})} &:=\sum_{k_n\notin(k_1,...,k_{n-1})} V_{\ket{(k_1,...,k_{n-1})}}^{\to k_{n}}\otimes \ket{k_{n}}
 & \text{for }1\leq n\leq N - 1\\
 W_{F} &:=\sum_{(k_1,...,k_{N})} V_{\ket{(k_1,...,k_{N})}}^{\to F}\otimes \bra{(k_1,...,k_{N})}
\end{align*}
There exist programs $\mathtt{W}_{(k_1,...,k_n)}$ ($1\leq n\leq N-1$) and $\mathtt{W}_F$ implementing the $W_{(k_1,...,k_n)}$ ($1\leq n\leq N-1$) and $W_F$. We recursively define the following programs:
\[
 \mathtt{P}_{(k_1,...,k_N)} := \lambda \pair{s_N,a_N,k_N}.\mathtt{let}\ t_N= x_{k_N}\ s_{N} \ \mathtt{in} \ \pair{t_N,a_N,k_N}
\]

\begin{align*}
  \mathtt{P}_{(k_1,...,k_n)} := \lambda \pair{s_n,a_n,k_n}.\qcase_{N-n+1} \ k_n\ \{k_n \to \ &
 \mathtt{let} \ t_n = x_{k_n}\,s_n \ \mathtt{in}
 \\
& \mathtt{let}\ \pair{s_{n+1},a_{n+1},k_{n+1}} = \mathtt{W}_{(k_{1},...,k_n)} \ \pair{t_n,a_n}\ \mathtt{in}
\\
&  \mathtt{P}_{(k_1,...,k_{n+1})} \ \pair{s_{n+1},a_{n+1},k_{n+1}}\ \}
\end{align*}
for $1\leq n\leq N-1$.
\[
\mathtt{P}_{()} := \lambda p.\mathtt{let}\ \pair{s_1,a_1,k_1} = \mathtt{W}_{()}\ p \ \mathtt{in} \ \mathtt{P}_{(k_1)} \ \pair{s_1,a_1,k_1}
\]

By construction, the following program implements the above circuit for a \textbf{QC-QC} with memory:
\begin{align*}
 \mathtt{P} :=\ &\lambda\pair{x_1,...,x_n}.\lambda p. \mathtt{let} \ \pair{t_N,a_N,k_N} = \mathtt{P}_{()} \ p \ \mathtt{in}\\
 &\mathtt{let}\ \pair{f,a_{F}}= \mathtt{W}_F\ \pair{t_N,a_N,k_N} \ \mathtt{in}
 \\
 &\mathtt{discard} \ a_{F}\ ; \ f
\end{align*}
where $\mathtt{discard}$ is a generalization of Example~\ref{ex:discard} to multiple qubits.

This proves the universality result:
\propqcqc*

\section{An Extension for Non-Linearity and Recursion}

\subsection{Full Definition of the Type System}
\label{app:extension_typing}

In this section, we give the full definition of the extended type system.

Let $\Lambda$ and $\Lambda'$ be a pair of typing contexts, which  are either both linear or both unrestricted.
If all the variables appearing in both $\Lambda$ and $\Lambda'$ are assigned the same type in $\Lambda$ and $\Lambda'$, then we can define their union, which we will write as $\Lambda\cup\Lambda'$.
If $\Lambda$ and $\Lambda'$ do not have any variables in common, we will write their union as $\Lambda\uplus\Lambda':=\Lambda,\Lambda'$.

The typing rules are detailed in Figure~\ref{fig:typing_gen}. If a typing judgment is valid, each variable of $\Gamma$ is guaranteed to be used at least once in the term, and each variable of $\Delta$ exactly once (counting uses in both branches of $\meas$ or $\qcase$ as the same). In particular, having unrestricted variables be used at least once prevents qubit deletion outside of a measurement (for instance $\lambda^* u. \ket{0}$).

\begin{figure}[!ht]
\hrulefill
\begin{center}
\[
\scalebox{0.9}{
 \begin{prooftree}
 \infer0[ax]{\cdot ; x:A \vdash x:A}
\end{prooftree}}
 \qquad
 \scalebox{0.9}{
 \begin{prooftree}
 \infer0[ax$^*$]{ u:A;\cdot \vdash u:A}
 \end{prooftree}}
 \qquad
 \scalebox{0.9}{
 \begin{prooftree}
  \infer0[unit]{\cdot;\cdot \vdash U:\q \lolli \q}
 \end{prooftree}}
 \qquad
 \scalebox{0.9}{
 \begin{prooftree}
  \hypo{k\in\{0,1\}}
  \infer1[qbit]{\cdot;\cdot \vdash \ket{k}:\q}
 \end{prooftree}}
 \]

 \[\scalebox{0.9}{
 \begin{prooftree}
  \hypo{\Gamma;(\Delta,x:A,y:B,\Delta') \vdash M:C}
  \infer1[X]{\Gamma;(\Delta,y:B,x:A,\Delta')\vdash M:C}
 \end{prooftree}}
 \qquad
 \scalebox{0.9}{
 \begin{prooftree}
  \hypo{(\Gamma,u:A,v:B,\Gamma');\Delta \vdash M:C}
  \infer1[X$^*$]{(\Gamma,v:B,u:A,\Gamma');\Delta\vdash M:C}
 \end{prooftree}}
\]

 \[
 \scalebox{0.9}{
\begin{prooftree}
 \hypo{\Gamma ;(\Delta,x:A)\vdash M:B}
 \infer1[$\lolli$I]{\Gamma;\Delta\vdash \lambda x. M:A\lolli B}
 \end{prooftree}}
 \quad
 \scalebox{0.9}{
 \begin{prooftree}
  \hypo{\Gamma;\Delta\vdash M:A\lolli B}
  \hypo{\Gamma';\Delta'\vdash N:A}
  \infer2[$\lolli$E]{\Gamma\cup \Gamma' ; \Delta\uplus\Delta'\vdash M N:B}
 \end{prooftree}}
 \quad
 \scalebox{0.9}{
 \begin{prooftree}
  \hypo{\Gamma;\Delta \vdash M:A}
  \hypo{\Gamma';\Delta'\vdash N:B}
  \infer2[$\otimes$I]{\Gamma\cup\Gamma';\Delta\uplus\Delta' \vdash \pair{M , N}:A\otimes B}
 \end{prooftree}}
 \]

 \[
\scalebox{0.9}{
  \begin{prooftree}
\hypo{\Gamma;\Delta \vdash M:A\otimes B}
\hypo{\Gamma';(\Delta',x:A,y:B)\vdash N:C}
\infer2[$\otimes$E]{\Gamma\cup\Gamma'; \Delta\uplus\Delta' \vdash \mathtt{let}\ \pair{x,y}=M\ \mathtt{in}\ N:C}
\end{prooftree}}
\quad
\scalebox{0.9}{
 \begin{prooftree}
  \infer0[$\one$I]{\cdot;\cdot \vdash \unit:\one}
 \end{prooftree}}
 \quad
 \scalebox{0.9}{
 \begin{prooftree}
  \hypo{\Gamma;\Delta\vdash M: \one}
  \hypo{\Gamma';\Delta' \vdash N: A}
  \infer2[seq]{\Gamma\cup\Gamma';\Delta\uplus\Delta'\vdash M;N:A}
 \end{prooftree}}
\]

\[
\scalebox{0.9}{
 \begin{prooftree}
 \hypo{(\Gamma,u:A) ;\Delta\vdash M:B}
 \infer1[$\Rightarrow$I]{\Gamma;\Delta\vdash \lambda^* u. M:A\Rightarrow B}
 \end{prooftree}}
 \qquad
 \scalebox{0.9}{
 \begin{prooftree}
  \hypo{\Gamma;\Delta\vdash M:A\Rightarrow B}
  \hypo{\Gamma';\cdot \vdash N:A}
  \infer2[$\Rightarrow$E]{\Gamma\cup\Gamma' ; \Delta\vdash M^*N:B}
 \end{prooftree}}
\]

\[
\scalebox{0.9}{
 \begin{prooftree}
 \hypo{\Gamma;\Delta \vdash P:\q}
  \hypo{\Gamma';\Delta' \vdash M:A}
  \hypo{\Gamma';\Delta' \vdash N: A}
  \infer3[meas]{\Gamma\cup\Gamma';\Delta\uplus\Delta'\vdash \meas\ P\ \{0\rightarrow M\ |\ 1\rightarrow N\}: A}
 \end{prooftree}}
 \]

 \[
 \scalebox{0.9}{
 \begin{prooftree}
  \hypo{\Gamma;\Delta \vdash P:\q}
  \hypo{\Gamma';\Delta' \vdash M:A \lolli \q^n}
  \hypo{\Gamma';\Delta' \vdash N: A\lolli \q^n}
  \hypo{M,N\in\mathscr{C}^*}
  \infer4[qcase]{\Gamma\cup \Gamma';\Delta\uplus\Delta'\vdash \qcase\ P\ \{0\rightarrow M\ |\ 1\rightarrow N\}: A\lolli \q^{n+1}}
 \end{prooftree}}
\]

\[
\scalebox{0.9}{
 \begin{prooftree}
  \hypo{(\Gamma,f:A\lolli B); x:A \vdash M: B}
  \infer1[rec]{\Gamma;\cdot\vdash \letrec\ f \ x = M : A\lolli B}
 \end{prooftree}}
\qquad
\scalebox{0.9}{
 \begin{prooftree}
  \hypo{(\Gamma,f:A\Rightarrow B,u:A);\cdot \vdash M: B}
  \infer1[rec$^*$]{\Gamma;\cdot\vdash \letrec^*\ f \ u = M : A\Rightarrow B}
 \end{prooftree}}
\]

\end{center}
\hrulefill
\caption{Typing rules}
\label{fig:typing_gen}
\end{figure}

\subsection{Full Definition of the Operational Semantics}
\label{app:extension_op}

In this section, we detail the definition of the operational semantics of the language's extension.

\subsubsection{A Language for Program Executions}

The set of (extended) \emph{execution terms}, which we will write as $\mathscr{E}^*$, is defined in Figure~\ref{fig:syntax_execution_terms_gen}. Compared to the set of extended terms $\mathscr{T}^*$, the measurement is assigned the extended reference $d^*\in\device$, which is of the form $d^*= d \cdot n_1\cdot ...\cdot n_k$ where $d\in\device$ is a device reference.
\begin{figure}[!ht]
\hrulefill
\begin{align*}
 \begin{array}{lrl}
 \
 \text{(Extended execution\ terms)}&\ \ \mathscr{E}^* \ni M,N,P \  ::=  &  x \ |\ M N \ |\ \lambda x.M \ |\ u\ |\ M^*N \ | \ \lambda^* u. M \\
  && \!\!\!\! |\ \ \pair{M,N} \ |\ \mathtt{let}\ \pair{x,y} = M \  \mathtt{in}\ N \ |\ \unit\ |\ M;N \\
  && \!\!\!\! |\ \ U\ | \ \ket{0}\ |\ \ket{1}
  \\
  && \!\!\!\! |\ \ \meas \ d^* \triangleright P\ \{0\rightarrow M\ |\ 1\rightarrow N\} \\
  && \!\!\!\! |\ \ \qcase \ P\ \{0\rightarrow M\ |\ 1\rightarrow N\} \\
 && \!\!\!\! |\ \ \letrec \ f\ x=M \ |\ \letrec^* \ f\ u=M
 \end{array}
\end{align*}
\hrulefill
\caption{Extended syntax of execution terms}
\label{fig:syntax_execution_terms_gen}
\end{figure}

Similarly to the linear fragment, we define the following typing rules on terms of $\mathscr{E}^*$:
\[
\scalebox{0.9}{
 \begin{prooftree}
 \hypo{\Gamma;\Delta \vdash P:\q}
  \hypo{\Gamma';\Delta' \vdash M:A}
  \hypo{\Gamma';\Delta' \vdash N: A}
  \infer3[meas]{\Gamma\cup\Gamma';\Delta\uplus\Delta'\vdash \meas\ d^* \triangleright P\ \{0\rightarrow M\ |\ 1\rightarrow N\}: A}
 \end{prooftree}}
 \]

 \[
 \scalebox{0.9}{
 \begin{prooftree}
  \hypo{\Gamma;\Delta \vdash P:\q}
  \hypo{\Gamma';\Delta' \vdash M:A \lolli \q^n}
  \hypo{\Gamma';\Delta' \vdash N: A\lolli \q^n}
  \infer3[qcase]{\Gamma\cup \Gamma';\Delta\uplus\Delta'\vdash \qcase\ P\ \{0\rightarrow M\ |\ 1\rightarrow N\}: A\lolli \q^{n+1}}
 \end{prooftree}}
\]
and the rules for all other primitives are the same as in $\mathscr{T}^*$.

In turn, the set of (extended) \emph{configurations}, written as $\vec{\mathscr{E}}^*$ is defined by the following grammar:
\[
 \begin{array}{lrl}
  (\text{Extended configurations}) &\quad \vec{\mathscr{E}}^* \ni \vec{M},\vec{N} \ ::= & \vec{0}\ |\ M^{\sigma}\ |\ \vec{M}+\vec{N} \ |\ \alpha \cdot \vec{M}
 \end{array}
\]
where $M\in\mathscr{E}^*$, $\alpha \in\mathbb{C}$ and $\sigma \in \mathcal{M}^*$.

\subsubsection{Reduction Rules}

To define the reduction of a term $M\in\mathscr{T}^*$, we assign to each measurement in $M$ a unique device reference $d\in\device\subseteq \device^*$, resulting in a term $M'\in \mathscr{E}^*$.
It is important to choose $d$ in $\device$ rather than $\device^*\setminus \device$ to avoid conflicting device references.

\emph{Contexts} are defined by the following grammar:
\begin{align*}
 E\ ::=\ & [] \ | \ EM \ |\ ME\ |\ E^*M\ |\ M^*E\ |\ \pair{E,M}\ |\ \pair{M,E}\ |\ \mathtt{let}\ \pair{x,y}=E\ \mathtt{in}\ M\ |\ E;M \\
 &\!\!\!\! |\ \ \meas \ d\triangleright E\ \{0\to M\ |\ 1\to N\} \ |\ \qcase \ E\ \{0\to  M\ |\ 1\to N\}
\end{align*}
where $M,N\in\mathscr{E}^*$.
Therefore we can use the following notation:
\[
 E\left[\sum_{i=1}^n \alpha_i \cdot M_i^{\sigma_i}\right] \ :=\ \sum_{i=1}^n \alpha_i \cdot E[M_i]^{\sigma_i}
\]

Finally, the set of extended values $\mathscr{V}^*\subseteq \mathscr{E}^*$ is defined by the following grammar:
 \[
 \begin{array}{lll}
  \text{(Extended values)} \quad \mathscr{V}^* \ni V,W \ ::= \!\!\!\!& \ \ \
  \lambda x. M \ |\ \lambda^* u. M\ |\ \unit  \ |\ \pair{V,W} \ |\ U \ |\ \ket{0}\ |\ \ket{1}
 \\ & |\ \ \letrec \ f\ x = M \ |\ \letrec^* \ f\ u=M
 \end{array}
\]

The reduction relation ${\twoheadrightarrow}\subseteq \vec{\mathscr{E}}^*\times\vec{\mathscr{E}}^*$ is defined in Figure~\ref{fig:rel_gen}.

\begin{figure}[!ht]
\hrulefill
\begin{center}

\[
\scalebox{0.9}{
 \begin{prooftree}
  \hypo{N\twoheadrightarrow \vec{N}}
  \infer1{MN\twoheadrightarrow M\vec{N}}
 \end{prooftree}}
\qquad
\scalebox{0.9}{
 \begin{prooftree}
  \hypo{M\twoheadrightarrow \vec{M}}
  \infer1{MV\twoheadrightarrow \vec{M}V}
 \end{prooftree}}
\qquad
\scalebox{0.9}{
 \begin{prooftree}
  \hypo{\phantom{\vec{M}}}
  \infer1{(\lambda x. M)V\twoheadrightarrow M\{V/x\}}
 \end{prooftree}}
\qquad
\scalebox{0.9}{
 \begin{prooftree}
  \hypo{N\twoheadrightarrow \vec{N}}
  \infer1{M^*N\twoheadrightarrow M^*\vec{N}}
 \end{prooftree}}
\]
\[
\scalebox{0.9}{
 \begin{prooftree}
  \hypo{M\twoheadrightarrow \vec{M}}
  \infer1{M^*V\twoheadrightarrow \vec{M}^*V}
 \end{prooftree}}
\qquad
\scalebox{0.9}{
 \begin{prooftree}
  \hypo{\phantom{\vec{M}}}
  \infer1{(\lambda^* u. M)^*V\twoheadrightarrow M\db{V/u}}
 \end{prooftree}}
\qquad
\scalebox{0.9}{
 \begin{prooftree}
  \hypo{N\twoheadrightarrow \vec{N}}
  \infer1{\pair{M,N}\twoheadrightarrow \pair{M,\vec{N}}}
 \end{prooftree}}
\qquad
\scalebox{0.9}{
 \begin{prooftree}
  \hypo{M\twoheadrightarrow \vec{M}}
  \infer1{\pair{M,V}\twoheadrightarrow \pair{\vec{M},V}}
 \end{prooftree}}
\]
\[
\scalebox{0.9}{
 \begin{prooftree}
  \hypo{M\twoheadrightarrow \vec{M}}
  \infer1{\mathtt{let} \ \pair{x,y} = M \ \mathtt{in}\ N \twoheadrightarrow \mathtt{let} \ \pair{x,y}=\vec{M}\ \mathtt{in}\ N}
 \end{prooftree}}
\qquad
\scalebox{0.9}{
 \begin{prooftree}
   \hypo{\phantom{\vec{M}}}
  \infer1{\mathtt{let} \ \pair{x,y} = \pair{V,W} \ \mathtt{in}\ N \twoheadrightarrow N\{V/x,W/y\}}
 \end{prooftree}}
\]
\[
\scalebox{0.9}{
 \begin{prooftree}
  \hypo{M\twoheadrightarrow \vec{M}}
  \infer1{M;N \twoheadrightarrow \vec{M};N}
 \end{prooftree}}
\qquad
\scalebox{0.9}{
\begin{prooftree}
 \hypo{\phantom{\vec{M}}}
  \infer1{\unit;N\twoheadrightarrow N}
 \end{prooftree}}
\qquad
\scalebox{0.9}{
 \begin{prooftree}
 \hypo{k\in\{0,1\}}
 \hypo{U = {\begin{bmatrix}
                  u_{00} & u_{01} \\
                  u_{10} & u_{11}
                 \end{bmatrix}}
}
  \infer2{U\ket{k}\twoheadrightarrow u_{0k} \cdot\ket{0} + u_{1k} \cdot\ket{1}}
 \end{prooftree}}
\]
\[
\scalebox{0.9}{
 \begin{prooftree}
  \hypo{P\twoheadrightarrow \vec{P}}
  \infer1{\meas\ d^*\triangleright P\ \{0\rightarrow M_0\ |\ 1\rightarrow M_1\}\twoheadrightarrow \meas\ d^*\triangleright\vec{P}\  \{0\rightarrow M_0\ |\ 1\rightarrow M_1\}}
 \end{prooftree}}
\]
\\
\[
\scalebox{0.9}{
 \begin{prooftree}
  \hypo{k'\in\{0,1\}}
  \infer1[$(m_k) \quad (\scriptsize k\in\{0,1\})$]{\meas\ d^*\triangleright \ket{k'}\ \{0\rightarrow M_0\ |\ 1\rightarrow M_1\}\twoheadrightarrow \delta_{k,k'}\cdot M_{k}^{[d^* \mapsto k]}}
 \end{prooftree}}
\]
\[
\scalebox{0.9}{
 \begin{prooftree}
  \hypo{P\twoheadrightarrow \vec{P}}
  \infer1{\qcase\ P\ \{0\rightarrow M_0\ |\ 1\rightarrow M_1\}\twoheadrightarrow \qcase\ \vec{P}\ \{0\rightarrow M_0\ |\ 1\rightarrow M_1\}}
 \end{prooftree}}
\]
\\
\[
\scalebox{0.9}{
\begin{prooftree}
\hypo{k\in\{0,1\}}
\hypo{t \text{ is fresh}}
 \infer2{\qcase\ \ket{k}\ \{0\rightarrow M_0 \ |\  1\rightarrow M_1\} \twoheadrightarrow \lambda t. \pair{\ket{k},M_k\, t}}
\end{prooftree}}
\]
\[
\scalebox{0.9}{
\begin{prooftree}
 \hypo{\phantom{\vec{M}}}
 \infer1{(\letrec\ f\ x=M )V \twoheadrightarrow M\{V/x\}\db{\letrec \ f\ x=M/f}}
\end{prooftree}}
\]
\[
\scalebox{0.9}{
\begin{prooftree}
 \hypo{\phantom{\vec{M}}}
 \infer1{(\letrec^*\ f\ u=M)^*V \twoheadrightarrow M\db{V/u}\db{\letrec^* \ f\ u=M/f}}
\end{prooftree}}
\]
\[
\scalebox{0.9}{
\begin{prooftree}
\hypo{M \twoheadrightarrow \vec{M}}
\hypo{\sigma \sqcup \Mem(\vec{M})\ ||\ \Mem(\vec{N})}
\hypo{M^\sigma \notin \Supp(\vec{N})}
 \infer3{\alpha \cdot M^\sigma + \vec{N}\twoheadrightarrow \alpha \cdot \vec{M}^{\{\sigma\}}+ \vec{N}}
\end{prooftree}}
\]

\end{center}
\hrulefill
\caption{Reduction relation $\twoheadrightarrow$}
\label{fig:rel_gen}
\end{figure}

\end{appendices}

\end{document}